%% file: main.tex
\documentclass[format=acmsmall, screen=true]{acmart} %% For TOCL
\usepackage{pstring}
\usepackage{pgfcore}
\usepackage[utf8]{inputenc}
\usepackage{latexsym}
\usepackage{graphicx}
\pstrSetArrowColor{black}
\usepackage{tikz}
\usetikzlibrary{arrows,shapes}
\usepackage{xspace}
\usepackage[all]{xy}
\usepackage{manfnt}

\newif\ifdraft\draftfalse
\definecolor{LimeGreen}{rgb}{0.2, 0.8, 0.2}
\definecolor{Vert}{rgb}{0.235, 0.80, 0.294}
\definecolor{Crimson}{rgb}{0.86, 0.08, 0.24}
\definecolor{Orange}{rgb}{0.98, 0.6, 0.01}
\definecolor{Blue}{rgb}{0, 0.5, 1}

\input{macros.tex}

\begin{document}

\theoremstyle{acmdefinition}
\newtheorem{remark}[theorem]{Remark}
\newtheorem{property}[theorem]{Property}

\title{Collapsible Pushdown Parity Games}

\author{Christopher~H. Broadbent}
\affiliation{%
  \institution{Department of Computer Science, University of Oxford}
  \city{Oxford}
  \country{UK}
}
\email{chbroadbent@gmail.com}

\author{Arnaud Carayol}
\affiliation{%
  \institution{CNRS, LIGM (Université Paris Est \& CNRS)}
  \streetaddress{5 boulevard Descartes — Champs sur Marne}
  \city{Marne-la-Vallée Cedex 2}
  \postcode{77454}
  \country{France}
}
\email{Arnaud.Carayol@univ-mlv.fr}

\author{Matthew Hague}
\affiliation{%
  \institution{Royal Holloway, University of London}
  \city{London}
  \country{UK}
}
\email{Matthew.Hague@rhul.ac.uk}

\author{Andrzej~S. Murawski}
\affiliation{%
  \institution{Department of Computer Science, University of Oxford}
  \city{Oxford}
  \country{UK}
}
\email{Andrzej.Murawski@cs.ox.ac.uk}

\author{C.-H. Luke Ong}
\affiliation{%
  \institution{Department of Computer Science, University of Oxford}
  \city{Oxford}
  \country{UK}
}
\email{Luke.Ong@cs.ox.ac.uk}

\author{Olivier Serre}
\orcid{0000-0001-5936-240X}
\affiliation{%
  \institution{Université de Paris, IRIF, CNRS}
  \streetaddress{Bâtiment Sophie Germain, Case courrier 7014, 
8 Place Aurélie Nemours}
  \city{Paris Cedex 13}
  \postcode{75205}
  \country{France}}
\email{Olivier.Serre@cnrs.fr}

\begin{abstract}
This paper studies a large class of two-player perfect-information turn-based parity games on infinite graphs, namely those generated by collapsible pushdown automata. The main motivation for studying these games comes from the connections from collapsible pushdown automata and higher-order recursion schemes, both models being equi-expressive for generating infinite trees. Our main result is to establish the decidability of such games and to provide an effective representation of the winning region as well as of a winning strategy. Thus, the results obtained here provide all necessary tools for an in-depth study of logical properties of trees generated by collapsible pushdown automata/recursion schemes.
\end{abstract}

\acmSubmissionID{}

%
% The code below should be generated by the tool at
% http://dl.acm.org/ccs.cfm
% Please copy and paste the code instead of the example below. 
%



%
% End generated code
%


\maketitle

% The default list of authors is too long for headers.
\renewcommand{\shortauthors}{Broadbent et al.}

%\newpage\tableofcontents\newpage

\input{Introduction.tex}

\input{Preliminaries.tex}

\input{Results.tex}
\input{RankAware.tex}

\input{Outmost.tex}
\input{ReducingOrder.tex}
\input{Summary.tex}

\input{Consequences.tex}

\bibliographystyle{ACM-Reference-Format}
\bibliography{abbrevs,main}

\end{document}

%% file: macros.tex
\ifdraft
\newcommand\todo[1]{{\footnotesize \color{OliveGreen}
[\textbf{To do:} #1]}}
\newcommand\os[1]{\marginpar[{\color{cyan}\small\dbend}]{
\color{cyan}\small\dbend}{\footnotesize \color{cyan}
[#1 - \textbf{Olivier}]}}
\newcommand\mh[1]{\marginpar[{\color{orange}\small\dbend}]{
\color{orange}\small\dbend}{\footnotesize \color{orange} 
[#1 - \textbf{Matt}]}} 
\newcommand\am[1]{\marginpar[{\color{orange}\small\dbend}]{
\color{orange}\small\dbend}{\footnotesize \color{red} 
[#1 - \textbf{Andrzej}]}} 
\newcommand\cb[1]{\marginpar[{\color{orange}\small\dbend}]{
\color{orange}\small\dbend}{\footnotesize \color{OliveGreen} 
[#1 - \textbf{Chriss}]}}

\else
\newcommand\todo[1]{}
\newcommand\os[1]{}
\newcommand\cb[1]{}
\newcommand\am[1]{}
\newcommand\mh[1]{}
\newcommand\mhnopar[1]{}
\newcommand\review[1]{}

\fi

\newcommand   \ie{\hbox{\sl i.e.\ }}
\newcommand   \eg{\hbox{\sl e.g.\ }}

\newcommand   \resp{\hbox{\sl resp.\ }}

\newcommand\textbfit[1]{{\bf\em #1}}
\newcommand{\concept}[1]{\textbfit{\boldmath #1}}
\newcommand{\defin}[1]{\concept{#1}}

\newcommand\makeset[1]{\{\,#1\,\}}
\newcommand\anglebra[1]{\langle\, #1 \,\rangle}
\renewcommand\anglebra[1]{(#1)}
\newcommand\mor{\longrightarrow}

\newcommand\domain[1]{{\sf dom}(#1)}

\newcommand\Op[2]{\mathit{Op}_{#1}^{#2}}
\newcommand\topone[1]{{\mathit top}_1\,{#1}}
\newcommand\pushone[1]{{\mathit push}_1^{#1}}

\newcommand\collapse{{\mathit collapse}}

\newcommand\pushlk[2]{{\mathit push}_{1}^{#2,#1} }
\newcommand\popn[1]{{\mathit pop}_{#1}}
\newcommand\pushn[1]{{\mathit push}_{#1}}
\newcommand\topn[1]{{\mathit top}_{#1}}
\newcommand{\toprew}[1]{\mathit rew_1^{#1}}

\newcommand\mksk[1]{\hbox{\tt{[}} #1 \hbox{\tt{]}}}
\newcommand\lsk{\hbox{\tt{[}}}
\newcommand\rsk{\hbox{\tt{]}}}
\newcommand\order[1]{{\mathit ord}(#1)}

\newcommand\id{\mathit{id}}

\newcommand\arrlab[1]{\xrightarrow{#1}}

\newcommand\stack{s}
\newcommand\varstack{t}

\renewcommand{\epsilon}{\varepsilon}
\renewcommand{\phi}{\varphi}

\newcommand{\Eloise}{\'Elo\"ise\xspace}
\newcommand{\Abelard}{Abelard\xspace}

\renewcommand{\phi}{\varphi}
\newcommand{\ttrue}{t\! t}
\newcommand{\ffalse}{f\!\! f}

\renewcommand{\ttrue}{q_{\mathfrak{t}}}
\renewcommand{\ffalse}{q_{\mathfrak{f}}}

\newcommand{\VE}{V_\Ei}
\newcommand{\VA}{V_\Ai}
\newcommand{\Ei}{\mathrm{E}}
\newcommand{\Ai}{\mathrm{A}}
\newcommand{\fstrat}{\widetilde{\phi}}
\newcommand{\strat}{\phi}
\newcommand{\arena}{\mathcal{G}}
\newcommand{\pgraph}{G}
\newcommand{\pggraph}{\mathcal{G}}
\newcommand{\game}{\mathbb{G}}
\newcommand{\pgame}{\mathbb{G}}
\newcommand{\fgraph}{\widetilde{G}}
\newcommand{\fggraph}{\widetilde{\mathcal{G}}}
\newcommand{\fgame}{\widetilde{\mathbb{G}}}
\newcommand{\fplay}{{\widetilde{\lambda}}}

\newcommand{\play}{\lambda}
\newcommand{\WC}{\Omega}

\newcommand{\col}{\rho}
\newcommand{\colors}{C}
\newcommand{\exptime}{\textsc{ExpTime}}
\newcommand{\pprocess}{\mathcal{P}}

\newcommand{\stacks}{St}
\newcommand{\sh}{sh}
\newcommand {\Stepsg}[1]{\ensuremath{\mathit{Steps}_{#1}}}

\newcommand{\pcol}[1]{mcol^{#1}}
\newcommand{\vect}[1]{\overrightarrow{#1}}
\newcommand {\Rounds}[1]{\ensuremath{\mathit{Rounds}_{#1}}}

\newcommand{\ini}{\textrm{in}}
\newcommand{\eat}[1]{}

\renewcommand{\pprocess}{\mathcal{A}}
\newcommand{\maxcolor}{d}
\newcommand{\LinkRk}{LinkRk}

\newcommand{\sdp}{abstract pushdown automaton\xspace}
\newcommand{\sdps}{abstract pushdown automata\xspace}

\newcommand{\sda}{abstract pushdown alphabet\xspace}
\newcommand{\sd}{abstract pushdown content\xspace}
\newcommand{\sdg}{abstract pushdown graph\xspace}
\newcommand{\sdgg}{abstract pushdown arena\xspace}
\newcommand{\sdpgame}{abstract pushdown parity game\xspace}
\newcommand{\sdgame}{abstract pushdown game\xspace}

\newcommand{\SDPGAMES}{Abstract Pushdown Parity Games\xspace}

\newcommand\mklksk[2]{\hbox{\tt{[}} #2 \nd(#1){\hbox{\tt{]}}}}

\newcommand{\transgraph}[1]{\mathrm{Graph}(#1)}

\newcommand{\flaten}[1]{\widetilde{#1}}

\newcommand{\target}[1]{\mathit{target}(#1)}
\newcommand{\targetf}{\mathit{target}}
\newcommand{\stratCPDA}[1]{\varphi_{#1}}

\newcommand{\rk}{{\rm rk}}
\newcommand{\indef}{\circlearrowleft}% to be updated when seeing
\newcommand{\nalr}{\dagger}% not a link rank
\newcommand{\mcolsofar}{\theta}% minimal color seen so far
\newcommand{\pprocessrk}{\mathcal{A}_{\rk}}
\newcommand{\colrk}{\col_{\rk}}
\newcommand{\Ark}{\pprocessrk}

\newcommand{\Gammark}{\Gamma_{\rk}}
\newcommand{\Qrk}{Q_{\rk}}
\newcommand{\QrkE}{Q_{\rk,\Ei}}
\newcommand{\QrkA}{Q_{\rk,\Ai}}
\newcommand{\Deltark}{\Delta_{\rk}}
\newcommand{\pgraphrk}{{G}_{\rk}}
\newcommand{\pggraphrk}{\mathcal{G}_{\rk}}
\newcommand{\pgamerk}{\pgame_{\rk}}
\renewcommand{\path}{\pplay}
\newcommand{\ind}[1]{ind(#1)}
\newcommand{\pstratrk}{\strat_{\rk}}
\newcommand{\pplayrk}{\play_{\rk}}
\newcommand{\Srk}{\mathcal{S}_{\rk}}

\newcommand{\Vrk}{V_{\rk}}
\newcommand{\Brk}{\mathcal{B}_{\rk}}
\newcommand{\lf}{{\rm lf}}
\newcommand{\Blf}{\mathcal{B}_{\lf}}
\newcommand{\pprocesslf}{\mathcal{A}_{\lf}}

\newcommand{\Gammalf}{\Gamma_{\lf}}
\newcommand{\Qlf}{Q_{\lf}}
\newcommand{\Deltalf}{\Delta_{\lf}}
\newcommand{\pgraphlf}{{G}_{\lf}}
\newcommand{\pggraphlf}{\mathcal{G}_{\lf}}
\newcommand{\pgamelf}{\pgame_{\lf}}
\newcommand{\pstratlf}{\strat_{\lf}}
\newcommand{\pplaylf}{\play_{\lf}}
\newcommand{\Slf}{\mathcal{S}_{\lf}}

\newcommand{\abs}{{\rm abs}}
\newcommand{\pggraphabs}{\mathcal{G}_{\abs}}
\newcommand{\pgameabs}{\mathbb{G}_{\abs}}

\newcommand{\saut}{\textcolor{white}{.\smallskip}}

\newcommand{\Qf}{\widetilde{Q}}
\newcommand{\Gammaf}{\widetilde{\Gamma}}
\newcommand{\fpprocess}{\widetilde{\pprocess}}
\newcommand{\Deltaf}{\widetilde{\Delta}}
\newcommand{\qinif}{\widetilde{q_0}}
\newcommand{\Sf}{\widetilde{\mathcal{S}}}

\renewcommand{\S}{\mathcal{S}}
\newcommand{\Last}{Last}
\newcommand{\Extremal}{Ext}
\newcommand{\Memory}{Mem}

\newcommand{\expn}[1]{{\rm Exp}_{#1}}

\newcommand{\cpp}{C\nolinebreak\hspace{-.05em}\raisebox{.4ex}{\tiny\textbf{ +}}\nolinebreak\hspace{-.10em}\raisebox{.4ex}{\tiny\textbf{ +}}}

%% file: Introduction.tex
\section{Introduction}

This paper studies a large class of two-player perfect-information turn-based parity games on infinite graphs, namely those generated by collapsible pushdown automata (CPDA). 

\subsection*{Parity Games on Infinite Graphs}

A two-player perfect-information turn-based parity game on a graph (or simply a parity game) is played by two players, \Eloise and \Abelard, who move a pebble along edges of a graph whose vertices have been partitioned between the two players and coloured by a function assigning to every vertex a colour chosen in a finite subset of $\mathbb{N}$. The player owning the current vertex, chooses where to move the pebble next and so on forever. Hence, a play is an infinite path in the graph, and the winner is determined thanks to the colouring function by declaring \Eloise to win if and only if the smallest colour appearing infinitely often is even. 

Parity games have been widely studied since the 80s because of their close links to important problems arising from logic. A fundamental result of Rabin is that $\omega$-regular tree languages, equivalently tree languages definable in monadic second-order (MSO) logic, form a Boolean algebra~\cite{Rabin69}. The difficult part of the proof is complementation, and since the publication of this result in 1969, it has been a challenging problem to simplify it. A much simpler one was obtained by Gurevich and Harrington in \cite{GurevichH82} making use of parity games for checking membership of a tree in the language accepted by an automaton: \Eloise builds a run on the input tree while \Abelard  tries to exhibit a rejecting branch in the run. The proof of Gurevich and Harrington was followed by many others trying to simplify the original proof of Rabin, and beyond this historical result, the tight connection between automata and games is one of the main tools in the areas of automata theory and logic (see \eg \cite{Thomas97,Wilke2001,Walukiewicz04}). 

The above-mentioned result of Rabin is equivalent to the fact that, given a formula from MSO logic, one can decide whether it holds in the complete infinite binary tree. Whether this result can be extended to more and more complex classes of trees is an active line of research since then. While decidability of MSO logic on the complete binary tree is equivalent to deciding whether \Eloise has a winning strategy in a parity game played on a \emph{finite} graph, extensions to more complex trees require one to consider games played on infinite graphs (and the more general the trees, the more general the graphs to be considered). 

Since the late 1990s, another important motivation for considering games played on infinite graphs emerged because of their connections with program verification. Here, there is a trade-off between richness of the graph describing the program to verify and decidability of the logic used to express the property to check. Regarding logic, most of the logics considered in program verification are captured by the $\mu$-calculus (an extension of modal logic with fixpoint operators) and therefore the model-checking problem is reduced again to solving a parity game played on a graph that is a synchronised product between the graph describing the system to verify and a finite graph describing the dynamic of the formula. Hence, the quest here is to look for graphs that model programs using natural features in programming languages (\eg recursion, higher-order arguments, rich data domains, etc.) and whose associated parity games remain decidable.

Both objectives —~extending Rabin's result to richer trees and verifying programs with natural features in programming languages~— games played on graphs generated by pushdown automata and their extensions, in particular collapsible pushdown automata, have proven to be fruitful. In a nutshell, collapsible pushdown automata extend usual pushdown automata by replacing the (order-$1$) stack by an order-$n$ stack that is defined as a stack whose elements are order-$(n-1)$ stacks and whose base symbols are equipped with links pointing deeper in the stack and that can later be used to collapse the stack.

\subsection*{Main Results}

Collapsible pushdown automata are equi-expressive with higher-order recursion schemes —~these are essentially finite typed deterministic term rewriting systems that generate an infinite tree when one applies the rewriting rules \emph{ad infinitum}~— for generating trees~\cite{HMOS08,HMOS17}, this class of trees subsumes all known classes of trees with decidable MSO theories. Regarding programs, collapsible pushdown automata permit to capture higher-order procedure calls —~a central feature in modern day programming and supported by many languages such as  \cpp, Haskell, OCaML, Javascript, Python, or Scala. 

Hence, considering parity games played on transition graphs of (collapsible) pushdown automata is a central problem for both extending Rabin's seminal result and verifying real-life programs. The study of such games raises three questions of increasing difficulty.
\begin{enumerate}
	\item Decide, for a given initial position, whether \Eloise has a winning strategy, \ie whether she has a way to play that guarantees she wins regardless of the choices of \Abelard. In the context of program verification, the counterpart of this question is the (local) model-checking problem.
	\item Finitely describe \Eloise's winning region, \ie the set of all positions from which she has a winning strategy. While in the setting of games on finite graphs this is equivalent to the previous question, when considering an infinite graph it is unclear whether a finite presentation of the winning region exists and, when it does, specific tools must be used to describe such an object. In the context of program verification, the counterpart of this question is the global model-checking problem.
	\item Finitely describe, for a given initial position, a winning strategy for \Eloise. Note that a classical result (positional determinacy~\cite{EJ91}) on parity games states that winning strategies can always be chosen to be positional, \ie to depend only on the current vertex; however, when describing a winning strategy in a game played on an infinite graph, the purpose is to find a suitable machine model of implementing a winning strategy rather  than focusing on capturing a special (simple) form of winning strategies. In the context of program verification, the counterpart of this question is the synthesis problem.
\end{enumerate}

In this paper we positively answer those questions. More specifically, our main Theorem implies the following.
\begin{enumerate}
	\item One can decide, for a given initial position, whether \Eloise has a winning strategy and this is an $n$-\exptime-complete problem, where $n$ is the order of the underlying collapsible pushdown automaton.
	\item We introduce a model of finite-state automata defining regular sets of configurations of collapsible pushdown automata and prove that the winning region is always such an (effective) regular set.
	\item We introduce a model of collapsible pushdown automata  tailored to describing strategies and prove that, for any game, we can compute a winning strategy described by such a machine. 
\end{enumerate}

Note that the above-mentioned results were presented by the authors in a series of papers in the LiCS conference~\cite{HMOS08,BCOS10,CS12} and that the current paper gives 
a unifying and complete presentation of their proofs.

\subsection*{Related Work}
We briefly review the known results on collapsible pushdown parity games (and subclasses). See Table~\ref{fig:CPDAgames} for a summary.

The first paper explicitly considering pushdown games (\ie order-$1$ CPDA games) is \cite{Walukiewicz96,Walukiewicz01}: an optimal algorithm for deciding the winner is given (\exptime-complete) as well as a construction of a strategy realised by a synchronised pushdown automaton. However, decidability can be derived from the MSO decidability of pushdown graphs \cite{MullerS85} in combination with the existence of positional winning strategies in parity games on infinite graphs \cite{EJ91}: indeed one can write an MSO formula stating the existence of a positional winning strategy for \Eloise (see \emph{e.g.} \cite{CachatPHD} for such a formula). A construction similar to the one in \cite{Walukiewicz96,Walukiewicz01} was given by Serre in his Ph.~D.~\cite{SerrePHD}, and we partly build upon it in the present paper. Another approach, using two-way alternating parity tree automata, was developed by Vardi in~\cite{Vardi98}.  The winning region was characterised in \cite{Serre03,Cachat02} and later in \cite{Hag08,HagueO11} using saturation techniques.

Cachat first considered parity games played on transition graphs of higher-order pushdown automata (HOPDA, a strict subclass of collapsible pushdown automata) in \cite{Cachat03} providing an optimal algorithm for deciding the winner ($n$-\exptime-complete, where $n$ is the order). As for pushdown games, decidability can be derived from the MSO decidability of higher-order pushdown graphs \cite{Caucal02} in combination with the existence of positional winning strategies in parity games on infinite graphs \cite{EJ91}. An alternative simpler proof was given in \cite{CHMOS08} that permits moreover to characterise the winning region and to construct a synchronised order-$n$ higher-order pushdown automaton realising a winning strategy. Also see \cite{CS08} for an approach extending the techniques of \cite{Vardi98} to higher-order, and \cite{BM04,HagueO08} for saturation techniques (for the reachability winning condition only).

Order-$2$ collapsible pushdown parity games were considered in \cite{KNUW05} (under the name of panic automata), where an optimal algorithm for deciding the winner (2-\exptime-complete) was given. The general case was later solved in \cite{HMOS08}. Winning regions were characterised in \cite{BCOS10} and the winning strategies in \cite{CS12} (even if the results are somehow implicit in \cite{HMOS08}). 
Finally, in \cite{BCHS12}, for the case of the reachability winning condition, the approach of \cite{HagueO08} was extended, leading to an algorithm based on the saturation method to compute the winning region, and on top of this algorithm the C-SHORe tool was developed~\cite{BCHS13}.

\definecolor{LightSkyBlue}{rgb}{0.53, 0.81, 0.98}
	\definecolor{Gold}{rgb}{1.0, 0.84, 0.0}

\begin{table}
\begin{center}
\begin{tikzpicture}

\shadedraw [top color=orange,shading angle=-90]  (-5,0) rectangle (7,4);
\shadedraw [top color=LightSkyBlue,shading angle=-90]  (-5,4) rectangle (7,8);
\shadedraw [top color=Gold,shading angle=-90]  (-5,8) rectangle (7,12);

%%Cadre
%\draw[very thick] (8,0) -- (-5,0);
%\draw[very thick] (-5,0) -- (-5,10);
\draw[very thick] (-5,0) rectangle (7,12);

%classes
\node at (-3,-0.5) {\textbf{Pushdown}};
\draw[very thick,dotted] (-1,0) -- (-1,12);
\node at (1,-0.5) {\textbf{$\mathbf{n}$-HOPDA}};
\draw[very thick,dotted] (3,0) -- (3,12);
\node at (5,-0.5) {\textbf{$\mathbf{n}$-CPDA}};
%Pb
\node[rotate=90] at (-5.5,2) {\textbf{Solving}};
\node[rotate=90] at (-5.5,6) {\textbf{Winning region}};
\node[rotate=90] at (-5.5,10) {\textbf{Winning strategy}};

%Solving
\node[text width=5.5cm,align=center] at (-3,2) {Decidable\\ {\cite{MullerS85} + \cite{EJ91}}\\\exptime-complete\\ {\cite{Walukiewicz96,Vardi98,SerrePHD}}};

\node[text width=5.5cm,align=center] at (1,2) {Decidable\\ {\cite{Caucal02} + \cite{EJ91}}\\\bigskip $n$-\exptime-complete\\ {\cite{Cachat03,CHMOS08}}};

\node[text width=5.5cm,align=center] at (5,2) {$n$-\exptime-complete\\ {\cite{HMOS08}}\\ \bigskip See also {\cite{KNUW05}} for\\ a previous study\\ at order-$2$};

%Winning region
\node[text width=5cm,align=center] at (-3,6) {Regular {\cite{Serre03,Cachat02,Hag08,HagueO11}}};

\node[text width=4cm,align=center] at (1,6) {Regular\\ {\cite{CHMOS08,CS08}}\\ \bigskip See also {\cite{BM04,HagueO08}} for reachability using saturation methods};

\node[text width=4cm,align=center] at (5,6) {Regular\\ {\cite{BCOS10}}\\ \bigskip See also {\cite{BCHS12}} for reachability using saturation methods};

%%Winning startegies
\node[text width=3.9cm,align=center] at (-3,10) {Realised by a synchronised pushdown automaton\\ {\cite{Walukiewicz96,SerrePHD}}};
\node[text width=4cm,align=center] at (1,10) {Realised by a synchronised $n$-HOPDA\\ {\cite{CHMOS08,CS08}}};
\node[text width=4cm,align=center] at (5,10) {Realised by a synchronised $n$-CPDA\\ {\cite{HMOS08,CS12}}};

\end{tikzpicture}
\caption{Known results on collapsible pushdown parity games and subclasses.}\label{fig:CPDAgames}
\end{center}
\end{table}

\subsection*{Consequences}

The consequences of the results presented here, together with the equi-expressivity result~\cite{HMOS08,HMOS17,CS12} between  higher-order recursion schemes and collapsible pushdown automata for generating trees, are mainly for the study of logical properties of the infinite trees generated by recursion schemes. In particular, they imply the decidability of the MSO model-checking problem, both its local~\cite{HMOS08} and global version (also known as reflection)~\cite{BCOS10}, and the MSO selection problem (a synthesis-like problem)~\cite{CS12}. 

Due to space constraints, these results are discussed in full detail in a companion paper~\cite{BCOS20}.

\subsection*{Structure of This Paper}

The article is organised as follows. Section~\ref{section:Preliminaries} introduces the main concepts and some intermediate results. In Section~\ref{section:Results} we state our main result. Its proof is by induction and each induction step is divided into three sub-steps, which are respectively described in Section~\ref{section:rankAware} (providing a normal form for CPDA), Section~\ref{section:outermostLinks} (getting rid of the outmost links in the stack structure) and Section~\ref{section:reducingOrder} (reducing the order of the CPDA). Section~\ref{section:summary} summarises the proof and establishes matching upper and lower complexity bounds. Finally, Section~\ref{section:consequences} discusses some logical consequences for collapsible pushdown graphs.

%% file: Preliminaries.tex
\section{Preliminaries}\label{section:Preliminaries}

\subsection{Basic Objects}

An \defin{alphabet} $A$ is a (possibly infinite) set of letters. In the sequel $A^*$ denotes the set of \defin{finite words} over $A$, and $A^\omega$ the set of \defin{infinite words} over $A$. The empty word is written $\epsilon$ and the length of a word $u$ is denoted by $|u|$. Let $u$ be a finite word and $v$ be a (possibly infinite) word. Then $u\cdot v$ (or simply $uv$) denotes the concatenation of $u$ and $v$; the word $u$ is a prefix of $v$ iff there exists a word $w$ such that $v=u\cdot w$.

A \defin{graph} is a pair $G=(V,E)$, where $V$ is a (possibly infinite) set of
\defin{vertices} and $E\subseteq V\times V$ is a (possibly infinite) set of
\defin{edges}. For every vertex $v$ we let $E(v)=\{w\mid (v,w)\in E\}$. A \defin{dead-end} is a vertex $v$ such that $E(v)=\emptyset$.

When $\tau$ is a (partial) mapping, we let $\domain{\tau}$ denote its domain.

\subsection{Two-Player Perfect-Information Parity Games}

An \defin{arena} is a triple $\arena=(G,\VE,\VA)$, where $G=(V,E)$ is a graph and $V=\VE\uplus\VA$ is a partition of the vertices among two players, \Eloise and \Abelard. For simplicity in the definitions, we assume that $G$ has no dead-end.

\Eloise and \Abelard play in $\arena$ by moving a pebble along edges. A \defin{play} from an initial vertex $v_0$ proceeds as follows: the player owning $v_0$ (\ie \Eloise if $v_0\in \VE$, \Abelard otherwise) moves the pebble to a vertex $v_1\in E(v_0)$. Then the player owning $v_1$ chooses a successor $v_2\in E(v_1)$ and so on. As we assumed that there is no dead-end, a play is an infinite word $v_0v_1v_2\cdots \in V^\omega$ such that for all $0\leq i$ one has $v_{i+1}\in E(v_i)$. A \defin{partial play} is a prefix of a play, \ie it is a finite word $v_0v_1\cdots v_\ell \in V^*$  such that for all $0\leq i<\ell$ one has $v_{i+1}\in E(v_i)$.

A \defin{strategy} for \Eloise is a function $\strat_\Ei:V^*V_\Ei\rightarrow V$ assigning, to every partial play ending in some vertex $v\in \VE$, a vertex $v'\in E(v)$. Strategies of \Abelard are defined likewise, and usually denoted $\strat_\Ai$.
In a given play $\play=v_0v_1\cdots$ we say that \Eloise (\resp \Abelard) \defin{respects a strategy} $\strat_\Ei$ (\resp $\strat_\Ai$) if whenever $v_i\in V_\Ei$ (\resp $v_i\in V_\Ai$) one has $v_{i+1} = \strat_\Ei(v_0\cdots v_i)$ (\resp $v_{i+1} = \strat_\Ai(v_0\cdots v_i)$).

A \defin{winning condition} is a subset $\WC\subseteq V^\omega$ and a (two-player perfect information) \defin{game} is a pair $\game=(\arena,\WC)$ consisting of an arena and a winning condition. A game is finite if it is played on a finite arena.

A play $\lambda$ is \defin{won} by \Eloise if and only if $\lambda\in \WC$; otherwise $\lambda$ is won by \Abelard. A strategy $\strat_\Ei$ is \defin{winning} for \Eloise in $\game$ from a vertex $v_0$ if any play starting from $v_0$ where \Eloise respects $\strat_\Ei$ is won by her. Finally a vertex $v_0$ is \defin{winning} for \Eloise in $\game$ if she has a winning strategy $\strat_\Ei$ from $v_0$. Winning strategies and winning vertices for \Abelard are defined likewise.

A \defin{parity winning condition} is defined by a \defin{colouring function} $\col$, \ie a mapping $\col: V \rightarrow \colors \subset \mathbb{N}$, where $\colors$ is a \emph{finite} set of \defin{colours}. The parity winning condition associated with $\col$ is the set $\WC_\col = \{v_0 v_1 \cdots \in V^\omega \mid \liminf (\col(v_i))_{i \geq 0} \text{ is even}\}$, \ie a play is winning if and only if the smallest colour visited infinitely often is even.
A \defin{parity game} is a game of the form $\game=(\arena,\WC_\col)$ for some colouring function.

\subsection{Stacks with Links and Their Operations}

Fix an alphabet $\Gamma$ of \concept{stack symbols} and a distinguished
\concept{bottom-of-stack symbol} $\bot \in \Gamma$. An \concept{order-$0$
  stack} (or simply \concept{$0$-stack}) is just a stack symbol. An
\concept{order-${(n+1)}$ stack} (or simply \concept{${(n+1)}$-stack}) $\stack$ is
a non-null sequence, written $\mksk{\stack_1 \cdots \stack_l}$, of $n$-stacks
such that every non-$\bot$ $\Gamma$-symbol $\gamma$ that occurs in $\stack$ has
a \emph{link} to a stack of some order $e$ (say, where $0 \leq e \leq n$)
situated below it in $\stack$; we call the link an \concept{${(e+1)}$-link}. The
\concept{order} of a stack $\stack$ is written $\order{\stack}$. 
The \concept{height} of a stack $\mksk{\stack_1 \cdots \stack_{l}}$ is defined as $l$.

As usual, the bottom-of-stack symbol $\bot$ cannot be popped from or
pushed onto a stack. Thus we require an \concept{order-1 stack} to be a
non-null sequence $\mksk{\gamma_1 \cdots \gamma_l}$ of elements of $\Gamma$ such
that for all $1 \leq i \leq l$, $\gamma_i = \bot$ iff $i = 1$. We inductively define
$\bot_k$, the \concept{empty ${k}$-stack}, as follows: $\bot_0 = \bot$
and $\bot_{k+1} = \mksk{\bot_k}$.

We first define the operations $\popn{i}$ and $\topn{i}$ with $i \geq
1$: $\topn{i}(\stack)$ returns the top $(i-1)$-stack of $\stack$, and
$\popn{i}(\stack)$ returns $\stack$ with its top $(i-1)$-stack
removed. Precisely let $\stack = \mksk{\stack_1 \cdots \stack_{l+1}}$ be a stack with
$1 \leq i \leq \order{\stack}$:
\[\begin{array}{rll}
\topn{i}(\underbrace{\mksk{\stack_1 \cdots \stack_{l+1}}}_{\hbox{$\stack$}}) & = &
\left\{\begin{array}{ll} \stack_{l+1} & \hbox{if $i = \order{\stack}$}\\
    \topn{i} (\stack_{l+1}) \quad & \hbox{if $i < \order{\stack}$}
\end{array}\right.\\
\popn{i}(\underbrace{\mksk{\stack_1 \cdots \stack_{l+1}}}_{\hbox{$\stack$}}) & = &  
\left\{\begin{array}{ll}
\mksk{\stack_1 \cdots \stack_l} & \hbox{if $i = \order{\stack}$ and $l \geq 1$}\\
\mksk{\stack_1 \cdots \stack_l \,
\popn{i}(\stack_{l+1})} \quad & \hbox{if $i < \order{\stack}$}
\end{array}\right.\\
\end{array}\]
By abuse of notation, we set $\topn{\order{\stack}+1}(\stack) = \stack$. Note that
$\popn{i}(\stack)$ is undefined if $\topn{i+1}(\stack)$ is a one-element
$i$-stack. For example $\popn{2}(\mksk{\mksk{\bot \, \alpha \, \beta}})$ and $\popn{1}(\mksk{\mksk{\bot \, \alpha \, \beta}\mksk{\bot}})$ are both undefined.

There are two kinds of $\mathit{push}$ operations. {We start with the
\emph{order-$1$} $\mathit{push}$}. Let $\gamma$ be a non-$\bot$ stack symbol and
$1 \leq e \leq \order{\stack}$, we define a new stack operation
$\pushlk{e}{\gamma}$ that, when applied to $\stack$, first attaches a link from
$\gamma$ to the $(e-1)$-stack \emph{immediately} below the top
$(e-1)$-stack of $\stack$, then pushes $\gamma$ (with its link) onto the top
1-stack of $\stack$. Formally, for $1 \leq e \leq \order{\stack}$ and $\gamma \in
(\Gamma \setminus \makeset{\bot})$, we define
\[ \pushlk{e}{\gamma}( 
\underbrace{\mksk{\stack_1 \cdots \stack_{l+1}}}_{\hbox{$\stack$}}) = 
\left\{
\begin{array}{ll}
\mksk{\stack_1 \cdots \stack_l \, \pushlk{e}{\gamma}(\stack_{l+1})} \quad &
\hbox{if $e < \order{\stack}$}\\
\mksk{\stack_1 \cdots \stack_l \, \stack_{l+1} \, \gamma^\dag} & \hbox{if
$e = \order{\stack} = 1$}\\
\mksk{\stack_1 \cdots \stack_l \, \pushone{{\widehat{\gamma}}}(\stack_{l+1})} & \hbox{if
$e = \order{\stack} \geq 2$ and $l \geq 1$}\\
\end{array} 
\right.  \]
where 
\begin{itemize}
\item $\gamma^\dag$ denotes the symbol $\gamma$ with a link to the
0-stack $\stack_{l+1}$
\item $\widehat{\gamma}$ denotes the symbol $\gamma$ with a link to the
  $(e-1)$-stack $\stack_l$; and we define
\[\pushone{\widehat{\gamma}}(\underbrace{\mksk{\varstack_1 \cdots \varstack_{r+1}}}_{\hbox{$\varstack$}}) = 
\left\{
\begin{array}{ll}
\mksk{\varstack_1 \cdots \varstack_r \,\pushone{\widehat{\gamma}}(\varstack_{r+1})} \quad &
\hbox{if $\order{\varstack} > 1$}\\
\mksk{\varstack_1 \cdots \varstack_{r+1} \, \widehat{\gamma} } & \hbox{otherwise \ie~$\order{\varstack} = 1$}\\
\end{array}
\right.\] 
\end{itemize}

The higher-order $\pushn{j}$, where $j \geq 2$, simply duplicates the
top $(j-1)$-stack of $\stack$. Precisely, let $\stack = \mksk{\stack_1 \cdots
  \stack_{l+1}}$ be a stack with $2 \leq j \leq \order{\stack}$:
\[\begin{array}{lll}
  \pushn{j}(\underbrace{\mksk{\stack_1 \cdots \stack_{l+1}}}_{\hbox{$\stack$}}) & = & \left\{\begin{array}{ll}
      \mksk{\stack_1 \cdots \stack_{l+1} \, \stack_{l+1}} & \hbox{if $j = \order{\stack}$}\\
      \mksk{\stack_1 \cdots \stack_l \, \pushn{j} (\stack_{l+1})} \quad & \hbox{if $j < \order{\stack}$}
\end{array}\right.\\
\end{array}\]
Note that in case $j = \order{\stack}$ above, the link structure of $\stack_{l+1}$ 
is preserved by the copy that is pushed on top by $\pushn{j}$. 

We also define, for any stack symbol $\gamma$, an operation on stacks that rewrites the topmost stack symbol \emph{without modifying} its link. 
Formally: 
\[\begin{array}{lll}
  \toprew{\gamma} \, \underbrace{\mksk{s_1 \cdots s_{l+1}}}_{\hbox{$s$}} & = & \left\{\begin{array}{ll}
      \mksk{s_1 \cdots s_l \, \toprew{\gamma} s_{l+1}} \quad & \hbox{if $\order{s}>1$}\\
      \mksk{s_1 \cdots s_{l} \, \widehat{\gamma}} & \hbox{if $\order{s}=1$ and $l\geq 1$}
\end{array}\right.\\
\end{array}\]
where $\widehat{\gamma}$ denotes the symbol $\gamma$ with a link to the same target as the link from $s_{l+1}$. Note that $\toprew{\gamma}(\stack)$ is undefined if $\topn{2}(s)$ is the empty $1$-stack.

Finally, there is an important operation called $\collapse$. We say
that the $n$-stack $\stack_0$ is a \concept{prefix} of an $n$-stack $\stack$,
written $\stack_0 \leq \stack$, just in case $\stack_0$ can be obtained from $\stack$ by a
sequence of (possibly higher-order) ${\mathit pop}$ operations.
Take an $n$-stack $\stack$ where $\stack_0 \leq \stack$, for some $n$-stack $\stack_0$, and $\topone{\stack}$ has a link to $\topn{e}(\stack_0)$. Then $\collapse \; \stack$ is defined to be $\stack_0$.

 \begin{example}\label{eg:3stack}
To avoid clutter, when displaying $n$-stacks in examples, we shall omit 1-links (indeed by construction they can only point to the symbol directly below), writing e.g.~$\mksk{\mksk{\bot} \mksk{\bot \alpha \, \beta}}$ instead of 
$\pstr[.1cm][1pt]{
\mksk{
\mksk{\bot}
\mksk{
\nd(n1){\bot}\,\,
\;
\nd(n2-n1,50){\alpha}\,\,
\;
\nd(n3-n2,50){\beta}
}
}
}
$.

    Take the 3-stack $\stack = \mksk{\mksk{\mksk{ \, \bot \, \alpha}} \;
      \mksk{\mksk{ \, \bot} \mksk{ \, \bot \, \alpha}}}$. We have
    \[
       \begin{array}{rll}
         \pushlk{2}{\gamma}(s)  & = &  \pstr[.1cm]{\mksk{\mksk{\mksk{ \, \bot \, \alpha}} \;
                                            \mksk{\mklksk{n1}{ \, \bot}
                                                  \mksk{ \, \bot \, \alpha \, \nd(n2-n1){\gamma}}}}} \\

\collapse \, (\pushlk{2}{\gamma}(s)) & = & \mksk{ \mksk{\mksk{ \, \bot\, \alpha}} \; \mksk{\mksk{ \, \bot}}}
\\
         \underbrace{ \pushlk{3}{\gamma}(\toprew{\beta} ( \pushlk{2}{\gamma}(s)))}_\theta
                              &= & \pstr[.75cm]{\mksk{\mklksk{n1}{\mksk{ \, \bot \, \alpha}} \;
                                             \mksk{\mklksk{n2}{ \, \bot}
                                                   \mksk{ \, \bot \, \alpha \, \nd(n3-n2){\beta}\, \nd(n4-n1){\gamma}}}}}.
       \end{array}
    \]
    Then $\pushn{2} (\theta)$ and $\toprew{\alpha}(\pushn{3}(\theta))$ are respectively
    \[
        \begin{array}{c}
          \pstr[.6cm]{\mksk{\mklksk{n1}{\mksk{ \, \bot \, \alpha}} \;
                            \mksk{\mklksk{n2}{ \, \bot}
                                  \mksk{ \, \bot \, \alpha \, \nd(n3-n2){\beta} \, \nd(n4-n1){\gamma}}
                                  \mksk{ \, \bot \, \alpha \, \nd(n5-n2){\beta} \, \nd(n6-n1){\gamma}}}}}
          \;\hbox{and} \\

          \pstr[1cm]{\mksk{\mklksk{n1}{\mksk{ \, \bot \, \alpha}} \;
                     \mksk{\mklksk{n2}{ \, \bot}
                           \mksk{ \, \bot \, \alpha \, \nd(n3-n2,40){\beta} \, \nd(n4-n1,39){\gamma}}} \;
                     \mksk{\mklksk{n5}{}
                           \mksk{ \, \bot \, \alpha \, \nd(n6-n5,40){\beta} \, \nd(n7-n1,37){\alpha}}}}}.
        \end{array}
    \]

    We have $\collapse \, (\pushn{2}( \theta)) = \collapse \, (\toprew{\alpha}(\pushn{3}(\theta))) = \collapse( \theta)  = \mksk{\mksk{\mksk{ \, \bot \, \alpha}}}$.
  \end{example}

The set $\Op{n}{\Gamma}$ of order-$n$ CPDA \concept{stack operations} over stack alphabet $\Gamma$ (or simply $\Op{n}{}$ if $\Gamma$ is clear from the context) comprises six types of operations:

\begin{enumerate}
\item $\popn{k}$ for each $1 \leq k \leq n$,
\item $\pushn{j}$ for each $2 \leq j \leq n$, 
\item $\pushlk{e}{\gamma}$ for each $1 \leq e \leq n$ and each $\gamma \in (\Gamma
  \setminus \makeset{\bot})$,
\item $\toprew{\gamma}$ for each $\gamma \in (\Gamma
  \setminus \makeset{\bot})$,
\item $\collapse$, and
\item $\id$ for the identity operation (\ie $id(\stack)=\stack$ for all stack $\stack$).
\end{enumerate}

\begin{remark}
One way to give a formal semantics of the stack operations is to work with appropriate numeric representations of the links as explained in \cite[Section~3.2]{HMOS17}. We believe that the informal presentation should be sufficient for this work and hence refer the reader to \cite{HMOS17} for a formal definition of stacks.
\end{remark}

\subsection{Collapsible Pushdown Automata (CPDA) and their Transition Graphs}

\emph{Collapsible pushdown automata} are a generalisation (to all
finite orders) of \emph{pushdown automata with links} \cite{AdMO05a}. %, which are essentially the same as \emph{panic automata} \cite{KNUW05}.
They are defined as automata with a finite control and a stack as memory. In this work, we are interested in CPDA as generators for infinite graphs rather than word acceptors or generators of an infinite tree (see \cite{HMOS17} for corresponding definitions), hence we consider a non-deterministic version of them but do not equip them with an input alphabet.

An \concept{order-${n}$ collapsible pushdown automaton} (\concept{${n}$-CPDA}) is a 4-tuple 
$\mathcal{A} = \anglebra{\Gamma, Q,\Delta, q_0}$, where 
$\Gamma$ is the stack alphabet, 
$Q$ is the finite set of control states, 
$q_0\in Q$ is the initial state, 
and $\Delta \, : \, Q \times \Gamma\rightarrow \, 2^{Q \times \Op{n}{\Gamma}\times \Op{n}{\Gamma}}$ is the transition function and satisfies the following constraint. For any $q,\gamma\in Q\times \Gamma$, for any $(q',op_1,op_2)\in \Delta(q,\gamma)$ one has that $op_1 \in\{\toprew{\alpha}\mid \alpha\in \Gamma\}\cup\{id\}$ and $op_2\notin \{\toprew{\alpha}\mid \alpha\in \Gamma\}$: hence a transition will always act on the stack by (possibly) rewriting the top symbol and then (possibly) performing another kind of operation on the stack. In the following, we will use notation $(q',op_1;op_2)$ instead of $(q',op_1,op_2)$ (to stress that one performs $op_1$ followed by $op_2$).

\begin{remark}
Obviously allowing a top-rewriting operation followed by another stack operation does not add expressive power to the model. However, for technical reasons, this choice simplifies the presentation.
\end{remark}

\concept{Configurations} of an $n$-CPDA are pairs of the form $(q, \stack)$
where $q \in Q$ and $\stack$ is an $n$-stack over $\Gamma$; we call $(q_0,
\bot_n)$ the \concept{initial configuration}.

An $n$-CPDA $\mathcal{A}=\anglebra{\Gamma, Q,\Delta, q_0}$ naturally defines a transition graph $\transgraph{\mathcal{A}} :=(V,E)$ whose
vertices $V$ are the configurations of $\mathcal{A}$ and whose edge relation $E\subseteq V\times V$ is given by: 
$((q,\stack),(q',\stack'))\in E$ iff $\exists (q',op_1;op_2)\in\Delta(q,\topn{1}(\stack))$ such that $\stack'=op_2(op_1(\stack))$. Such a graph is called an \concept{${n}$-CPDA graph}. 

\begin{example}\rm\label{Example:CPDAAnBnCn}
Consider the following $2$-CPDA (that actually does not make use of links) $\mathcal{A}=\anglebra{\{\bot,\alpha\}, \{q_a,q_b,q_c,q_\sharp,\widetilde{q}_a,\widetilde{q}_b,\widetilde{q}_c\},\Delta, \widetilde{q}_a}$ with $\Delta$ as follows (we only give those transitions that may happen):
\begin{itemize}
\item $\Delta(\widetilde{q}_a,\bot)=\{(q_a,\id;\pushone{\alpha})\}$
\item $\Delta({q}_a,\alpha)=\{(q_a,\id;\pushone{\alpha}),(\widetilde{q}_b,\id;\pushn{2})\}$;
\item $\Delta(\widetilde{q}_b,\alpha)=\Delta({q}_b,\alpha)=\{(q_b,\id;\popn{1})\}$;
\item $\Delta({q}_b,\bot)=\{(\widetilde{q}_c,\id;\popn{2})\}$;
\item $\Delta(\widetilde{q}_c,\alpha)=\Delta({q}_c,\alpha)=\{(q_c,\id;\popn{1})\}$;
\item $\Delta({q}_c,\bot)=\{({q}_\sharp,\id;\id)\}$;
\item $\Delta({q}_\sharp,\bot,\_)=\emptyset$.
\end{itemize}

Then $\transgraph{\mathcal{A}}$ is given in Figure \ref{Fig:Example:CPDAAnBnCn}. 

\end{example}

\begin{figure}
\begin{center}
\begin{tikzpicture}[>=stealth',thick,scale=1,transform shape]
\node(a1) at (0,0) {$(\widetilde{q}_a,[[\bot]])$};
\node(a2) at (3,0) {$(q_a,[[\bot\alpha]])$};
\node(a3) at (6,0) {$(q_a,[[\bot\alpha\alpha]])$};
\node(a4) at (10,0) {$(q_a,[[\bot\alpha\alpha\alpha]])$};
\node(a5) at (12,0) {};
\draw[->] (a1) to (a2);
\draw[->] (a2) to (a3);
\draw[->] (a3) to (a4);
\draw[->,dotted] (a4) to (a5);
%%%
\node(b2) at (3,-1.5) {$(\widetilde{q}_b,[[\bot\alpha][\bot\alpha]])$};
\node(c2) at (3,-3) {$({q}_b,[[\bot\alpha][\bot]])$};
\node(d2) at (3,-4.5) {$(\widetilde{q}_c,[[\bot\alpha]])$};
\node(e2) at (3,-6) {$({q}_c,[[\bot]])$};
\draw[->] (a2) to (b2);\draw[->] (b2) to (c2);\draw[->] (c2) to (d2);\draw[->] (d2) to (e2);
\node(e1) at (0,-6) {$({q}_{\sharp},[[\bot]])$};
\draw[->] (e2) to (e1);
%%%
\node(b3) at (6,-1.5) {$(\widetilde{q}_b,[[\bot\alpha\alpha][\bot\alpha\alpha]])$};
\node(c3) at (6,-3) {$({q}_b,[[\bot\alpha\alpha][\bot\alpha]])$};
\node(d3) at (6,-4.5) {$({q}_b,[[\bot\alpha\alpha][\bot]])$};
\node(e3) at (6,-6) {$(\widetilde{q}_c,[[\bot\alpha\alpha]])$};
\node(f3) at (6,-7.5) {$({q}_c,[[\bot\alpha]])$};
\draw[->] (a3) to (b3);\draw[->] (b3) to (c3);\draw[->] (c3) to (d3);\draw[->] (d3) to (e3);\draw[->] (e3) to (f3);\draw[->] (f3) to (e2);
%%%
\node(b4) at (10,-1.5) {$(\widetilde{q}_b,[[\bot\alpha\alpha\alpha][\bot\alpha\alpha\alpha]])$};
\node(c4) at (10,-3) {$({q}_b,[[\bot\alpha\alpha\alpha][\bot\alpha\alpha]])$};
\node(d4) at (10,-4.5) {$({q}_b,[[\bot\alpha\alpha\alpha][\bot\alpha]])$};
\node(e4) at (10,-6) {$({q}_b,[[\bot\alpha\alpha\alpha][\bot]])$};
\node(f4) at (10,-7.5) {$(\widetilde{q}_c,[[\bot\alpha\alpha\alpha]])$};
\node(g4) at (10,-9) {$({q}_c,[[\bot\alpha\alpha]])$};
\draw[->] (a4) to (b4);\draw[->] (b4) to (c4);\draw[->] (c4) to (d4);\draw[->] (d4) to (e4);\draw[->] (e4) to (f4);\draw[->] (f4) to (g4);\draw[->] (g4) to (f3);
\node(g5) at (12,-9) {};
\draw[->,dotted] (g5) to (g4);
\end{tikzpicture}
\end{center}
\caption{Transition graph of the CPDA of Example \ref{Example:CPDAAnBnCn}.}\label{Fig:Example:CPDAAnBnCn}
\end{figure}
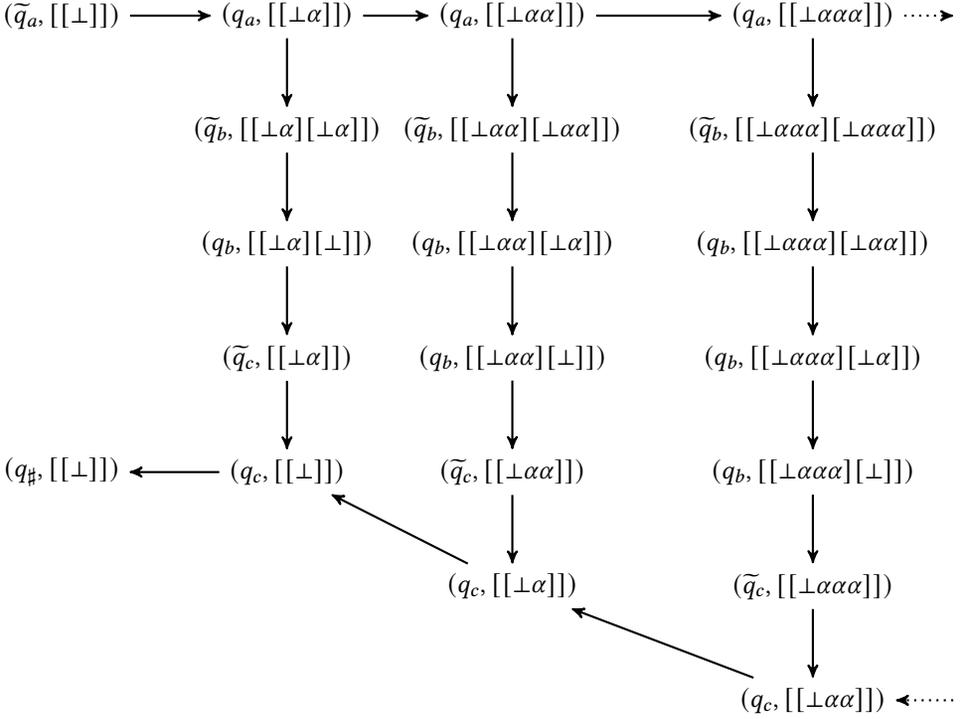

\subsection{CPDA Parity Games}

We now explain how CPDA can be used to define parity games. Let $\mathcal{A}=\anglebra{\Gamma, Q,\Delta, q_0}$ be an order-$n$ CPDA and let $\transgraph{\mathcal{A}}=(V,E)$ be its transition graph. 
Let $Q_\Ei\uplus Q_\Ai$ be a partition of $Q$ and let $\col:Q \mor
\colors\subset\mathbb{N}$ be a colouring function (over states). Altogether they
define a partition $V_\Ei\uplus V_\Ai$ of $V$, whereby a vertex belongs
to $V_\Ei$ iff its control state belongs to $Q_\Ei$, and a colouring
function $\col:V \mor \colors$, where a vertex is assigned the colour
of its control state. The structure
$\pggraph=(\transgraph{\mathcal{A}},V_\Ei,V_\Ai)$ defines an arena and the pair $\pgame=(\pggraph,\WC_\col)$ defines a parity game that we call an 
\defin{$n$-CPDA parity game}. 

Given an $n$-CPDA parity game, there are three main algorithmic questions:
\begin{enumerate}
\item Decide whether $(q_0,\bot_{n})$ is winning for \Eloise.
\item Provide a description of the winning region for \Eloise.
\item If $(q_0,\bot_{n})$ is winning for \Eloise, provide a description of a winning strategy for \Eloise from $(q_0,\bot_{n})$.
\end{enumerate}

\begin{remark}
Note that the first question is equivalent to the following one: given a vertex $v\in V$ decide whether $v$ is winning for \Eloise. Indeed, one can always design a new $n$-CPDA parity game that simulates the original one except that from the initial configuration the players are first forced to go to $v$, from where the simulation really starts.
\end{remark}

To answer the second question, we will introduce the notion of \emph{regular sets of stacks}, and to answer the third one we will consider \emph{strategies realised by $n$-CPDA transducers}.

\subsection{Regular Sets of Stacks with Links}\label{ssection:regSets}

We start by introducing a class of automata with a finite state-set that can be used to {recognize} sets of stacks. 
Let $\stack$ be an order-$n$ stack. We first associate with $\stack=\stack_{1},\cdots,\stack_{\ell}$ a well-bracketed word of depth $n$, $\flaten{\stack}\in(\Sigma\cup\{\lsk,\rsk\})^*$:
  \[ \flaten{\stack} \; := \; \begin{cases}
                   \lsk\flaten{\stack_{1}}\cdots\flaten{\stack_{\ell}}\rsk & \text{if }n\geq 1\\
                   \stack & \text{if }n=0 \text{ (\ie $\stack\in\Sigma$)}\\
                 \end{cases}
  \]
In order to reflect the link structure, we define a partial function $\target{\stack}:\{1,\cdots ,|\flaten{\stack}|\}\rightarrow \{1,\cdots ,|\flaten{\stack}|\}$ that assigns to every position in $\{1,\cdots ,|\flaten{\stack}|\}$ the index of the end of the stack targeted by the corresponding link (if exists; indeed this is undefined for $\bot, \lsk$ and $\rsk$). Thus with $\stack$ is associated the pair $\anglebra{\flaten{\stack},\target{\stack}}$; and with a set $S$ of stacks is associated the set $\widetilde{S}=\{\anglebra{\flaten{\stack},\target{\stack}}\mid \stack\in S\}$.

\begin{example}
Consider the stack $\stack=\pstr[.75cm]{\mksk{\mklksk{n1}{\mksk{ \bot \, \alpha}} \;
                                             \mksk{\mklksk{n2}{ \bot } 
                                                   \mksk{ \bot \, a \, \nd(n3-n2){\beta}\, \nd(n4-n1){\gamma}}}}}$. 
                                                   Then $$\flaten{\stack} = \pstr[.75cm]{\mksk{\mklksk{n1}{\mksk{\, \bot \, \alpha}} \;
                                             \mksk{\mklksk{n2}{\bot}
                                                   \mksk{ \,\bot\, \alpha \, {\beta}\, {\gamma}}}}}$$
and $\target{\stack}=\tau$ where
$\tau(5)=4$, $\tau(14)=13$, $\tau(15)=11$ and $\tau(16)=7$.
\end{example}

We consider \emph{deterministic} finite automata working on such representations of stacks. The automaton reads the word $\flaten{\stack}$ from left to right {(that is, from bottom to top)}. On reading a letter that does not have a link (i.e.~$\targetf$ is undefined on its index) the automaton updates its state according to the current state and the letter; on reading a letter that has a link, the automaton updates its state according to the current state, the letter and the state it was in after processing the targeted position. A run is accepting if it ends in a final state. 
{One can think of these automata as a deterministic version of {Stirling's \emph{dependency tree automata} \cite{Sti09b}} restricted to words.}

Formally, an \concept{automaton} is a tuple $\anglebra{R,A,r_{in},F,\delta}$ where $R$ is a finite set of states, $A$ is a finite input alphabet, $r_{in}\in R$ is the initial state, $F\subseteq R$ is a set of final states and $\delta: (R\times A) \cup (R\times A\times R) \rightarrow R$ is a transition function. With a pair $\anglebra{u,\tau}$ where $u=a_1\cdots a_n\in A^*$ and $\tau$ is a partial map from $\{1,\cdots n\}\rightarrow \{1,\cdots n\}$, we associate a \emph{unique} run $r_0\cdots r_n$ as follows:
\begin{itemize}
\item $r_0=r_{in}$;
\item for all $0\leq i< n$, $r_{i+1}=\delta(r_i,a_{i+1})$ if $i+1\notin Dom(\tau)$;
\item for all $0\leq i< n$, $r_{i+1}=\delta(r_i,a_{i+1},r_{\tau(i+1)})$ if $i+1\in Dom(\tau)$.
\end{itemize} 
The run is \emph{accepting} just if $r_{n}\in F$, and the pair $(u,\tau)$ is \emph{accepted} just if the associated run is accepting.

To recognise configurations instead of stacks, we use the same machinery but now add the control state at the end of the coding of the stack. We code a configuration $(q,\stack)$ as the pair $\anglebra{\flaten{\stack}\cdot q,\target{\stack}}$ (hence the input alphabet of the automaton also contains a copy of the control state of the corresponding CPDA). 

Finally, we say that a set $L$ of $n$-stacks over alphabet $\Gamma$ is \defin{regular} just if there is an automaton $\mathcal{B}$ such that for every $n$-stack $\stack$ over $\Gamma$, $\mathcal{B}$ accepts $\anglebra{\flaten{\stack},\target{\stack}}$ iff $\stack\in L$. Regular sets of configurations are defined in the same way.

Regular sets of stacks ({resp}.~configurations) form an effective Boolean algebra.

\begin{property}
Let $L_1,L_2$ be regular sets of $n$-stacks over an alphabet $\Gamma$.  Then $L_1\cup L_2$, $L_1\cap L_2$ and $Stacks(\Gamma)\setminus L_1$ are also regular (here $Stacks(\Gamma)$ denotes the set of all stacks over $\Gamma$).  The same holds for regular sets of configurations.
\end{property}

\begin{proof}
Closure under complement comes from the fact that we consider \emph{deterministic} automata. Closure under union or intersection is achieved by considering a Cartesian product, as in the case of finite automata on finite words.
\end{proof}

The following result shows that the notion of regular sets of $n$-stacks is robust with respect to the computational model of CPDA. The result is used only when discussing consequences in Section~\ref{section:markingWR} and therefore its proof can safely be skipped by the reader.

\begin{theorem}\label{theo:closure-reg-test}
Let $\mathcal{A}$ be an order-$n$ CPDA with a state-set $Q$ and a stack alphabet $\Gamma$, and let $L$ be a regular set of configurations.

Then, one can build an order-$n$ CPDA $\mathcal{A}'$ with a state-set $Q'$, a subset $F\subseteq Q'$ and a mapping $\chi:Q'\rightarrow Q$ such that the following holds.
\begin{enumerate}
\item Restricted to the reachable configurations from their respective initial configuration, the transition graph of $\mathcal{A}$ and $\mathcal{A'}$ are isomorphic.
\item For every configuration $(q,\stack)$ of $\mathcal{A}$ that is reachable from the initial configuration, the corresponding configuration $(q',s')$ of $\mathcal{A}$ is such that $q=\chi(q')$ and belongs $(q,s)$ belongs to $L$ if and only if $q'\in F$.
\end{enumerate}
\end{theorem}

\os{Do we put the proof here or do we postpone it to the final section where we actually need the result? I ask the question because it is the first proof of the paper and it is fairly technical even if I hope it is simpler to understand now than it was in the original LiCS10 paper…}
\os{The following proof was deeply rewritten and definitely needs proof reading (there may be some issues with indices…)}
\begin{proof}
Fix an order-$n$ CPDA $\mathcal A$ and an automaton $\mathcal{B}=\anglebra{R,\Gamma\cup\{[,]\},r_{in},F,\delta}$ accepting $L$.

Let $s$ be an order-$n$ stack. Let $0\leq k\leq n$ and let $t$ be the topmost $k$-stack of $s$, \ie $t=\topn{k+1}(s)$. We are interested in describing how $\mathcal{B}$ behaves when reading $\popn{k}(t)$  (for some technical reason we do not care of the topmost $(k-1)$-stack in $t$ as we will later compose those behaviours), with the convention that $\popn{0}(t) = t$. If there was no link, this behaviour could simply be described as a function from $R$ into $R$. However, as we extracted $t$ from $s$, there may be some “dangling link” of order greater than $k$.

We refer to Figure~\ref{fig:closure-reg-test} for an illustration of the concepts below for the case where $n=4$.
To retrieve the states attached to the respective targets of the links (of order $n, \cdots, k+1$ respectively) in $s$, we will use as a parameter $n-k$ states $r_n,\cdots,r_{k+1}$ in $R$: for $n$-links, we consider the run induced by reading $s$ starting from ${r_n}$ and this gives the values for the respective targets of the $n$-links; for $(n-1)$-links, we consider the run induced by reading $\topn{n}(s)$ starting from $r_{n-1}$ (note that states in dangling $n$-links are known thanks to $r_n$ from the previous step) and this gives the values for the respective targets of the $(n-1)$-links; \ldots; and for $(k+1)$-links, we consider the run induced by reading $\topn{k+2}(s)$ starting from $r_{k+1}$ (note that states in dangling $i$-links for $i>k$ are known thanks to $r_i$) and this gives the values for the respectives targets of the $(k+1)$-links.

Hence, we associate with $t$ a function $\tau_k:R^{n-k}\rightarrow(R\rightarrow R)$ such that $\tau_k(r_n,\dots,r_{k+1})$ defines a function from $R$ into $R$ that maps every state $r\in R$ to the state $\tau_k(r_n,\dots,r_{k+1})(r)$ that is reached by $\mathcal{B}$ when reading $\popn{k}(t)$ starting from $r$ and where the states attached to the respective targets of the links are determined by $r_n,\cdots,r_{k+1}$ as explained above.

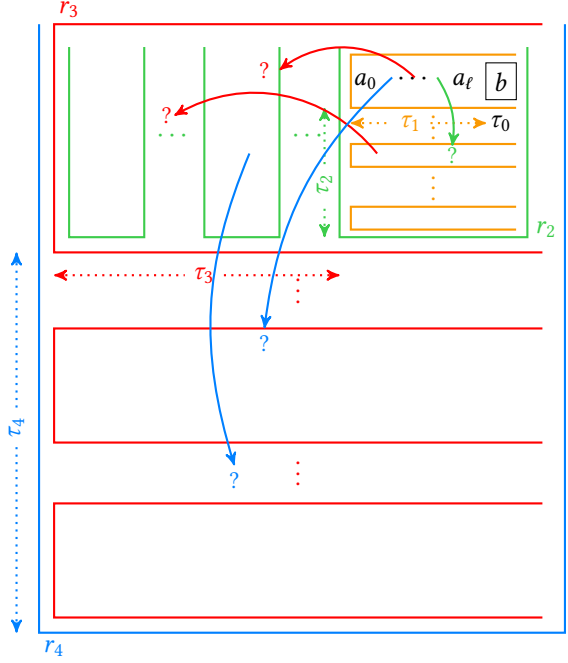
\begin{figure}
	\begin{tikzpicture}[>=stealth',scale=1,transform shape]
		\tikzstyle{ordre4}=[color=Blue];
		\tikzstyle{ordre3}=[color=red];
		\tikzstyle{ordre2}=[color=Vert];
		\tikzstyle{ordre1}=[color=Orange];
		\tikzstyle{pile}=[thick];
		\tikzstyle{tau}=[dotted,thick];
		\tikzstyle{lien}=[thick,->];
		\draw[ordre4,pile] (0,0) -- (0,-8) -- (7,-8) -- (7,0);
		\node[ordre4] at (.2,-8.2) {$r_4$};
		\node[rotate=90,ordre4] (tau4) at (-.3,-5.3) {$\tau_4$};
		\draw[ordre4,tau,<-] (-0.3,-8) -- (tau4); \draw[ordre4,tau,->] (tau4) -- (-0.3,-3);
		\begin{scope}[shift={(0.2,0)}]
			\draw[ordre3,pile] (6.5,0) -- (0,0) -- (0,-3)-- (6.5,-3);
			\node[ordre3,] at (.2,0.2) {$r_3$};
			\node[ordre3] (tau3) at (2,-3.3) {$\tau_3$};
			\draw[ordre3,tau,<-] (0,-3.3) -- (tau3);\draw[ordre3,tau,->] (tau3) -- (3.8,-3.3);
			\begin{scope}[shift={(0.2,-.3)}]
				\draw[ordre2,pile] (0,0) -- (0,-2.5) -- (1,-2.5)-- (1,0);
			\end{scope}
			\node[ordre2] at (1.6,-1.5) {$\cdots$};
			\begin{scope}[shift={(2,-.3)}]
				\draw[ordre2,pile] (0,0) -- (0,-2.5) -- (1,-2.5)-- (1,0);
			\end{scope}
			\node[ordre2] at (3.4,-1.5) {$\cdots$};
			\begin{scope}[shift={(3.8,-.3)}]
				\draw[ordre2,pile] (0,0) -- (0,-2.5) -- (2.5,-2.5)-- (2.5,0);
				\node[ordre2] at (2.75,-2.4) {$r_2$};
				\node[ordre2,rotate=90] (tau2) at (-.2,-1.8) {$\tau_2$};
				\draw[ordre2,tau,<-] (-.2,-2.5) -- (tau2);\draw[ordre2,tau,->] (tau2) -- (-.2,-.8);
				\begin{scope}[shift={(0.15,-.1)}]
					\draw[ordre1,pile] (2.2,0) -- (0,0) -- (0,-.7) -- (2.2,-.7); 
					\node[ordre1] (tau1) at (.8,-.9) {$\tau_1$};
					\draw[ordre1,tau,<-] (0,-.9) -- (tau1);\draw[ordre1,tau,->] (tau1) -- (1.8,-.9);

					\node at (.2,-.35) {$a_0$};
					\node at (.85,-.35) {$\cdots$};					
					\node at (1.5,-.35) {$a_\ell$};
					\node[rectangle,draw] at (2,-.35) {$b$};
					\node at (2,-.9) {$\tau_0$};
				\end{scope}
				\node[ordre1] at (1.25,-0.95) {$\vdots$};
				\begin{scope}[shift={(0.15,-1.275)}]
					\draw[ordre1,pile] (2.2,0) -- (0,0) -- (0,-.3) -- (2.2,-.3); 
				\end{scope}
				\node[ordre1] at (1.25,-1.75) {$\vdots$};
				\begin{scope}[shift={(0.15,-2.1)}]
					\draw[ordre1,pile] (2.2,0) -- (0,0) -- (0,-.3) -- (2.2,-.3); 
				\end{scope}
			\end{scope}
		\end{scope}
		\node[ordre3] at (3.45,-3.4) {$\vdots$};
		\begin{scope}[shift={(0.2,-4)}]
			\draw[ordre3,pile] (6.5,0) -- (0,0) -- (0,-1.5)-- (6.5,-1.5);
		\end{scope}
		\node[ordre3] at (3.45,-5.8) {$\vdots$};
		\begin{scope}[shift={(0.2,-6.3)}]
			\draw[ordre3,pile] (6.5,0) -- (0,0) -- (0,-1.5)-- (6.5,-1.5);
		\end{scope}
		\path[lien,ordre2] (5.3,-0.7) edge [bend left=20] (5.5,-1.6);
		\node[ordre2] at (5.5,-1.725) {\small ?};
		\path[lien,ordre3] (5,-0.7) edge [bend right=40] (3.2,-.6);
		\node[ordre3] at (3,-.6) {\small ?};
		\path[lien,ordre3] (4.5,-1.7) edge [bend right=40] (1.8,-1.2);
		\node[ordre3] at (1.7,-1.2) {\small ?};
		\path[lien,ordre4] (4.7,-0.7) edge [bend right=20] (3,-4);
		\node[ordre4] at (3,-4.2) {\small ?};
		\path[lien,ordre4] (2.8,-1.7) edge [bend right=20] (2.6,-5.8);
		\node[ordre4] at (2.6,-6) {\small ?};
	\end{tikzpicture}
	\caption{Illustration for the proof of Theorem~\ref{theo:closure-reg-test} when $n=4$. Missing states ({\small ?}) in $k$-link's target are retrieve by reading $\topn{k+1}(s)$ from $r_k$. For every $k$, $\tau_k^+$ is obtained by composing the $\tau_i$s for $i\leq k$.}\label{fig:closure-reg-test}
\end{figure}

A stack symbol of the CPDA $\mathcal{A}'$, is a pair, consisting of a  stack symbol of $\mathcal A$, and an $(n+1)$-tuple of the form $(\tau_{n}, \cdots, \tau_{0})$ where the $\tau_{i}$s are as above.

As the function $\tau_k$ describes the behaviour of $\popn{k}(\topn{k+1}(s))$, if we want to reconstruct the behaviour of $\topn{k+1}(s)$ we need to compose, in the appropriate way, the various $\tau_i$ function for $i\leq k$ which leads the following definition. We define $\tau^+_0(r_n \cdots r_1)$ to be the same function as $\tau_0(r_n \cdots,r_1)$; and for each $1 \leq k \leq n$, $$\tau^+_{k}(r_n \cdots r_{k+1}): \begin{cases}
 	R\rightarrow R\\
r\mapsto  \tau_{k-1}^+(r_n \cdots r_{k})(\tau_{k}(r_n \cdots r_{k+1})(r)) 	\end{cases}
$$.

Hence, each $\tau^+_k$ is a function from $R$ to $R$ induced by reading (the segment of) $s$ starting from $\topn{k+1}(s)$. As each $\tau^+_k$ can be obtained from the $\tau_i$s, we safely assume that we can access them directly in $\mathcal{A}'$ when reading the $\topn{1}$ element of the stack. Note that, considering $\tau^+_n$ applied to the initial state $r_\ini$ of $\mathcal{B}$ we deduce whether the current stack is accepted by $\mathcal{B}$: hence this information will be maintained, together with a state from $Q$, in the control state of $\mathcal{A}'$ and is used to define $F$. The function $\chi$ is the one erasing all auxiliary informations used by $\mathcal{A}'$ in its control state.

We now explain how $\mathcal{A}'$ behaves. Assume that the topmost stack symbol is $(a, (\tau_{n}, \cdots, \tau_{0}))$ and that the $\mathcal{A}$-state stored is $q$. Then, the possible transitions of $\mathcal{A}'$ mimic the ones of $\mathcal{A}$ when being in state $q$ with topmost stack symbol $a$. For each order-$n$ stack operation $op$ of $\mathcal A$, we define the corresponding stack operation of $\mathcal{A}'$:
\begin{itemize}
\item If $op=\pushn{k}$ then $\mathcal{A}'$ performs $\pushn{k}$ followed by $\toprew{a,(\tau_{n}, \cdots, \tau_{k+1}, \tau, \tau_{k-1}, \cdots, \tau_0)}$, where for every $r\in R$, $\tau(r_{n}, \cdots, r_{k+1})(r)=\delta(\tau_{k-1}^+(r_n, \cdots, r_{k+1},r)(r'),\rsk_{k})$ with $r'=\tau_k(r_n, \cdots, r_{k+1})(r)$. Indeed, after performing a $\pushn{k}$ operation the only $\topn{i}$ stack that is different from the one before, is for $i=k$. Hence, one only needs to update $\tau_k$, which now maps a state $r$ to the state $r'$ obtained by first applying the previous $\tau_k$ followed by the transformation induced by the former top ${k-1}$-stack (with the missing $k$-links being retrieve starting from $r$) together with the missing closing parenthesis $\rsk_k$.
\item If $op=\pushlk{k}{b}$ then $\mathcal{A}'$ performs $\pushlk{k}{(b,(\tau_{n}, \cdots, \tau_2,\tau, \tau_b))}$ where 
$\tau$ and $\tau_b$ are defined as follows. The function $\tau$ is equal to $\tau_1^+$ while the function $\tau_b(r_n,\dots,r_1)$ maps a state $r$ to $\delta(r,b,\tau_k(r_n,\dots,r_{k+1})(r_k))$. Indeed, one simply has to update $\tau_1$ and $\tau_0$. Regarding $\tau_1$ one needs now to take into the former topmost symbol which is exactly what does $\tau_1^+$. For $\tau_0$ one simulates the behaviour of $\mathcal{B}$ when reading a $b$ and uses $\tau_k$ with the appropriate parameters to retrieve the state in the target of the newly created link.
\item If $op=\popn{k}$ (\resp $\collapse$ following a $k$-link) then $\mathcal{A}'$ performs $\popn{k}$ (\resp $\collapse$), considers the new topmost stack symbol $(a',(\tau'_{n},\cdots,\tau'_0))$ and does a $\toprew{(a',(\tau_{n},\cdots,\tau_{k+1},\tau'_{k},\dots\tau'_0)}$. Indeed, for any stack $s$ and any $i>k$, $\popn{i}(\topn{i+1}(s))=\popn{i}(\topn{i+1}(\popn{k}(s)))$ and therefore $\tau_n,\cdots, \tau_{k+1}$ are inherited from the previous configuration while the other components are preserved from the last time where (possibly a copy of) the topmost symbol was on top of the stack (being inductively assumed to be correct).
\end{itemize}

Correctness of the construction follows inductively from the above definition.
\end{proof}

\subsection{CPDA strategies}

Let $\mathcal{A}=\anglebra{\Gamma, Q,\Delta, q_0}$ be an order-$n$ CPDA, let $\transgraph{\mathcal{A}}=(V,E)$ be its transition graph, let
$\pggraph=(\transgraph{\mathcal{A}},V_\Ei,V_\Ai)$ be an arena associated with $\mathcal{A}$ and let $\pgame=(\pggraph,\WC_\col)$ be a corresponding $n$-CPDA parity game.

We aim at defining a notion of $n$-CPDA \emph{transducers} that provide a description for strategies in $\pgame$, that is describe a function from partial plays in $\pgame$ into $V$. 

Consider a partial play $\play=v_0v_1\cdots v_\ell$ in $\pgame$ where $v_0 = (q_0,\bot_n)$. An alternative description of $\play$ is by a sequence $(q_1,rew_1;op_1)\cdots (q_\ell,rew_\ell;op_\ell)\in (Q\times\Op{n}{\Gamma}\times\Op{n}{\Gamma})^*$ such that $v_i=(q_i,\stack_i)$ for all $1\leq i\leq \ell$ and $\stack_i=op_i(rew_i(\stack_{i-1}))$ (with the convention that $\stack_0=\bot_n$). We may in the following use implicitly this representation of $\play$ when needed. Similarly, one can represent a strategy as a (partial) function $$\strat:(Q\times\Op{n}{\Gamma}\times\Op{n}{\Gamma})^*\rightarrow Q\times\Op{n}{\Gamma}\times\Op{n}{\Gamma}$$ the meaning being that in a partial play $\play$ ending in some vertex $(q,\stack)$ if $\strat(\play)=(q',rew;op)$ then the player moves to $(q',op(rew(\stack)))$.

An \concept{$n$-CPDA transducer realising a strategy} in $\pgame$ is a tuple $\mathcal{S}=\anglebra{\Sigma, R,\delta,\tau, r_0}$ where $\Sigma$ is a stack alphabet, $R$ is a finite set of states, $r_0\in R$ is the initial state,  
$$\delta: R\times \Sigma \times (Q\times \Op{n}{\Gamma}\times \Op{n}{\Gamma}) \rightarrow R\times \Op{n}{\Sigma}\times \Op{n}{\Sigma}$$
is a \emph{deterministic} transition function and
$$
\tau:R\times \Sigma \rightarrow Q\times\Op{n}{\Gamma}\times\Op{n}{\Gamma}
$$
is a deterministic choice function (note that we do not require $\tau$ to be total). For both $\delta$ and $\tau$ we have the same requirement as for the transition function for CPDAs, namely that the first stack operation should be a top-rewriting (or the identity) and that the second one should not be a top-rewriting.

A configuration of $\mathcal{S}$ is a pair $(r,\varstack)$ where $r$ is a state and $\varstack$ is an $n$-stack over $\Sigma$; the initial configuration of $\mathcal{S}$ is $(r_0,\bot_n)$. With a configuration $(r,\varstack)$ is associated, when defined, a (unique) move in $\pgame$ given by $\tau(r,\topn{1}(\varstack))$.
A partial play $\play= (q_1,rew_1;op_1)\cdots (q_\ell,rew_\ell;op_{\ell})$ in $\pgame$ induces a (unique, when defined) \concept{run} of $\mathcal{S}$ which is the sequence
$$ 
(r_0,\varstack_0)(r_1,\varstack_1)\cdots(r_\ell,\varstack_\ell)
$$
where $(r_0,\varstack_0)=(r_0,\bot_n)$ is the initial configuration of $\mathcal{S}$ and for all $0\leq i\leq \ell-1$ one has $\delta(r_i,\topn{1}(\varstack_i),(q_{i+1},rew_{i+1};op_{i+1}))=(r_{i+1},rew_{i+1}';op_{i+1}')$ with $\varstack_{i+1} = op_{i+1}'(rew_{i+1}'(\varstack_i))$. In other words, the control state and the stack of $\mathcal{S}$ are updated accordingly to $\delta$.

We say that $\mathcal{S}$ is \concept{synchronised} with $\mathcal{A}$ iff for all $(r,\alpha,(q,rew;op))\in R\times \Sigma \times (Q\times \Op{n}{\Gamma}\times \Op{n}{\Gamma}) $ such that $\delta(r,\alpha,(q,rew;op))=(r',rew';op')$ is defined one has that $op$ and $op'$ are of the same kind, \ie either they are both a $\popn{k}$ (for the same $k$) or both a $\pushn{k}$ (for the same $k$) or both a $\pushlk{e}{\_}$ (the symbol pushed being possibly different but the order of the link being the same) or both $\collapse$ or both $\id$. In particular, if one defines the \concept{shape} of a stack $\stack$ as the stack obtained by replacing all symbols appearing in $\stack$ by a fresh symbol $\sharp$ (but keeping the links) one has the following.

\begin{proposition}
Assume that $\mathcal{S}$ is synchronised with $\mathcal{A}$. Then, for any partial play $\play$ in $\pgame$ ending in a configuration with stack $\stack$, the run of $\mathcal{S}$ on $\play$, when exists, ends in a configuration with stack $\varstack$ such that $\stack$ and $\varstack$ have the same shape.
\end{proposition}

The \concept{strategy realised by $\mathcal{S}$} is the (partial) function $\stratCPDA{\mathcal{S}}$ defined by letting 
$\stratCPDA{\mathcal{S}}(\play)=\tau((r,\topn{1}(\varstack)))$ where $(r,\varstack)$ is the last configuration of the run of $\mathcal{S}$ on $\play$.

We say that $\stratCPDA{\mathcal{S}}$ is \concept{well-defined} iff 
for any partial play $\play=(q_1,rew_1;op_1)\cdots (q_\ell,rew_\ell;op_{\ell})$ where \Eloise respects $\stratCPDA{\mathcal{S}}$ whenever the last vertex $(q_\ell,\stack_\ell)$ in $\play$ belongs to $V_\Ei$ one has $\stratCPDA{\mathcal{S}}(\play)\in\Delta(q,\topn{1}(\stack_\ell))$, \ie the move given by $\stratCPDA{\mathcal{S}}$ is a valid one.

%% file: Results.tex
% !TEX root = main.tex
\section{Main Result}\label{section:Results}

The following theorem is the central result of this paper.

\begin{theorem}\label{theorem:main}
Let $\mathcal{A}=\anglebra{\Gamma, Q,\delta, q_0}$ be an $n$-CPDA and let $\pgame$ be an $n$-CPDA parity game defined from $\mathcal{A}$. Then one has the following results.
\begin{enumerate}
\item Deciding whether $(q_0,\bot_n)$ is winning for \Eloise is an $n$-\exptime-complete problem.
\item The winning region for \Eloise (\emph{resp.} for \Abelard) is regular. Moreover, one can compute an automaton that recognises it.
\item If $(q_0,\bot_n)$ is winning for \Eloise then one can effectively construct an $n$-CPDA transducer $\mathcal{S}$ \emph{synchronised} with $\mathcal{A}$ realising a well-defined \emph{winning} strategy $\S$ for \Eloise in $\pgame$ from $(q_0,\bot_n)$.
\end{enumerate}
\end{theorem}

The proof is by induction on the order and each induction step is itself divided into three steps: the first one is a normalisation result (Section~\ref{section:rankAware}), the second one removes the outermost links (Section~\ref{section:outermostLinks}) while the third one lowers the order (Section~\ref{section:reducingOrder}). Finally Section~\ref{section:summary} combines the previous constructions and provides the proof of Theorem~\ref{theorem:main}.

%Before going to the proof, we give in Section~\ref{section:rankAware} a normalisation result (Theorem~\ref{lemma:rank-aware}). Then Section~\ref{section:outermostLinks} explains how to remove the outermost links and Section~\ref{section:reducingOrder} shows how to reduce the order. Finally Section~\ref{section:summary} combines the previous constructions and provides the proof of Theorem~\ref{theorem:main}.

%% file: RankAware.tex
% !TEX root = main.tex

\section{Rank-aware CPDA}\label{section:rankAware}

Intuitively, a CPDA is ``rank-aware'' whenever, during any run of the CPDA, one can easily determine the smallest colour seen since the creation of the link on the topmost symbol.  In particular, one only needs to inspect the current control state and topmost stack symbol.  This information will be crucial in the next section when we show how to remove the outermost links from a CPDA.  In this section, we show that any CPDA can be transformed into an equivalent rank-aware CPDA.  The notion of equivalence is formalised in the statement of Theorem~\ref{lemma:rank-aware}.

Fix, for the whole section, an $n$-CPDA $\pprocess=\anglebra {\Gamma,Q,\Delta,q_0}$, a partition $Q_\Ei\uplus Q_\Ai$ of $Q$ and a colouring function $\col:Q\rightarrow\colors\subset\mathbb{N}$. Denote by $\pgraph$ its transition graph, by $\pggraph$ the arena induced by $\pgraph$ and the partition $Q_\Ei\uplus Q_\Ai$ and by $\pgame$ the parity game $(\pggraph,\WC_\col)$.

\subsection{Definitions}

Our main goal in this sub-section is to define the notion of rank-awareness.  To do this we will define the notion of \emph{link-rank}. Assume that in configuration $v_m$ the $\topn{1}$-element has a link (that is possibly a copy of a link) that was created in configuration $v_j$: then the link-rank in $v_m$ is defined as the smallest colour since the creation of the link, \ie $\min\{\col(v_j),\cdots \col(v_m)\}$.  Ultimately, we will show how to enrich the stack alphabet to be able to compute the link-rank. In order to maintain this information, we need to define several other concepts. First we will define indexed stacks, from which, we can then define the \emph{collapse-rank} (for updating after performing a $\collapse$) and the \emph{pop-rank} for $k$ (for updating after performing a $\popn{k}$).

A \defin{finite path} in $\pgraph$ is a non-empty sequence of configurations
$v_0 v_1 \cdots v_m$ such that for all $0\leq i\leq m-1$, there is an edge 
in $\pgraph$ from $v_i$ to $v_{i+1}$. 
An \defin{infinite path} is an infinite sequence of configurations
$v_0 v_1 \cdots$ such that for all $i\geq 0$, there is an edge 
in $\pgraph$ from $v_i$ to $v_{i+1}$. 
Note that we do not 
require $v_0$ to be the initial configuration.

We now define a generalisation of $n$-stacks called \emph{indexed $n$-stacks}. 
Following the same notations as in Section~\ref{ssection:regSets}, a stack $\stack$ is equivalently described as a pair
$\anglebra{\flaten{\stack},\target{\stack}}$ (recall that $\flaten{\stack}$ is a well-bracketed word description of $\stack$ and that $\target{\stack}$ gives the link structure). An \concept{indexed $n$-stack} is described by a triple $\anglebra{\flaten{\stack},\target{\stack},\ind{\stack}}$ where $\flaten{\stack}=\flaten{\stack}_1\cdots \flaten{\stack}_{|\flaten{\stack}|}$ and $\target{\stack}$ are as previously and where $\ind{\stack}:\{1,\dots,|\flaten{\stack}|\}\rightarrow \mathbb{N}$ is a partial function that is defined in any position $j< |\flaten{\stack}| - n$ such that $\flaten{\stack}_j\notin\{\lsk,\rsk\}$. The previous conditions on the domain of $\ind{\stack}$ ensure that any stack symbol in $\stack$ which is not the topmost one has a value by $\ind{\stack}$ that we refer to as its \concept{index}.
An \concept{indexed configuration} is a pair formed by a control state and an indexed stack. 

The {\concept{erasure}} of an indexed $n$-stack $\anglebra{\flaten{\stack},\target{\stack},\ind{\stack}}$ is the $n$-stack $\anglebra{\flaten{\stack},\target{\stack}}$. We extend the notion of erasure to indexed configurations in the obvious way.

The intended meaning of the index of some symbol in the stack is the
following.  The index is equal to the largest integer $i$ such that since
$v_i$ the symbol no longer appears as a $\topn{1}$-element. Hence, if one
uses the stack to store (and maintain) some information, the index is the
moment from which this information was no longer updated. Therefore when some
symbol appears again as the $\topn{1}$-element, one has to update the
information by taking into account all that happened since $v_i$
(included).

%\mh{above paragraph moved from below following def}

With any path $\play=v_0 v_1\cdots$, with $v_i=(p_i,\stack_i)$ for all $i\geq 0$, we inductively associate a
sequence of indexed configurations $\play'=v'_0 v'_1 \cdots$ such that the
following holds.
\begin{itemize}
\item The erasure of $\play'$ equals $\play$ (the \emph{erasure} of a sequence of indexed configurations being defined as the sequence of the respective erasures).
\item For any indexed configuration $v'_m=(q_m,\stack'_m)$ the following holds. Let $\stack'_m=\anglebra{\flaten{\stack'_m},\target{\stack'_m},\ind{\stack'_m}}$, let $\flaten{\stack'_m}=x_1\cdots x_h$, and let $j$ be in the domain of $\ind{\stack'_m}$ and such that $x_{j+1}=\rsk$. Then let $j'>j$ be the largest integer such that $x_{k}=\rsk$ for all $j+1\leq k\leq j'$ and let $i$ be the unique integer such that $x_i\cdots x_{j'}$ is well-bracketed. Then, for any $i<k<j'$, if $\ind{\stack'_m}(k)$ is defined, one has  $\ind{\stack'_m}(k)\leq \ind{\stack'_m}(j)$, and this inequality is strict if $\ind{\stack'_m}(j)\neq 0$. 
Intuitively, position $j$ is the topmost symbol of some $(j'-j)$-stack, and any symbol in this stack has an index smaller than the topmost symbol.
\end{itemize}

The intuitive idea behind the forthcoming definition of $\play'$ is rather
simple. The indices are always preserved, so one only cares about new positions
in the stack. On doing a $\pushn{k}$ the indices of the copied stack are
inherited from the original copy. Then when new indices are needed (because a
position is no longer the $\topn{1}$ one, it gets index $m+1$ if the current
configuration is $v_{m+1}$). 

Before going to the formal definition, we start with an example.

\begin{example}
In Figure~\ref{fig:exampleIndexedStacks}, we give an example (at order $3$) that illustrates the previous intuitive idea as well as the formal description below (ignore the information on colours for this example). We only describe the indexed stacked (omitting the control states), and indicate the stack operation (but omit the $\id$ operation). Indices are written as superscripts.

\begin{figure}[t]
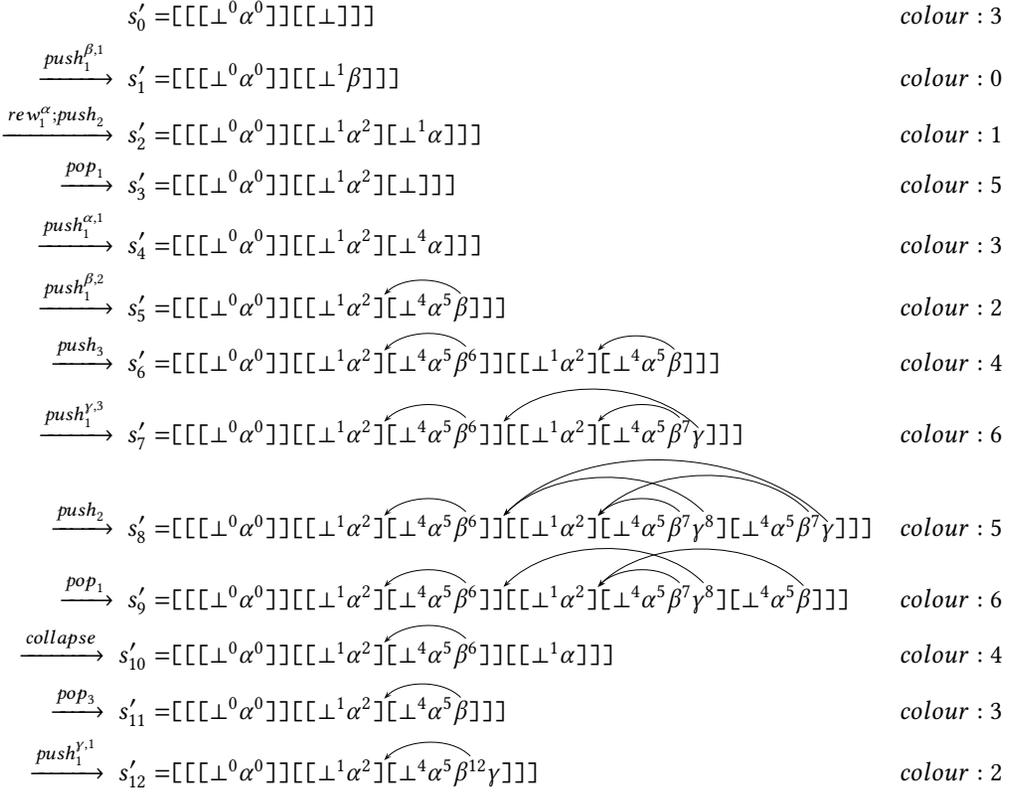

\begin{align*}
\stack_0'=&\mksk{\mksk{\mksk{\bot^0\alpha^0}}\mksk{\mksk{\bot}}}& colour: 3\\
\arrlab{\small\pushlk{1}{\beta}}\  \stack_1'=&\mksk{\mksk{\mksk{\bot^0\alpha^0}}\mksk{\mksk{\bot^1\beta}}}& colour: 0\\
\arrlab{\small \toprew{\alpha};\pushn{2}}\  \stack_2'=&\mksk{\mksk{\mksk{\bot^0\alpha^0}}\mksk{\mksk{\bot^1\alpha^2}\mksk{\bot^1\alpha}}}& colour: 1\\
\arrlab{\small \popn{1}}\  \stack_3'=&\mksk{\mksk{\mksk{\bot^0\alpha^0}}\mksk{\mksk{\bot^1\alpha^2}\mksk{\bot}}}& colour: 5\\
\arrlab{\small \pushlk{1}{\alpha}} \ \stack_4'=&\mksk{\mksk{\mksk{\bot^0\alpha^0}}\mksk{\mksk{\bot^1\alpha^2}\mksk{\bot^4\alpha}}}& colour: 3\\
\arrlab{\small \pushlk{2}{\beta}}\ \stack_5'=&\pstr[0.5cm]{\mksk{\mksk{\mksk{\bot^0\alpha^0}}\mksk{\mklksk{n1}{\bot^1\alpha^2}\mksk{\bot^4\alpha^5\nd(s1-n1){\beta}}}}}& colour: 2\\
\arrlab{\small \pushn{3}}\  \stack_6'=&\pstr[.5cm]{\mksk{\mksk{\mksk{\bot^0\alpha^0}}\mksk{\mklksk{n1}{\bot^1\alpha^2}\mksk{\bot^4\alpha^5\nd(s1-n1){\beta^6}}}\mksk{\mklksk{n11}{\bot^1\alpha^2}\mksk{\bot^4\alpha^5\nd(s11-n11){\beta}}}}}& colour: 4\\
\arrlab{\small \pushlk{3}{\gamma}} \ \stack_7'=&\pstr[0.7cm]{\mksk{\mksk{\mksk{\bot^0\alpha^0}}\mklksk{n3}{\mklksk{n1}{\bot^1\alpha^2}\mksk{\bot^4\alpha^5\nd(s1-n1){\beta^6}}}\mksk{\mklksk{n11}{\bot^1\alpha^2}\mksk{\bot^4\alpha^5\nd(s11-n11){\beta^7}\nd(s3-n3){\gamma}}}}}& colour: 6\\
\arrlab{\small \pushn{2}}\  \stack_8'=&\pstr[1cm]{\mksk{\mksk{\mksk{\bot^0\alpha^0}}\mklksk{n3}{\mklksk{n1}{\bot^1\alpha^2}\mksk{\bot^4\alpha^5\nd(s1-n1){\beta^6}}}\mksk{\mklksk{n11}{\bot^1\alpha^2}\mksk{\bot^4\alpha^5\nd(s11-n11){\beta^7}\nd(s3-n3){\gamma^8}}\mksk{\bot^4\alpha^5\nd(s112-n11){\beta^7}\nd(s31-n3){\gamma}}}}}& colour: 5\\
\arrlab{\small \popn{1}}\ \stack_9'=&\pstr[0.7cm]{\mksk{\mksk{\mksk{\bot^0\alpha^0}}\mklksk{n3}{\mklksk{n1}{\bot^1\alpha^2}\mksk{\bot^4\alpha^5\nd(s1-n1){\beta^6}}}\mksk{\mklksk{n11}{\bot^1\alpha^2}\mksk{\bot^4\alpha^5\nd(s11-n11){\beta^7}\nd(s3-n3){\gamma^8}}\mksk{\bot^4\alpha^5\nd(s112-n11){\beta}}}}}& colour: 6\\
\arrlab{\small \collapse}\ \stack_{10}'=&\pstr[0.5cm]{\mksk{\mksk{\mksk{\bot^0\alpha^0}}\mklksk{n3}{\mklksk{n1}{\bot^1\alpha^2}\mksk{\bot^4\alpha^5\nd(s1-n1){\beta^6}}}\mksk{\mklksk{n11}{\bot^1\alpha}}}}& colour: 4\\
\arrlab{\small \popn{3}}\  \stack_{11}'=&\pstr[.5cm]{\mksk{\mksk{\mksk{\bot^0\alpha^0}}\mklksk{n3}{\mklksk{n1}{\bot^1\alpha^2}\mksk{\bot^4\alpha^5\nd(s1-n1){\beta}}}}}& colour: 3\\
\arrlab{\small \pushlk{1}{\gamma}}\ \stack_{12}'=&\pstr[.5cm]{\mksk{\mksk{\mksk{\bot^0\alpha^0}}\mklksk{n3}{\mklksk{n1}{\bot^1\alpha^2}\mksk{\bot^4\alpha^5\nd(s1-n1){\beta^{12}}\gamma}}}} & colour: 2 \\
\end{align*}
\caption{Example of a sequence of indexed stacks.}\label{fig:exampleIndexedStacks}
\end{figure}
\end{example}

Now, we formally give the construction of $\play'$ (the previously mentioned properties easily follow from the definition). The initial configuration $v_0'=(p_0,\stack'_0)$, is obtained by letting $\ind{\stack'_0}$ be the constant (partial) function equal to $0$. Assume now that $v'_1\cdots v'_m$ has been
constructed, let $v'_{m}=(p_{m},\stack'_{m})$ with $\stack'_{m}=\anglebra{{\flaten{\stack}}_m,\target{\stack_m},\ind{\stack'_m}}$ 
and let $v_{m+1}=(p_{m+1},\stack_{m+1})$ with $\stack_{m+1}=\anglebra{{\flaten{\stack}}_{m+1},\target{\stack_{m+1}}}$. We let $v'_{m+1}=(p_{m+1},\stack'_{m+1})$ with $\stack'_{m+1}=\anglebra{{\flaten{\stack}}_{m+1},\target{\stack_{m+1}},\ind{\stack'_{m+1}}}$ where $\ind{\stack'_{m+1}}$ is defined thanks to the following  case distinction on which stack oprations have been applied to go from $v_m$ to $v_{m+1}$.
\begin{itemize}
\item A top-rewriting operation (possibly equal to $\id$) followed by a $\pushlk{k}{\gamma}$ operation is applied in configuration $v_m$. Then all previous indices are inherited and the former $\topn{1}$-element gets index $m+1$. Formally, $\ind{\stack'_{m+1}}(j) = \ind{\stack'_{m}}(j)$ whenever $j<|{\flaten{\stack}}_{m}|-n$ and $\ind{\stack'_{m+1}}(|{\flaten{\stack}}_{m}|-n) = m+1$.

\item A top-rewriting operation (possibly equal to $\id$) followed by a $\pushn{k}$ operation is applied.%: $\stack_{m+1}=\pushn{k}(\stack_{m})$. 
First, all existing indices are preserved, \ie $\ind{\stack'_{m+1}}(j) = \ind{\stack'_{m}}(j)$ whenever $j$ belongs to the domain of $\ind{\stack'_{m}}$.
Then one writes ${\flaten{\stack}}_{m}$ as $\lsk\cdots \lsk t \rsk \rsk^{n-k+1}$ with $t$ being well-bracketed; hence, 
${\flaten{\stack}}_{m+1}=\lsk\cdots \lsk t' \rsk\lsk t' \rsk \rsk^{n-k+1}$ where $t'$ is obtained from $t$ by (possibly )changing its last symbol to reflect the top-rewriting operation. Then we let $\ind{\stack'_{m+1}}(|\flaten{\stack'_m}|-(n-k+1)+j )= \ind{\stack'_{m}}(|\flaten{\stack'_m}|-(n-k+1)-(|t|+2)+j)$ for all $j\geq 1$ such that the second member of the equality is defined: the indices are simply copied from the former top $(k-1)$-stack. 
Finally, the former $\topn{1}$-element gets index $m+1$: $\ind{\stack'_{m+1}}(|{\flaten{\stack}}_{m}|-n+k-3) = m+1$.

\item A top-rewriting operation (possibly equal to $\id$) followed by either a $pop_{k}$ operation or a $collapse$ or $\id$ is applied in
  configuration $v_m$ in $\play$. Then all indices are inherited from the previous indexed stack. 
  Formally, $\ind{\stack'_{m+1}}(j) = \ind{\stack'_{m}}(j)$ whenever $j$ belongs to the domain of $\ind{\stack'_{m+1}}$.
\end{itemize}

The following straightforward proposition is crucial. In
particular, it means that if we stored some information on the stack,
the index gives the ``expiration date'' of the stored information, that is the
step in the computation starting from which the information has no
longer been updated.

\begin{proposition}\label{proposition:expirydate}
  Let $\Lambda=v_0v_1\cdots$ be a path and $\Lambda'=v_0'v_1'\cdots$ be as above. Let $m\geq 0$, let $\stack'_m=\anglebra{{\flaten{\stack}}_m,\target{\stack_m},\ind{\stack'_m}}$  be the indexed stack in $v'_m$. Let $j$ be such that $i=\ind{\stack'_m}(j)$ is defined. If $i>0$, then $(i-1)$ is the largest integer such that the $j$-th letter of ${\flaten{\stack}}_m$ is a copy of $\topn{1}(\stack_{i-1})$. If $i=0$, there is no $i'$ such that the $j$-th letter of ${\flaten{\stack}}_m$ is a copy of $\topn{1}(\stack_{i'})$.
\end{proposition}

\begin{proof}
Immediate by induction on $m$ and from the definition of $\play'$ from $\play$.
\end{proof}

  Consider a finite path $\play=v_0v_1\cdots v_m$ in $\pgame$ ending in a
  configuration $v_m=(q,\stack)$ such that $\topn{1}(\stack)$ has an $n$-link
  (if the link is a $k$-link for some $k<n$ the following concepts are not
  relevant).  The \concept{link-ancestor} of {$v_m$} is the configuration $v_j$
  where the original copy of the $n$-link in $\topn{1}(\stack)$ was created%
  \footnote{Formally, one could index links as well: whenever performing, in
  configuration $v_j$, a $\pushlk{e}{\gamma}$, one attaches to the newly
  created link the index $j+1$. Later, if the link is copied (by doing a
  $\pushn{k}$ operation) then the index is copied as well.}, or $v_0$ if the
  link was present in the stack of the configuration $v_0$. The
  \concept{link-rank} of {$v_m$} is the minimum colour of a state occurring in
  $\play$ since its link-ancestor $v_j$ (inclusive) \ie it is
  $\min\{\col(v_{j}),\cdots \col(v_{m})\}$.

 \begin{example}
Consider the sequence of indexed stacks given in Figure~\ref{fig:exampleIndexedStacks}.
The link-ancestor of configuration $v_8$ is configuration $v_7$ and its link-rank is $5$.
The link-ancestor of configuration $v_{11}$ is configuration $v_5$ and its link-rank is $2$.
\end{example}

%\os{Est-ce que l'on parle du pop/collapse-rank? En fait je ne garde que le link-rank car c'est lui qui sert}

\begin{definition}
  An $n$-CPDA $\pprocess=\anglebra{\Gamma,Q,\Delta,q_0}$ equipped with a colouring function is \concept{rank-aware} 
  from a configuration $v_0$ if there exists a function  
 % $\ColRk:\Gamma\rightarrow\mathbb{N}$ and 
  $\LinkRk:Q\times \Gamma\rightarrow \mathbb{N}$ such that
  for any finite path $\play = v_0 v_1 \cdots v_\ell$, 
   the link-rank (if defined) of the configuration
  $v_\ell=(q,\stack)$ is equal to $\LinkRk(q,\topn{1}(\stack))$. In other words, the link
  rank can be retrieved from the control state together with the $\topn{1}$-element of the stack.
\end{definition}

%\os{Il y a une petite ambiguité sur le fait que le concept n'a de sens qu'à partir d'une position initiale (et du coup le début du paragraphe suivant est un peu inexact}

To show that any CPDA can be transformed into a rank-aware CPDA, we
need to define the collapse-rank and the pop-rank.  First, we introduce the
notion of \emph{ancestor}.  Fix a finite path $\Lambda=v_0v_1\cdots v_m$, let
$v_m=(q,s)$ be some configuration in $\Lambda$ and let $x$ be a symbol in
$\stack$. Then the \concept{ancestor} of $x$ is the configuration $v_i$ where
$i$ is the index of $x$ in $v'_m$ (the indexed version of $v_m$).

We now {introduce the notion of} \emph{collapse-rank}. Fix a finite path
$\Lambda=v_0v_1\cdots v_m$ and assume that the $\topn{1}$-element of $v_m$ has
a $(k+1)$-link for some $k$.  Then the \concept{collapse-ancestor} in $v_m$ is
the ancestor of the $\topn{1}$-element of the target $k$-stack and the
\concept{collapse-rank} in $v_m$ is the smallest colour visited since the
collapse-ancestor (included). 

\begin{example}
Consider the sequence of indexed stacks given in Figure~\ref{fig:exampleIndexedStacks} (the colours of the corresponding configurations are indicated on the right part of the figure).

In $v'_8$ the collapse-ancestor is $v'_6$ and the collapse-rank is therefore $4$.
In $v'_9$ the collapse-ancestor is $v'_2$ and the collapse-rank is therefore $1$.
\end{example}

Next, we give a notion of \emph{pop-rank}. Fix a partial play
$\Lambda=v_0v_1\cdots v_m$ and a configuration $v_m=(q,\stack)$ in
$\Lambda$. Then, for any $1\leq k\leq n$, the \concept{pop-ancestor} for $k$, when
defined, is the ancestor of the $\topn{1}$-element of 
$\popn{k}(\stack)$ and the \concept{pop-rank} for $k$, when
defined, is the smallest colour visited since the pop-ancestor for $k$ (included).
In particular, the pop-rank for $n$ is the smallest colour
visited since the stack has height at least the height of $\stack$.

\begin{example}
Again, consider the sequence of indexed stacks given in Figure~\ref{fig:exampleIndexedStacks}.

In configuration $v'_9$ the pop-ancestor (\emph{resp.} pop-rank) for $3$ is $v'_6$ (\emph{resp.} $4$), the pop-ancestor (\emph{resp.} pop-rank) for $2$ is $v'_8$ (\emph{resp.} $5$) and the pop-ancestor (\emph{resp.} pop-rank) for $1$ is $v'_5$ (\emph{resp.} $2$).

In configuration $v'_{12}$ the pop-ancestor (\emph{resp.} pop-rank) for $3$ is $v'_0$ (\emph{resp.} $0$), the pop-ancestor (\emph{resp.} pop-rank) for $2$ is $v'_2$ (\emph{resp.} $1$) and the pop-ancestor (\emph{resp.} pop-rank) for $1$ is $v'_12$ (\emph{resp.} $2$).
\end{example}

\begin{remark}\label{rk:undefinedancestor}
In the current setting, if the ancestor of the pointed stack (\emph{resp} the ancestor of the $\topn{1}$-element of $pop_k(\stack)$ / the link-ancestor) is $v_0$, then the collapse-rank (\emph{resp} the pop-rank / the link-rank) is simply the smallest colour seen since the beginning of the play. Hence, it does not make much sense but it permits the construction to remain uniform.
\end{remark}

\subsection{Main Result}

The next theorem shows that we can restrict our attention to CPDA games where the underlying CPDA is rank-aware.

\begin{theorem}\label{lemma:rank-aware}
For any $n$-CPDA $\pprocess=\langle \Gamma,Q,\Delta,q_0\rangle$ and any associated parity game $\pgame$, one can construct an
 $n$-CPDA $\pprocessrk$ and an associated parity game $\pgamerk$ such that the following holds.
 \begin{itemize} 
 \item There exists a mapping $\nu$ from the configurations of $\pprocess$ to that of $\pprocessrk$ such that:
	\begin{itemize}
		\item for any configuration $v_0$ of $\pprocess$, $\pprocessrk$ is rank-aware from $\nu(v_0)$;
		\item \Eloise has a winning strategy in $\pgame$ from a configuration $v_0$ iff she has a winning 	strategy in $\pgamerk$ from $\nu(v_0)$;
		\item both $\nu$ and $\nu^{-1}$ preserve regularity of sets of configurations.
	\end{itemize}
\item If there is an $n$-CPDA transducer $\Srk$ synchronised with $\pprocessrk$ realising a well-defined winning strategy for \Eloise in $\pgamerk$ from $\nu(q_0,\bot_{n})$, then one can effectively construct an $n$-CPDA transducer $\S$ synchronised with $\pprocess$ realising a well-defined winning strategy for \Eloise in $\pgame$ from the initial configuration $(q_0,\bot_{n})$.
\end{itemize}
\end{theorem}

\subsection{Proof of Theorem~\ref{lemma:rank-aware}}

The proof of Theorem~\ref{lemma:rank-aware} is a non-trivial generalisation of \cite[Lemma 6.3]{KNUW05} (which concerns 2-CPDA) to the general setting of $n$-CPDA {and starting from an arbitrary configuration}. 

%\os{Retravailler le paragraphe qui suit. En particulier pour les etats.}

%We turn now to the technical core of the proof.  
Fix an $n$-CPDA
$\pprocess=\anglebra{\Gamma,Q,\Delta,q_0}$, a partition $Q_\Ei\uplus
Q_\Ai$ of $Q$ and a colouring function
$\col:Q\rightarrow\colors\subset\mathbb{N}$. Denote by $\pgame$ the
induced parity game. 
We define a rank-aware (to be proven) $n$-CPDA
$\pprocessrk=\anglebra{\Gammark,\Qrk,\Deltark,q_{0,\rk}}$ such that
$\Qrk=Q\times \colors$ and 
{$$\Gammark={\Gamma\times (\colors\cup\{\indef\})\times (\colors\cup\{\indef,\nalr\}) \times (\colors^{\{1,\dots,n\}}}\cup\{\indef\})$$}

%The \emph{main} configurations $((q,\mcolsofar),s)$ of
%$\pprocessrk$ (by \emph{main} we mean a configuration that is reached after simulating a transition of $\pprocess$: indeed simulating one transition of $\pprocess$ need several steps in $\pprocessrk$, hence goes through intermediate configurations that we do not care about when stating our invariant) 

We define a map $\nu$ that associates with any configuration of $\pprocess$ a configuration of $\pprocessrk$.  Let $(q,\stack)$ be a configuration in $\pprocess$. 
Then $\nu(q,\stack)=((q,\col(q)),\stack')$ where $\stack'$ is obtained by: 
\begin{itemize}
\item Replacing every internal (\ie that is not the $\topn{1}$-element) symbol $\gamma$ by $(\gamma,\indef,\indef,\indef)$ if it has an $n$-link and by $(\gamma,\indef,\nalr,\indef)$ otherwise.
\item Replacing the $\topn{1}$-element $\gamma$ by $(\gamma,\col(q),\col(q),\tau_{\col(q)}$ if it has an $n$-link and otherwise by $(\gamma,\col(q),\nalr,\tau_{\col(q)})$, where $\tau_{\col(q)}$ is the constant function assigning to any $1\leq i\leq n$ the value $\col(q)$.
\end{itemize}

%\os{One should explain the colouring function for $\pprocessrk$}
We equip $\pprocessrk$ with a colouring function $\colrk$ by letting $\colrk(q,\mcolsofar)=\col(q)$.
Our construction will satisfy the following invariant. Let $\Lambda$ be a finite path in $\transgraph{\pprocessrk}$ starting in some configuration $\nu(q,\stack)$ ending in some configuration $((p,\mcolsofar),\stack)$ then the following holds.
 First, $\mcolsofar$ is the minimal colour visited from the beginning of the path. Second, if $\topn{1}(\stack)=(\alpha,m_c,m_l, \tau)$ then
\begin{itemize}
\item $m_c$ is the collapse-rank;
\item $m_l$ is the link-rank if it makes sense (\ie if there is an $n$-link in the current $\topn{1}$-symbol) or is $\nalr$ otherwise;
\item $\tau$ is the \emph{pop-rank}: $\tau(i)$ is the pop-rank for $i$ for every $1\leq i\leq n$.
\end{itemize}

%Let us now explain how $\nu$ is defined. Let $(q,s)$ be some configuration in $\pprocess$. 
%Then $\nu(q,s)=((q,\col(q),s')$ where $s'$ is obtained by: 
%\begin{itemize}
%\item Replacing every internal symbol $\gamma$ (\ie that is not the $\topn{1}$-element) by $(\gamma,\indef,\indef,\indef)$ if it has an $n$-link and by $(\gamma,\indef,\nalr,\indef)$ otherwise.
%\item Replacing the $\topn{1}$-element $\gamma$ by $(\gamma,\col(q),\col(q),\col(q))$ if it has an $n$-link and otherwise by $(\gamma,\col(q),\nalr,\col(q))$.
%\end{itemize}

Trivially, from the definition of $\nu$, the invariant holds at the beginning of the path. 

The transition function of $\pprocessrk$ mimics that of $\pprocess$ and updates the ranks as explained below. First, let us explain the meaning of symbols $\indef$. Such symbols will never been created using a $\pushlk{k}{\_}$ or a $\toprew{\indef}$ action: hence they can only be duplicated (using $\pushn{k}$) from symbols originally in the stack. The meaning of a symbol $\indef$ is that the corresponding object (collapse-rank, link-rank or pop-rank) has not yet been settled. However, when a $\indef$ symbol appears in the $\topn{1}$-element the various ranks can be easily retrieved as they necessarily equal the smallest colour visited so far (as noted in Remark~\ref{rk:undefinedancestor}): this is why we will compute the minimal colour visited so far in the control state of $\pprocessrk$.

In order to make the construction more readable, we do not formally describe $\Deltark$ but rather explain how $\pprocessrk$ behaves. It should be clear that $\Deltark$ can be formally described to fit this informal description (and that some extra control states are {actually} needed as we will allow to do several stack operation per transition); technical issues about this construction are discussed in Remark~\ref{rk:rankAwareTopRew}. Note that the description below also contains the inductive proof of its validity, namely that $m_c$, $m_l$ and $\tau$ are as stated above. To avoid case distinction on whether the link-rank is defined or not, we take the following convention that $\min(\nalr,i)=\nalr$ for every $i\in \mathbb{N}$.

The intuitive idea is the following. One stores in the stack information on the various ranks, and after performing a $\popn{k}$ or a $\collapse$, one needs to update the information stored in the new $\topn{1}$-element. Indeed this information has no longer been updated since the ancestor configuration (this was the last time it was on top of the stack). To update it, one uses either the collapse-rank / pop-rank in the previous configuration, which is exactly what is needed for this update.

Assume $\pprocessrk$ is in configuration $v_\ell=((q,\mcolsofar),\stack)$ with $\topn{1}(\stack)=({\alpha},m_c,m_l,\tau)$ and let $v_0v_1\cdots v_\ell$ be the beginning of the path of $\transgraph{\pprocessrk}$ where we denote $v_i=((q_i,\mcolsofar_i),\stack_i)$ (hence $q_\ell=q$ and $\stack_\ell=\stack$). 
%The following behaviours are those allowed in $((q,\mcolsofar),\stack)$.
For any $(q',\toprew{\gamma};op)\in \Delta(q,\alpha)$ (note that the case where no $\toprew{}$ is performed corresponds to the case where $\gamma=\alpha$) the following behaviours are those allowed in $((q,\mcolsofar),\stack)$.

\begin{enumerate}
%%%%%%%%%
%%Pop_n%%%%
%%%%%%%%%%

\item Assume $op=\popn{k}$ for some $1\leq k\leq n$, let $\popn{k}(\stack)=\stack'$ and let $\topn{1}(\stack')=(\alpha',m_c',m_l',\tau')$. Then $\pprocessrk$ can go to the configuration $((q',\mcolsofar'),\stack'')$ where $\mcolsofar'=\min(\mcolsofar,\col(q'))$ and $\stack''$ is obtained from $\stack'$ by replacing $\topn{1}(\stack')$ by 
\begin{enumerate}
\item $(\alpha',\mcolsofar',\mcolsofar',(\mcolsofar',\ldots,\mcolsofar'))$ if $m_c'=\indef$, $m_l'=\indef$ and $\tau'\indef$;
\item $(\alpha',\mcolsofar',\nalr,(\mcolsofar',\ldots,\mcolsofar'))$ if $m_c'=\indef$, $m_l'=\nalr$ and $\tau'\indef$; 
\item
$(\alpha',\min(m_c',\tau(k),\col(q')),\min(m_l',\tau(k),\col(q')),\tau'')$ otherwise, with 
$$\tau''(i)=
\begin{cases}
        \min(\tau'(i),\tau(k),\col(q')) & \text{ if } i\leq k\\
        \min(\tau(i),\col(q')) &\text{ if } i> k.
\end{cases}$$
\end{enumerate}

Cases $(a)$ and $(b)$ correspond to the case where one reaches (possibly a copy) of a symbol that was in the stack from the very beginning and that never appeared as a $\topn{1}$-element: then the value of the collapse-rank, link-rank —~if defined this is case $(a)$ otherwise it is case $(b)$~— and pop-ranks are all equal to $\mcolsofar'$.

We now explain case $(c)$. 
Let $v_x$ be the ancestor of $\topn{1}(\popn{k}(\stack))$. Then $x>0$ as otherwise we would be in case $(a)$ or $(b)$. By Proposition~\ref{proposition:expirydate}, it follows that $\topn{1}(\popn{k}(\stack))=\topn{1}(\stack_{x-1})$, and by induction hypothesis, at step $(x-1)$, $m_c'$, $m_l'$ and $\tau'$ had the expected meaning. Let $y$ be the index of the $\topn{1}$-element of the pointed stack in $\stack'$: $y$ is also the $\topn{1}$-element of the pointed stack in $\stack_{x-1}$, and moreover $y<x$. Hence, the collapse-rank in $v_{\ell+1}$ is 
\begin{align*}
& \min\{\col(q_{y}),\ldots,\col(q_{x-1}),\col(q_x),\ldots,\col(q_\ell),\col(q')\}\\
= & \min\{\min\{\col(q_{y}),\ldots,\col(q_{x-1})\},\min\{\col(q_x),\ldots,\col(q_\ell)\},\col(q')\}\\
= & \min\{m'_c,\tau(k),\col(q')\}
\end{align*}
Similarly, when defined, the link-ancestor of $\stack'$ is the same as the one in $\stack_{x-1}$: hence the pop-rank in $v_{\ell+1}$ is $\min\{m'_l,\tau(k),\col(q')\}$.

For any $i\leq k$, $\topn{1}(\popn{i}(\stack'))=\topn{1}(\stack_{x-1})$ and therefore the pop-rank for $i$ in $v_{\ell+1}$ is obtained by updating $\tau'(i)$ to take care of the minimum colour seen since $v_x$ which, as for the collapse-rank, is $\min\{\tau(k),\col(q')\}$: therefore the pop-rank for $i$ in $v_{\ell+1}$ equals $\min\{\tau'(i),\tau(k),\col(q')\}$.

For any $i>k$, $\popn{i}(\stack')=\popn{i}(\stack)$ and thus $\topn{1}(\popn{i}(\stack'))=\topn{1}(\popn{i}(\stack))$. Therefore the pop-rank for $i$ in $v_{\ell+1}$ is obtained by updating the one in $v_{\ell}$ to take care of the new visited colour $\col(q')$: hence the pop-rank for $i$ in $v_{\ell+1}$ equals $\min\{\tau(i),\col(q')\}$.

%%%%%%%%%%%
%%Collapse %%%%
%%%%%%%%%%%
\item Assume $op=collapse$, let $k$ be the order of the link in $\topn{1}(\stack)$, let $collapse(\stack)=\stack'$ and let $\topn{1}(\stack')=(\alpha',m_c',m_l',\tau')$. Then $\pprocessrk$ can go to the configuration $((q',\mcolsofar'),\stack'')$ where $\mcolsofar'=\min(\mcolsofar,\col(q'))$ and $\stack''$ is obtained from $\stack'$ by replacing $\topn{1}(\stack')$ by 
\begin{enumerate}
\item $(\alpha',\mcolsofar',\mcolsofar',(\mcolsofar',\ldots,\mcolsofar'))$ if $m_c'=\indef$, $m_l'=\indef$ and $\tau'=\indef$;
\item $(\alpha',\mcolsofar',\nalr,(\mcolsofar',\ldots,\mcolsofar'))$ if $m_c'=\indef$, $m_l'=\nalr$ and $\tau'=\indef$;
\item $(\alpha',\min(m_c',m_c,\col(q')),\min(m_l',m_c,\col(q')),\tau'')$ otherwise with 
$$\tau''(i)=
\begin{cases}
\min(\tau'(i),m_c,\col(q')) & \text{if }i\leq k\\ 
\min(\tau(i),\col(q')) & \text{if }i> k.\\ 
\end{cases}
$$ 
\end{enumerate}

The proof follows the same line as for the previous case. 
Cases $(a)$ and $(b)$ correspond to the case where one reaches (possibly a copy) of a symbol that was in the stack from the very beginning and that never appeared as a $\topn{1}$-element: then the value of the collapse-rank, link-rank —~if defined this is case $(a)$ otherwise it is case $(b)$~— and pop-ranks are all equal to $\mcolsofar'$.

We now explain case $(c)$. 
Let $v_x$ be the collapse-ancestor of $v_\ell$. Then $x > 0$ as otherwise we would be in case $(a)$ or $(b)$. By induction hypothesis, $m'_c$, $m'_l$ and $\tau'$ give the collapse-rank / link-rank / pop-ranks in $v_{x-1}$. 
Moreover the ancestor of the $\topn{1}$-element of the target of the top link in $\stack'$ is the same as the one in $v_{x-1}$. Therefore, the collapse-rank is obtained by taking the minimum of the collapse-rank in $v_{x-1}$ with $\min\{\col(q_x),\ldots \col(q_\ell),\col(q')\}=\min\{m_c,\col(q')\}$. Similarly (if defined) the link-ancestor in $\stack'$ being the same as the one in $v_{x-1}$, the link-rank is obtained by taking the minimum of the one in $v_{x-1}$ with $\min\{\col(q_x),\ldots ,\col(q_\ell),\col(q')\}=\min\{m_c,\col(q')\}$.

Let $i\leq k$. The ancestor of $\topn{1}(\popn{i}(\stack'))$ is the same as the ancestor of $\topn{1}(\popn{i}(\stack_{x-1}))$. Therefore the pop-rank for $i$ in $v_{\ell+1}$ is obtained by taking the minimum of the one in $v_{x-1}$ with $\min\{\col(q_x),\ldots \col(q_\ell),\col(q')\}=\min\{m_c,\col(q')\}$.

Let $i> k$. Then the ancestor of $\topn{1}(\popn{i}(\stack'))$ is the same as the ancestor of $\topn{1}(\popn{i}(\stack_\ell))$: indeed the collapse only modified the $\topn{k}$ stack, in other words $\popn{i}(\collapse(\stack))=\popn{i}(\stack)$. Therefore the pop-rank for $i$ in $v_{\ell+1}$ is obtained by taking the minimum of the one in $v_{\ell}$ with the new visited colour $\col(q')$.

%%%%%%%%%
%%Push_j%%%
%%%%%%%%%

\item Assume $op=\pushn{j}$ for some $2\leq j\leq n$, let $\pushn{j}(\toprew{(\gamma,m_c,m_l,\tau)}(\stack))=\stack'$ and let $\topn{1}(\stack')=(\gamma,m_c,m_l,\tau)$ (note that $\indef$ does not appear in $\topn{1}(\stack')$). Then, $\pprocessrk$ can go to the configuration $((q',\mcolsofar'),\stack'')$ where $\mcolsofar'=\min(\mcolsofar,\col(q'))$ and $\stack''$ is obtained from $\stack'$ when replacing $\topn{1}(\stack')$ by $(\gamma,\min(m_c,\col(q')),\min(m_l,\col(q')),\tau')$ with 
$$\tau'(i)=
\begin{cases}
        \min(\tau(i),\col(q')) & \text{if } i\neq j\\
        \col(q')  & \text{if } i= j
\end{cases}$$

Indeed, the collapse-ancestor in the new configuration is the same as the one in $\stack$. As by induction hypothesis $m_c$ is the collapse-rank in $v_{\ell}$, the collapse-rank in $v_{\ell+1}$ is obtained by updating $m_c$ to take care of the new visited colour, namely by taking $\min\{m_c,\col(q')\}$. Similarly, if defined, the link-ancestors in $v_{\ell}$ and $v_{\ell+1}$ are identical and then the link-rank in $v_{\ell+1}$ is $\min\{m_c,\col(q')\}$.

For any $i\neq j$, the ancestor of $\topn{1}(\popn{i}(\stack)')$ and the ancestor of $\topn{1}(\popn{i}(\stack'))$ are the same. Again using the induction hypothesis one directly gets that the pop-rank for $i$ in $v_{\ell+1}$ equals $\min\{\tau(i),\col(q')\}$.

The index of the ancestor of $top_1(pop_j(\stack'))$ is by definition $\ell+1$. Hence, as the only colour visited since $v_{\ell+1}$ is $\col(q')$ it equals the pop-rank for $j$.

%%%%%%%%%%%%%
%%Push^{\beta,k}%%%
%%%%%%%%%%%%%

\item  Assume $op=\pushlk{k}{\beta}$ with $1\leq k\leq n$, and $\beta\in(\Gamma\setminus\{\bot\})$. Then $\pprocessrk$ can go to $(q',\mcolsofar')$, where $\mcolsofar'=\min(\mcolsofar,\col'(q'))$, and apply successively $\toprew{(\gamma,m_c,m_l,\tau)}$ and $\pushlk{k}{(\beta,m'_c,m'_l,\tau')}$ where $m'_c=\min(\tau(k),\col(q'))$, $m'_l=\col(q')$ if $k=n$ and $m'_l=\nalr$ otherwise, and $\tau'(i)=\min(\tau(i),\col(q'))$ for every $i\geq 2$ and $\tau(1)=\col(q')$.

Indeed, the pointed stack in $\stack'$ is $\topn{k}(\popn{k}(\stack))$ and therefore the collapse-rank in $v_{\ell+1}$ is the minimum of the pop-rank for $k$ in $\stack$ and of the new visited colour $\col(q')$, that is $\min\{\tau(k),\col(q')\}$.

If $k=n$, the link-ancestor of $v_{\ell+1}$ is $v_{\ell+1}$ itself and hence the link-rank is the colour of the current configuration, namely $\col(q')$.

For any $i\geq 2$, as $\popn{i}(\stack)=\popn{i}(\stack')$ one also has that $\topn{1}(\popn{i}(\stack'))=\topn{1}(\popn{i}(\stack))$ and therefore the pop-rank for $i$ in $v_{\ell+1}$ equals the minimum of the one in $v_\ell$ with the new visited colour $\col(q')$, that is $\min\{\tau(i),\col(q')\}$. Finally as the ancestor of $\popn{1}(\stack')$ is $v_{\ell+1}$ then the pop-rank for $1$ is the current colour, namely $\col(q')$.

\end{enumerate} 

From the previous description (and the included inductive proof) we conclude that, for any configuration $v_0$ of $\pprocess$, $\pprocessrk$ is rank-aware from $\nu(v_0)$, where we let $\LinkRk((q,(\gamma,m_c,m_l,\tau)))=m_l$.

\begin{remark}\label{rk:rankAwareTopRew}
One may object that $\pprocessrk$ does not fit the definition of $n$-CPDA. Indeed, in a single transition it can do a top-rewriting followed by another stack operation and followed again by a top-rewriting (which itself depends on the new $\topn{1}$-element). One could add intermediate states and simply decompose such a transition into two transitions, but this would be problematic later when defining an $n$-CPDA transducer realising a winning strategy. 

Fortunately, one can define a variant $\pprocessrk'$ of $\pprocessrk$ that has the same properties as $\pprocessrk$ and additionally fits the definition of $n$-CPDA. The idea is simply to postpone the final top-rewriting to the next transition. Indeed, it suffices to add a new component on the control state where one encodes the top-rewriting that should be performed next: this top-rewriting is then performed in the next transition (note that this fits the definition as performing two top-rewriting is the same as doing only the last one). However, there is still an issue as the top-rewriting was actually depending on the $\topn{1}$-symbol (one updates the various ranks) hence, one cannot save the next top-rewriting in the control state without first observing the symbol to be rewritten. Again this is not a real problem, as it suffices to remember which kind of update should be done (one concerning a $\popn{k}$ or one concerning a $\collapse$) and to store in the control state the various objects needed for this update (for this, one can simply store the former $\topn{1}$-element).

One also needs to slightly modify the $\LinkRk$ function so that it returns the link-rank of the $\topn{1}$-symbol after it is rewritten. This can easily be done as the domain of $\LinkRk$ is $\Qrk\times \Gammark$.

Note that $\pprocessrk'$ and $\pprocessrk$ use the same stack alphabet, but that the state space of $\pprocessrk'$ uses an extra component of size linear in the one of the stack alphabet. 

In conclusion building a rank-aware (valid) $n$-CPDA from a non-aware one increases (by a multiplicative factor) the stack alphabet by $|\colors|^{n+3}$ and the state set by $\mathcal{O}(|\colors|^{n+3})$.

For now on, we uses $\pprocessrk$ to mean $\pprocessrk'$.
\end{remark}

We are now ready to conclude the proof of Theorem~\ref{lemma:rank-aware}. First recall that we defined $\colrk$ by letting $\colrk(q,\mcolsofar)=\col(q)$. Then, we define a partition $\QrkE\uplus\QrkA$ of $\Qrk$ by letting the states in $\QrkE$ be those states with their first component in $Q_\Ei$, and those states in $\QrkA$ be those states with their first component in $Q_\Ai$. Let $\pggraphrk$ be the corresponding arena and let $\pgamerk=(\pggraphrk,\WC_{\colrk})$ be the corresponding $n$-CPDA parity game.

%It is straightforward that \Eloise has a winning strategy in $\pgame$ from some configuration $v_0$ iff she has a winning strategy in $\pgamerk$ from $\nu(v_0)$. Indeed, consider some configuration $v_0$ of $\pprocess$. 

Consider the projection $\zeta$ defined from configurations of $\pprocessrk$ into configurations of $\pprocess$ by only keeping the first component of the control state, and by only keeping the $\Gamma$ part of the symbols in the stack. Note that, on the domain of $\nu^{-1}$, $\zeta$ and $\nu^{-1}$ coincide. Also note that $\zeta$ preserves the shape of stacks\footnote{
Recall that the \emph{shape} of a stack is the stack obtained by replacing all non-$\bot$ symbols appearing in $\stack$ by a fresh dummy symbol $\sharp$ (but keeping the links).
}, \ie for any configuration $v_\rk$, the stack in $v_\rk$ has the same shape as the stack in $\nu(v_\rk)$. 

We extend $\zeta$ as a function from (possibly partial) plays in $\pgamerk$ into (possibly partial) plays in $\pgame$ by letting $\zeta(v_0'v_1'\cdots)=\zeta(v_0')\zeta(v_1')\cdots$. It is obvious that for any play $\play'$ in $\pgamerk$ starting from $\nu(v_0)$, its image $\zeta(\play')$ is a play in $\pgame$ starting from $v_0$; moreover these two plays induce the same sequence of colours and at any round the player that controls the current configuration is the same in both plays. 
Conversely, from the definition of $\pprocessrk$ it is also clear that there is, for any play $\play$ in $\pgame$ starting from $v_0$, a \emph{unique} play $\play'$ in $\pgamerk$ starting from $\nu(v_0)$ such that $\zeta(\play')=\play$.

In particular, $\zeta$ can be used to construct a strategy in $\pgame$ from a strategy in $\pgamerk$. Indeed, let $\strat_\rk$ be a strategy for \Eloise from $\nu(v_0)$ in $\pgamerk$. We define a strategy $\strat$ in $\pgame$ from $\nu(v_0)$. This strategy maintains as a memory a partial play $\play_\rk$ in $\pgamerk$ such that, if \Eloise respects $\strat$, in $\pgame$ starting from $v_0$ after having played $\play$ one has $\zeta(\play_\rk)=\play$ and moreover $\play_\rk$ is a play in $\pgamerk$ starting from $\nu(v_0)$ where \Eloise respects $\strat_\rk$. Initially, we let $\play_\rk=\nu(v_0)$. Assume that we have been playing $\play$ and that \Eloise has to play next. Then she considers $v_\rk=\strat_\rk(\play_\rk)$ and she plays to $v$ where $v=\zeta(v_\rk)$. %is the unique configuration such that $\zeta(\play_\rk\cdot v_\rk)=\play\cdot v$. 
Finally one updates $\play_\rk$ to be $\play_\rk\cdot v_\rk$. If it is \Abelard that has to play next and if he moves to some $v$, then \Eloise updates $\play_\rk$ to be $\play_\rk\cdot v_\rk$ where $v_\rk$ is the unique configuration such that $\play_\rk\cdot v_\rk$ is a valid play and such that $\zeta(v_\rk)=v$. A similar construction can be done to build a strategy of \Abelard in $\pgame$ from one in $\pgamerk$.

Now, assume that $\nu(v_0)$ is winning for \Eloise (\emph{resp.} \Abelard) and call $\strat_\rk$ an associated winning strategy. Let $\strat$ be the strategy in $\pgame$ obtained as explained above. Then $\strat$ is winning for \Eloise (\emph{resp.} \Abelard) in $\pgame$ from $v_0$ (this follows directly from the fact that $\strat_\rk$ is winning and that we have the property that $\zeta(\play_\rk)=\play$ for any partial play $\play$ in $\pgame$ consistent with $\strat$). Hence this proves that \Eloise has a winning strategy in $\pgame$ from $v_0$ iff she has a winning strategy in $\pgamerk$ from $\nu(v_0)$.

The fact that both $\nu$ and $\nu^{-1}$ preserve regular sets of configurations is obvious: for this one basically needs to simulate an automaton on the image by $\nu$ (or $\nu^{-1}$) that can be computed on-the-fly (except for the very last steps of $\nu$ where one needs to know the control state before deducing the $\topn{1}$ stack element as it has information on the colour of the control state. However, this is not a problem to have a slight --- finite --- delay in the final steps of the simulation).

Finally, from the previous construction of a strategy $\strat$ from a strategy $\strat_\rk$ we prove that if there is an $n$-CPDA transducer $\Srk$ synchronised with $\pprocessrk$ realising a well-defined winning strategy $\strat_\rk$ for \Eloise in $\pgamerk$ from $\nu(q_0,\bot_{n})$, then one can effectively construct an $n$-CPDA transducer $\S$ synchronised with $\pprocess$ realising a well-defined winning strategy $\strat$ for \Eloise in $\pgame$ from the initial configuration $(q_0,\bot_{n})$. Indeed, in our previous construction of $\strat$, we maintained a partial play $\play_\rk$ in $\pgamerk$ and used the value of $\strat_\rk(\play_\rk)$ to define $\strat(\play)$. But if $\strat_\rk$ is realised by an $n$-CPDA transducer $\Srk$, it suffices to remember the configuration of this transducer after playing $\play_\rk$ (as this suffices to compute $\phi_\rk(\play_\rk))$. Hence, the only things that need to be modified from $\Srk$ to obtain $\S$ is that one needs to ``embed'' the transition function of $\pprocess_\rk$ into it, so that $\S$ can read/output elements in $Q\times \Op{n}{\Gamma}\times \Op{n}{\Gamma}$ instead of $\Qrk\times \Op{n}{\Gammark}\times \Op{n}{\Gammark}$. This can easily (but writing the formal construction would be quite heavy) be achieved by noting that the shape of stacks is preserved by $\zeta$: hence if $\Srk$ is synchronised with $\pprocessrk$ then $\S$ is synchronised with $\pprocess$ (as $\pprocessrk$ and $\pprocess$ are ``synchronised'', and $\Srk$ and $\S$ are ``synchronised'' as well).

%\os{Should I give more details below? I don't think so but I let you tell me}

\subsection{Complexity}

If we summarise, the overall blowup in the transformation from $\pgame$ to $\pgamerk$ given by Theorem~\ref{lemma:rank-aware} is as follows.

\begin{proposition}\label{proposition:complexity-step1}
Let $\pprocess$ and $\pprocessrk$ be as in Theorem~\ref{lemma:rank-aware}. Then the set of states of $\pprocessrk$ has size $\mathcal{O}(|Q|(|\colors|+1)^{n+3})$ and the stack alphabet of $\pprocessrk$ has size $\mathcal{O}(|\Gamma|(|\colors|+1)^{2n+5})$.
Moreover the set of colours used in $\pgame$ and $\pgamerk$ are the same.
\end{proposition}

\begin{proof}
By construction together with Remark~\ref{rk:rankAwareTopRew}.
\end{proof}

%% file: Outmost.tex
% !TEX root = main.tex

\section{Removing the $n$-links}\label{section:outermostLinks}

\subsection{Main Result}

In this section, we show how one can remove the outmost (\ie order-$n$) links.  {In the following $\lf$ intended to mean \emph{link-free}.}  

\begin{theorem}\label{theo:outmost}
For any \emph{rank-aware} $n$-CPDA $\pprocessrk=\anglebra{ \Gammark,\Qrk,\Deltark,q_{0,\rk}}$ and any associated parity game $\pgamerk$, one can construct an
 $n$-CPDA $\pprocesslf$ and an associated parity game $\pgamelf$ such that the following holds.
 \begin{itemize} 
 \item $\pprocesslf$ does not create $n$-links.
 \item There exists a mapping $\nu$ from the configurations of $\pprocessrk$ to that of $\pprocesslf$ such that:
	\begin{itemize}
		\item \Eloise has a winning strategy in $\pgamerk$ from a configuration $v_0$ iff she has a winning strategy in $\pgamelf$ from $\nu(v_0)$;
		\item If the set of winning configurations for \Eloise in $\pgamelf$ is regular, then the set of winning configurations for \Eloise in $\pgamerk$ is regular as well.
	\end{itemize}
\item If there is an $n$-CPDA transducer $\Slf$ synchronised with $\pprocesslf$ realising a well-defined winning strategy for \Eloise in $\pgamelf$ from $\nu(q_{0,\rk},\bot_{n})$, then one can effectively construct an $n$-CPDA transducer $\Srk$ synchronised with $\pprocessrk$ realising a well-defined winning strategy for \Eloise in $\pgamerk$ from the initial configuration $(q_{0,\rk},\bot_{n})$.
\end{itemize}
\end{theorem}

The whole section is devoted to the proof of Theorem \ref{theo:outmost} and we thus fix from now on, a \emph{rank-aware} $n$-CPDA $\Ark=\anglebra{\Gammark,\Qrk,\Deltark,q_{0,\rk}}$ (together with a function $\LinkRk$), a 
partition $Q_{\rk,\Ei}\uplus Q_{\rk,\Ai}$ of $\Qrk$, a colouring function
$\col:\Qrk\rightarrow\colors\subset\mathbb{N}$ and we let $\colors=\{0,\ldots,\maxcolor\}$. Denote by $\pgraphrk$ the transition graph of $\pprocessrk$, by $\pggraphrk$ the arena induced by $\pgraphrk$ and the partition $Q_{\rk,\Ei}\uplus Q_{\rk,\Ai}$, and by $\pgamerk$ the parity game $(\pggraphrk,\WC_\col)$.

There are now two tasks. The first one is to prove that the previous simulation game can be generated by an $n$-CPDA with the extra property that it never creates $n$-links. The second one is to prove that this game correctly simulates the original one (\ie \Eloise wins in $\pgamerk$ from some vertex $v$ iff she wins in the $\pgamelf$ from the configuration $\nu(v)$ for some mapping $\nu$ --- to be defined --- transforming vertices of the first game into vertices of the second one). The first task (see Section \ref{section:Step2Task1}) is simple as the initial $n$-CPDA defining $\pgamerk$ is rank aware and therefore comes with a function $\LinkRk$ as in Lemma \ref{lemma:rank-aware}. The second task (see Section \ref{section:Step2Task2}) is more involved because we have to define $\nu$ and to prove that it preserves (arbitrary) winning configurations.

\subsection{The Simulation Game: $\pgamelf$}\label{section:Step2Task1}

We now define $\pprocesslf$ and the associated game $\pgamelf$. We start with an informal description of $\pprocesslf$ and then formally describe its structure.

The $n$-CPDA $\pprocesslf$ \emph{simulates} $\pprocessrk$ as follows. Assume that the play is in some configuration $(q,\stack)$ and that the player that controls it wants to simulate a transition $(p,\toprew{\alpha};op)\in \Deltark(q,\topn{1}(\stack))$. In case $op$ is neither of the form $\pushlk{n}{\beta}$ nor of the form $\collapse$ with $\topn{1}(\stack)$ having an $n$-link then the same transition $(p,\toprew{\alpha};op)$ is available in $\pprocessrk$ and is performed. 
The interesting case is when $op=\pushlk{n}{\beta}$, and it is simulated by $\pprocesslf$ as follows.
\begin{itemize}
\item The control state of $\pprocesslf$ is updated to be $p^\beta$ and one performs $\toprew{\alpha}$.
\item From $p^\beta$, \Eloise has to move to a new control state $p^?$ and can push any symbol of the form $(\alpha,\vect{R})$ where $\vect{R}=(R_0,\cdots R_d)\in (2^Q)^{\maxcolor+1}$. A dummy $1$-link is attached (and will never be used for a $\collapse$). 
\item From $p^?$, \Abelard has to play and choose between one of the following two options:
\begin{itemize}
\item either go to state $p$ and perform no action on the stack, 
\item or pick a state $r$ in some $R_i$, go to an intermediate new state $r^i$ (of colour $i$) without changing the stack and from this new configuration go to state $r$ and  perform a $pop_n$ action.
\end{itemize}
\end{itemize} 

The intended meaning of such a decomposition of the $\pushlk{n}{\beta}$ operation is the following: when choosing the sets in $\vect{R}$, \Eloise is claiming that she has a strategy such that if the $n$-link (or a later copy of it) created by pushing $\beta$ is eventually used for collapsing the stack then the control state after collapsing will belong to $R_i$ where $i$ is meant to be the smallest colour from the creation of the link to the collapse of the stack (equivalently it will be the link rank ---~as computed in $\Ark$~--- just before collapsing). Note that the $R_i$ are arbitrary sets because \Eloise does not have full control over the play (and in general cannot force $R_i$ to be a singleton). Then \Abelard can either choose to simulate the $\collapse$ (here state $r^i$ is only used for going through a state of colour $i$). If he does not want to simulate a $\collapse$ then one stores $\vect{R}$ since its truth may be checked later in the play. 

Assume that later, in configuration $(p',\varstack)$ one of the two players wants to simulate a transition $(r,\toprew{\beta};\collapse)$ involving an $n$-link. By construction, $\topn{1}(\varstack)$ is necessarily of the form $(\gamma,\vect{R})$. Then the simulation is done by going to a sink configuration that is winning for \Eloise iff $r\in R_{\LinkRk(p,\gamma)}$, \ie \Eloise wins iff her former claim on $\vect{R}$ was correct.

Formally we let $\pprocesslf = \anglebra{\Gammalf,\Qlf,\Deltalf,q_{0,\lf}}$
with 
\begin{itemize}
\item $\Gammalf = \Gammark \cup \Gammark\times (2^{\Qrk})^{\maxcolor+1}$
\item $\Qlf = \Qrk\cup \{p^\beta\mid p\in \Qrk,\ \beta\in\Gammark\}\cup \{p^?\mid p\in \Qrk\}\cup \{r^i\mid r\in \Qrk,\ 0\leq i\leq d\} \cup\{\ttrue,\ffalse\}$
\item $\Deltalf$ is defined as follows, where $p,q,r$ range over $\Qrk$, $\alpha,\beta,\gamma$ range over $\Gammark$ and $\vect{R}=(R_0,\dots,R_\maxcolor)$ ranges over $ (2^{\Qrk})^{\maxcolor+1}$.
\begin{itemize}
	\item  If $(p,\toprew{\alpha};op)\in\Deltark(q,\gamma)$ and if  $op$ is neither of the form $\pushlk{n}{\beta}$ nor $\collapse$, then $(p,\toprew{\alpha};op)\in\Deltalf(q,\gamma)$ and $(p,\toprew{(\alpha,\vect{R})};op)\in\Deltalf(q,(\gamma,\vect{R}))$.
	\item If $(p,\toprew{\alpha};\pushlk{n}{\beta})\in\Deltark(q,\gamma)$, then $(p^\beta,\toprew{\alpha};id)\in\Deltalf(q,\gamma)$ and $(p^\beta,\toprew{(\alpha,\vect{R})};id)\in\Deltalf(q,(\gamma,\vect{R}))$.
	\item For all $p^\beta\in \Qlf$, $\Delta(p^\beta,\gamma)=\Delta(p^\beta,(\gamma,\vect{R}))=\{(p^?,\pushlk{1}{(\beta,\vect{S})})\mid \vect{S}\in (2^{\Qrk})^{\maxcolor+1})\}$.
	\item For all $p^?\in \Qlf$, $\Delta(p^?,(\gamma,\vect{R}))=
	\{
		(p,id)
	\}
	\cup
	\{
		(r^i,id)\mid 0\leq i\leq \maxcolor \text{ and }r\in R_i
	\}$.
	\item For all $r^i\in \Qlf$, $\Delta(r^i,(\gamma,\vect{R}))=\{(r,\popn{n})\}$.
	\item  If $(p,\toprew{\alpha};\collapse)\in\Deltark(q,\gamma)$, then $(p,\toprew{\alpha};\collapse)\in\Deltalf(q,\gamma)$.
	\item  If $(r,\toprew{\alpha};\collapse)\in\Deltark(q,\gamma)$, then $(\ttrue,id)\in\Deltalf(q,(\gamma,\vect{R}))$ if $r\in R_{\LinkRk(q,\gamma)}$ and $(\ffalse,id)\in\Deltalf(q,(\gamma,\vect{R}))$ if $r\notin R_{\LinkRk(q,\gamma)}$.
	\item $\Deltalf(\ttrue,(\gamma,\vect{R})) = \{(\ttrue,id)\}$ and $\Deltalf(\ffalse,(\gamma,\vect{R})) = \{(\ffalse,id)\}$.
\end{itemize}
\end{itemize}

We let $\pgraphlf$ be the transition graph of $\pprocesslf$. Now, in order to define a game graph $\pggraphlf$ out of $\pgraphlf$ we let 
$Q_{\lf,\Ei} = Q_{\rk,\Ei} \cup  \{p^\beta\mid p\in \Qrk,\ \beta\in\Gammark\}$. 
Finally to define a corresponding $n$-CPDA parity game $\pgamelf$ we extend $\col$ by letting, $\forall p,r\in \Qrk$ and $\beta\in \Gammark$, $\col(p^\beta)=\col(p^?)=d$ (as one cannot loop forever in such states, it means that they have no influence on whether a play will be winning or not), $\col(r^i)=i$ for every $0\leq i\leq d$, $\col(\ttrue)=0$ and $\col(\ffalse)=1$ (hence a play that visits $\ttrue$ is winning for \Eloise and a play that visits $\ffalse$ is winning for \Abelard, as these states are sinks). 

Note that $\pprocesslf$ never creates an $n$-link.

\subsection{Correctness of the Simulation}\label{section:Step2Task2}

Consider some configuration $v_0=(p_0,s_0)$ in $\pgamerk$. We explain now how to define an ``equivalent'' configuration $\nu(v_0)$ in $\pgamelf$ (here equivalent is in the sense of Lemma~\ref{lemma:removinglinks} below). The transformation consists in replacing any occurrence of a stack letter (call it $\gamma$) with an $n$-link in $\stack_0$ by another letter of the form $(\gamma,\vect{R})$ and replacing the $n$-link by a $1$-link. The vector $\vect{R}$ is defined as follows. Let $\stack'$ be the stack obtained by popping every symbol and stack above $\gamma$, and let $R=\{q\mid \text{ \Eloise wins in $\pgamerk$ from }(q,collapse(\stack'))\}$. Then one sets $\vect{R}=(R,\cdots,R)$.

\begin{example}
Assume we are playing a two-colour parity game and let
$$         \stack_0=\pstr[0.8cm]{\mksk{\mklksk{n1}{\mksk{ \, a}} \;
                     \mklksk{n19}{\mklksk{n2}{}
                           \mksk{ \, a \, \nd(n3-n2,40){b} \, \nd(n4-n1,39){c}}} \;
                     \mksk{\mklksk{n5}{}
                           \mksk{ \, a \, \nd(n6-n5,40){b} \, \nd(n7-n1,37){c}\, \nd(n117-n19,41){d}}}}},
 $$
$$R_1=\{r\mid (r,\mksk{\mksk{\mksk{a}}}) \text{ is winning for \Eloise in }\pgamerk\}$$
$$R_2=\{r\mid (r,\pstr[0.3cm]{\mksk{\mklksk{n1}{\mksk{ \, a}} \;
                     \mklksk{n19}{\mklksk{n2}{}
                           \mksk{ \, a \, \nd(n3-n2,40){b} \, \nd(n4-n1,39){c}}}}}) \text{ is winning for \Eloise in }\pgamerk\}$$
 Then 
 $$         \nu(\stack_0)=\pstr{\mksk{\mklksk{n1}{\mksk{ \, a}} \;
                     \mksk{\mklksk{n2}{}
                           \mksk{ \, a \, \nd(n3-n2,40){b} \, {(c,(R_1,R_1))}}} \;
                     \mksk{\mklksk{n5}{}
                           \mksk{ \, a \, \nd(n6-n5,40){b} \, {(c,(R_1,R_1))\, (d,(R_2,R_2))}}}}}.
 $$
\end{example}

The rest of this section is devoted to the proof of the following result.

\begin{lemma}\label{lemma:removinglinks}
\Eloise wins in $\pgamerk$ from some configuration $v_0$ if and only if she wins in $\pgamelf$ from $\nu(v_0)$.
\end{lemma}

%In the following, we intensively work with strategy (both in $\pgamerk$ and $\pgamelf$). When defining the value of such a strategy on some partial play we may alternatively define it as a vertex (which respects the definition of a strategy as we gave it) or as a pair $(q,op)$ formed by a control state and a stack operation. In the latter, one has to understand this as the strategy that goes to the configuration $(q,op(s))$ if $s$ denotes the stack in the current configuration.

Assume that the configuration $v_0=(p_0,\stack_0)$ is winning for \Eloise in $\pgamerk$, and let $\pstratrk$ be a winning strategy for her. Using $\pstratrk$, we define a strategy $\pstratlf$ for \Eloise in $\pgamelf$ from $\nu(v_0)$. The strategy $\pstratlf$ maintains as a memory a partial play $\pplayrk$ in $\pgamerk$, that is an element in $\Vrk^*$ (where $\Vrk$ denotes the set of vertices of $\pgraphrk$).  
%This memory will be denoted $\pplayrk$. 
At the beginning $\pplayrk$ is initialised to be $(p_{0},\stack_0)$. 
The play $\pplayrk$ will satisfy the following invariant: assume that the play ends in a configuration $(q,\stack)$, then the last configuration in $\pplayrk$ has control state $q$ and its $\topn{1}$-element is either $\topn{1}(\stack)$ or $(\topn{1}(\stack),\vect{R})$ for some $\vect{R}$ (and in this case there is an $n$-link from the $\topn{1}$-symbol of $\stack$).

We first describe $\pstratlf$, and then we explain how $\pplayrk$ is updated.

\medskip
\noindent\textbf{Choice of the move.}
Assume that the play is in some vertex $(q,\stack)$ with $q\in Q_{\lf,\Ei}\setminus\{p^\beta\mid q\in \Qrk,\ \beta\in\Gammark\}$. The move given by $\pstratlf$ depends on $\pstratrk(\pplayrk)=(p,\toprew{\alpha};op)$ (we shall later argue that $\pstratlf$ is well defined whilst proving that it is winning).
\begin{itemize}
\item If $op$ is neither of the form $\pushlk{n}{\beta}$ nor $\collapse$ then \Eloise plays $(p,\toprew{\alpha};op)$ if $\topn{1}(\stack)=\gamma$ and she plays  $(p,\toprew{(\alpha,\vect{R})};op)$ if $\topn{1}(\stack)=(\gamma,\vect{R})$.
\item If $op=collapse$ and $\topn{1}(\stack)=\gamma\in\Gammark$ then \Eloise plays $(p,\toprew{\alpha};\collapse)$.
\item If $op=collapse$ and $\topn{1}(\stack)=(\gamma,\vect{R})$ then \Eloise plays $(\ttrue,id)$. We shall later see that this move is always valid.
\item If $op=\pushlk{n}{\beta}$ then \Eloise plays $(p^\beta,\toprew{\alpha};id)$ if $\topn{1}(\stack)=\gamma$ and she plays  $(p^\beta,\toprew{(\alpha,\vect{R})};id)$ if $\topn{1}(\stack)=(\gamma,\vect{R})$.
\end{itemize}

In this last case, or in the case where $q\in Q_\Ai$ and \Abelard plays some $(p^\beta,\toprew{\alpha};id)$ (\resp some $(p^\beta,\toprew{(\alpha,\vect{R})};id)$), we also have to explain how \Eloise behaves from $(p^\beta,\toprew{\alpha}(\stack))$ (\resp $(p^\beta,\toprew{(\alpha,\vect{R})}(\stack))$.

\Eloise has to play $(p^?,\pushlk{1}{(\beta,\vect{S})})$ where $\vect{S}\in (2^{\Qrk})^{d+1}$ describes which states
can be reached if the $n$-link created by pushing $\beta$ (or a copy of it) is used for collapsing the stack, depending on the smallest visited colour in the
meantime. In order to define $\vect{S}$, she considers the set of
all possible continuations of $\pplayrk\cdot(p,\pushlk{n}{\beta}(\varstack))$ (where
$(q,\varstack)$ denotes the last vertex of $\pplayrk$) where she
respects her strategy $\pstratrk$. For each such play, she checks whether some
configuration of the form $(r,\popn{n}(\varstack))$ is eventually reached by collapsing (possibly a copy of the) $n$-link created by $\pushlk{n}{\beta}$. 
If such an $r$ exists, she considers the smallest colour $i$ visited from the moment where the link was created to the moment $\collapse$ is performed (\ie the link rank just before collapsing).
For every $i\in\{0,\dots d\}$, the set $S_i$ is defined to be the set of states
$r\in Q$ such that the preceding case happens. 
Formally, 
\begin{multline*}
S_i=\{r\mid \exists\ \pplayrk\cdot v_0\cdots v_k\cdot v_{k+1}\cdots\text{ play in } \pgamerk \text{ where \Eloise respects } \pstratrk \text{ and s.t. }\\ v_0=(p,\pushlk{n}{\beta}(\varstack)),\,   v_{k+1}=(r,pop_n(\varstack))\text{ is obtained by applying $\collapse$ from } v_k,\\ 
v_0 \text{ is the link ancestor of } v_k\text{ and }  i \text{ is the link rank in $v_k$}\}
\end{multline*}

Finally, we set $\vect{S}=(S_0,\dots,S_d)$ and \Eloise plays $(p^?,\pushlk{1}{(\beta,\vect{S})})$.

\medskip
\noindent\textbf{Update of $\pplayrk$.}
The memory $\pplayrk$ is updated after each visit to a configuration with a control state in $\Qrk\cup\{\ttrue,\ffalse\}$. We have several cases depending on the transition.
\begin{itemize}

\item If the last transition is of the form $(p,\toprew{\alpha};op)$ or $(p,\toprew{(\alpha,\vect{R})};op)$ with $op$ being neither of the form $\pushlk{n}{\beta}$ nor $\collapse$, then we extend $\pplayrk$ by applying transition $(p,\toprew{\alpha};op)$, \ie if $(q,\varstack)$ denotes the last configuration in $\pplayrk$, then the updated memory is $\pplayrk\cdot (p,op(\toprew{\alpha}(\varstack)))$.
%
%\item If the transition is of the form $(q,\toprew{\alpha})$ or $(q,\toprew{(\alpha,\vect{R})})$ and if $(p,\sigma)$ denotes the last configuration in $\pplayrk$, then the updated memory is $\pplayrk\cdot (q,\toprew{\alpha}(\sigma))$.
%%
\item If the last transition is of the form $(\ttrue,id)$ or $(\ffalse,id)$, the play is in a sink configuration. Therefore we do not update $\pplayrk$ as the play will loop forever.
\item If the last transitions form a sequence of the form $(p^\beta,\toprew{\alpha};id)\cdot(p^?,\pushlk{1}{(\beta,\vect{S})})\cdot (p,id)$ or of the form $(p^\beta,\toprew{(\alpha,\vect{R})};id)\cdot(p^?,\pushlk{1}{(\beta,\vect{S})})\cdot (p,id)$, then the updated memory is 
$\pplayrk\cdot (p,\pushlk{n}{\beta}(\varstack))$, where $(q,\varstack)$ denotes the last configuration in $\pplayrk$.
\item If the last transitions form a sequence of the form 
$(p^\beta,\toprew{\alpha};id)\cdot(p^?,\pushlk{1}{(\beta,\vect{S})})\cdot (r^i,id)\cdot(r,\popn{n})$
or of the form
$(p^\beta,\toprew{(\alpha,\vect{R})};id)\cdot(p^?,\pushlk{1}{(\beta,\vect{S})})\cdot (r^i,id)\cdot(r,\popn{n})$, then we extend $\pplayrk$ by a sequence of actions (consistent with $\pstratrk$) that starts by performing transition $(p,\pushlk{n}{\beta})$ and ends up by collapsing (possibly a copy of) the link created at this first step and goes to state $r$ whilst visiting $i$ as a minimal colour in the meantime. By definition of $\vect{S}$ such a sequence always exists. More formally, if $(q,\varstack)$ denotes the last configuration in $\pplayrk$, then the updated memory is a play in $\pgamerk$, $\pplayrk\cdot v_0\cdots v_k\cdot v_{k+1}$, where \Eloise respects $\pstratrk$ and such that 
$v_0=(p,\pushlk{n}{\beta}(\varstack))$, $v_{k+1}=(r,pop_n(\varstack))$ is obtained by applying $\collapse$ from $v_k$, $v_0$ is the link ancestor of $v_k$ and $i$ is the link rank in $v_k$.
\end{itemize}

Therefore, with any partial play $\pplaylf$ in $\pgamelf$ starting from $v_0$ in which \Eloise
respects her strategy $\pstratlf$, is associated a partial play $\pplayrk$ in
$\pgamerk$. An immediate induction shows that $\pplayrk$ is a play where \Eloise respects $\pstratrk$. 
The same argument works for any infinite play $\pplaylf$ that does not contain a state in $\{\ttrue,\ffalse\}$, and the
corresponding play $\pplayrk$ is therefore infinite, starts from
$\nu(v_0)$ and \Eloise respects $\pstratrk$ in that play. Therefore it is
a winning play.

Moreover, if $\pplaylf$ is an infinite play that does not contain a state in $\{\ttrue,\ffalse\}$, it easily follows from the definitions of $\pstratlf$ and $\pplayrk$ that the smallest infinitely visited colour in $\pplaylf$ is the same as the one in $\pplayrk$. Hence, any infinite play in $\pgamelf$ starting from $\nu(v_0)$ where \Eloise respects $\pstratlf$ and that does not contain a state in $\{\ttrue,\ffalse\}$ is won by \Eloise.

Now, consider a play that contains a state in $\{\ttrue,\ffalse\}$ (hence loops on it forever). Reaching a configuration with state in $\{\ttrue,\ffalse\}$ is necessarily by simulating a $\collapse$ from some configuration with a $\topn{1}$-element of the form $(\alpha,\vect{R})$. We should distinguish between those elements $(\alpha,\vect{R})$ that are ``created'' before (\ie by the $\nu$ function) or during the play (by \Eloise). For the second ones, note that whenever \Eloise wants to simulate a collapse, she can safely go to state $\ttrue$ (meaning $\pstratlf$ is well defined): indeed, if this was not the case, it would contradict the way $\vect{S}$ was defined when simulating the original creation of the link. For the same reason, \Abelard can never reach state $\ffalse$ provided \Eloise respects her strategy $\pstratlf$. Now consider an element $(\alpha,\vect{R})$ created by $\nu$ and assume that one player wants to simulate a collapse from some configuration with such a $\topn{1}$-element. Call $\pplaylf$ the partial play just before and call $\pplayrk$ the associated play in $\pgamerk$. Then in $\pplayrk$, \Eloise respects her winning strategy $\pstratrk$. If she has to play next in $\pplayrk$, strategy $\pstratrk$ indicates to play $\collapse$; if it is \Abelard's turn to move he can play $\collapse$. In both cases, the configuration that is reached after collapsing is winning for \Eloise (it is a configuration visited in a winning play). Hence, by definition of $\nu$, its control state belongs to $R$ where $\vect{R}=(R,\cdots,R)$, and therefore from the current vertex in $\pgamelf$ there is no transition to $\ffalse$ and there is at least one to $\ttrue$. Therefore plays  where \Eloise respects $\pstratlf$ and that contain a state in $\{\ttrue,\ffalse\}$ necessarily contains state $\ttrue$ hence are won by \Eloise.

Altogether, it proves that $\pstratlf$ is a winning strategy for \Eloise in $\pgamelf$ from $\nu(v_0)$.

\medskip

Let us now prove the converse implication. 
%\footnote{Note that in order to prove the converse implication one can follow the similar lines as for the direct implication and consider \Abelard's point of view. Nevertheless, the proof we give here starts from a winning strategy for \Eloise in $\pgamelf$ and constructs from it a strategy in $\pgamerk$: this induces a more involved proof but has the advantage to lead to an effective construction of a winning strategy for \Eloise in $\pgamerk$ if one has an effective strategy for her in $\pgamelf$ (later, we shall argue that this is the case}. 
%
Assume that the configuration $\nu(v_{0})$ is winning for \Eloise in $\pgamelf$, and let $\pstratlf$ be a winning strategy for her. Using $\pstratlf$, we define a strategy $\pstratrk$ for \Eloise in $\pgamerk$ from $v_0=(p_{0},\stack_0)$. 
First, recall how $\nu(v_0)$ is defined: every symbol $\gamma$ in $\stack_0$ with an $n$-link is replaced by a pair $(\gamma,(R,\ldots,R))$ where $R$ is the set of states $r$ such that \Eloise wins from $(r,\stack')$ where $\stack'$ is the stack obtained by first removing every symbol (and stack) above $\gamma$ and then performing a $\collapse$. We can therefore assume that we have a collection of winning strategies, one for each such configuration $(r,\stack')$; call such a strategy $\pstratrk^{r,\stack'}$. Then, during a play where \Eloise respects $\pstratrk$, if one eventually visits such a configuration $(r,\stack')$, the strategy $\pstratrk$ will mimic the winning strategy $\pstratrk^{r,\stack'}$ from that point and therefore the resulting play will be winning for \Eloise. Then in the rest of this description we mostly focus on the case of plays where this situation does not occur.

The strategy $\pstratrk$ maintains as a memory a partial play $\pplaylf$ in $\pgamelf$, that is an element in $V_{\lf}^*$ (where $V_{\lf}$ denotes the set of vertices of $\pgraphlf$).
At the beginning $\pplaylf$ is initialised to the configuration $\nu(v_{0})$. 
After having played $\pplayrk$, the play $\pplaylf$ will satisfy the following invariant.
Assume that the play $\pplaylf$ ends in a configuration $(q,\stack)$ then the following holds.
\begin{itemize}
\item If $\topn{1}(\stack)=\alpha$, the last configuration of $\pplayrk$ has control state $q$ and its $\topn{1}$-element is $\alpha$ and it has a $k$-link for some $k<n$.
\item If $\topn{1}(\stack)=(\alpha,\vect{R})$, the last configuration of $\pplayrk$ has control state $q$, its $\topn{1}$-element is $\alpha$ and it has an $n$-link. Moreover, if \Eloise keeps respecting $\pstratrk$ in the rest of the play, if (possibly a copy of) this link is eventually used in a $\collapse$, then the state that will be reached just after doing the $\collapse$ will belong to $R_i$ where $i$ will be the link rank just before collapsing.
\end{itemize}

We first describe $\pstratrk$ and we then explain how $\pplaylf$ is updated. Recall that we switch to a known winning strategy in case we do a $\collapse$ from (possibly a copy of) an $n$-link that was already in $\stack_0$.

\medskip
\noindent\textbf{Choice of the move.}
Assume that the play is in some vertex $(q,\stack)$ with $q\in Q_{\rk,\Ei}$. The move given by $\pstratrk$ depends on $\pstratlf(\pplaylf)=(q',rew;op)$ (we shall later argue that $\pstratrk$ is well defined whilst proving that it is winning).
\begin{itemize}
\item If $q'\in \Qrk$ then \Eloise plays $(q',\toprew{\alpha};op)$ where $\alpha$ is such that either $rew = \toprew{\alpha}$ or $rew = \toprew{(\alpha,\vect{R})}$. Note that in this case, $op$ is neither a $\collapse$ involving an $n$-link nor of the form $\pushlk{n}{\beta}$.
\item If $q'=p^\beta$ then \Eloise plays to $(p,\toprew{\alpha};\pushlk{n}{\beta})$ where $\alpha$ is such that either $rew = \toprew{\alpha}$ or $rew = \toprew{(\alpha,\vect{R})}$.
\item If $q'=\ttrue$ then \Eloise plays $(r,\collapse)$ for some arbitrary $r\in R_{\LinkRk(p,\topn{1}(\stack))}$ where $(\alpha,\vect{R})$ denotes the $\topn{1}$-element of the last vertex of $\pplaylf$. Note that in this case, the collapse involves an $n$-link.
\end{itemize}

\medskip
\noindent\textbf{Update of $\pplaylf$.}
The memory $\pplaylf$ is updated after each move (played by any of the two players). We have several cases depending on the last transition.
\begin{itemize}
\item If the last transition is of the form $(q',\toprew{\alpha};op)$ and $op$ is neither a $\collapse$ involving an $n$-link nor of the form $\pushlk{n}{\beta}$, then $\pplaylf$ is extended by mimicking the same transition, \ie if $(q,\varstack)$ denotes the last configuration in $\pplaylf$, then the updated memory is $\pplaylf\cdot (q',op(\toprew{\alpha}(\varstack))$ if $\topn{1}(\varstack)=\gamma$ for some $\gamma\in\Gammark$, and is $\pplaylf\cdot (q',op(\toprew{(\alpha,\vect{R})}(\varstack))$ if $\topn{1}(\varstack)=(\gamma,\vect{R})$ for some $(\gamma,\vect{R})\in\Gammalf$.
\item If the last transition is of the form $(p,\toprew{\alpha};\pushlk{n}{\beta})$ then, we let $(q,\varstack)$ denote the last configuration in $\pplaylf$. 
If $\topn{1}(\varstack)=\gamma$ for some $\gamma\in\Gammark$ then the updated memory is
$\pplaylf\cdot (p^\beta,\toprew{\alpha}(\varstack))\cdot (p^?,\pushlk{1}{(\beta,\vect{R})}(\toprew{\alpha}(\varstack))) \cdot (p,id)$ where $\pstratlf(\pplaylf\cdot (p^\beta,\toprew{\alpha}(\varstack)))=(p^?,\pushlk{1}{(\beta,\vect{R})}(\toprew{\alpha}(\varstack)))$.
\\
If $\topn{1}(\varstack)=(\gamma,\vect{S})$ for some $(\gamma,\vect{S})\in\Gammalf$ then the updated memory is
$\pplaylf\cdot (p^\beta,\toprew{(\alpha,\vect{S})}(\varstack))\cdot (p^?,\pushlk{1}{(\beta,\vect{R})}(\toprew{(\alpha,\vect{S})}(\varstack))) \cdot (p,id)$ where $\pstratlf(\pplaylf\cdot (p^\beta,\toprew{(\alpha,\vect{S})}(\varstack)))=(p^?,\pushlk{1}{(\beta,\vect{R})}(\toprew{(\alpha,\vect{S})}(\varstack)))$.
\item If the last transition is of the form $(r,\collapse)$ and the $\collapse$ follows an $n$-link, then we have two cases. In the first case, the $\collapse$ follows (possibly a copy of) an $n$-link that was already in $\stack_0$ and we claim (and prove later) that one ends up in a winning configuration and thus one switches to a corresponding winning strategy as already explained.  In the other case, it follows an $n$-link that was created during the play, in which case we let $\pplaylf=v_0\cdots v_m$ and denote by $v_i$ the link ancestor of $v_m$\footnote{Here we implicitly extend the notion of link ancestor as follows. In $\pgamelf$ instead of creating $n$-link one pushes symbol of the form $(\beta,\vect{R})$: hence whenever doing a $\pushlk{1}{(\beta,\vect{R})}$ one attaches to the vector $\vect{R}$ the index of the current configuration. Then if the $\topn{1}$ element of $v_n$ is some $(\beta,\vect{R})$ then the link ancestor of $v_m$ is defined to be $v_i$ where $i$ is the indexed attached with $\vect{R}$. Note in particular that the control state in the link ancestor is necessarily of the form $p^?$.}.
Then the updated memory is obtained by backtracking inside $\pplaylf$ until reaching the configuration where the (simulation of the) collapsed $n$-link was created (this configuration is $v_i$, the link ancestor) and then extending it by a choice of \Abelard consistent with the $\collapse$. 
That is the updated memory is $v_0\cdots v_i\cdot (r^\ell,\varstack)\cdot (r,\popn{n}(\varstack))$ where $v_i = (p^?,\varstack)$ and
$\ell$ denotes the link rank in the configuration $\pplayrk$ was just before doing the $\collapse$.
%That is, if $\lambda=v_0\cdots v_m$, and if $v_i=(p^?,\sigma)$ is the link ancestor of $v_m$, then the updated memory is $v_0\cdots v_i\cdot (r^i,pop_n(\sigma))\cdot (r,pop_n(\sigma))$ where $i$ denotes the link rank in the configuration $\pplayrk$ was just before collapsing.
\end{itemize}

Therefore, with any partial play $\pplayrk$ in $\pgamerk$ in which \Eloise
respects her strategy $\pstratrk$, is associated a partial play $\pplaylf$ in
$\pgamelf$. Note that if we end up in a configuration that is known to be winning, $\pplaylf$ becomes useless and is no longer extended. This also implies that when collapsing an $n$-link that was already in $\stack_0$ one necessarily ends up in a winning configuration. Indeed, assume the contrary and let $\pplaylf$ be the constructed play before collapsing: then either \Eloise has to play and therefore moves to $\ttrue$ (and therefore the configuration in $\pplayrk$ after collapsing is winning by definition of $\nu$, leading a contradiction) or \Abelard could move to $\ffalse$ (leading a contradiction with $\pstratlf$ being winning). Therefore, from now on, we restrict our attention to the case where the $n$-links (and their copies) originally in $\stack_0$ are never used to do a $\collapse$.

An easy induction shows that \Eloise respects $\pstratlf$ in
$\pplaylf$. The same argument works for an infinite play $\pplayrk$, and the
corresponding play $\pplaylf$ is therefore infinite (one simply considers the limit of the $\pplaylf$ in the usual way\footnote{
\label{footnote:limit}
Let $(u_m)_{m\geq 0}$ be a sequence of finite words. For any $m\geq 0$ let $u_m=u_{m,0}\cdots u_{m,{k_m}}$. Then the limit of the sequence $(u_m)_{m\geq 0}$ is the (possibly infinite) word $\alpha=\alpha_0\alpha_1\cdots$ such that $\alpha$ is maximal for the prefix ordering and for all $0\leq i<|\alpha|$ there is some $N_i$ such that $u_{m,i} = \alpha_i$ for all $m\geq N_i$. 

In our setting, the play $\pplaylf$ associated with an infinite play $\pplayrk$ is defined as the limit of the sequence of partial plays $(\pplaylf^m)_{m\geq 0}$ where $\pplaylf^m$ is the partial play associated with $\pplayrk$ truncated to its $m+1$ first vertices. From the definitions of the $\pplaylf^m$ it is easily verified that the limit $\pplaylf$ is infinite.
}), starts from
$\nu(v_0)$, never visits a state in $\{\ttrue,\ffalse\}$ and \Eloise respects $\pstratlf$ in that play. Therefore it is
a winning play.

Now, in order to conclude that any play $\pplayrk$ in $\pgamerk$ in which \Eloise
respects strategy $\pstratrk$ is winning for her, one needs to relate the sequence of colours in $\pplayrk$ with the one in $\pplaylf$. For this, we introduce a notion of factorisation of a partial play $\pplayrk=v_0v_1\cdots v_m$ in $\pgamerk$ (we should later note that it directly extends to infinite plays). A factor is a nonempty sequence of vertices of the following kind:
\begin{enumerate}
\item[(1)] it is a sequence $v_h\cdots v_k$ such that the stack operation from $v_{h-1}$ to $v_h$ is of the form $\toprew{\alpha};\pushlk{\beta}{n}$, the stack operation from $v_{k-1}$ to $v_{k}$ is a $\collapse$ involving an $n$-link, and $v_h$ is the link ancestor of $v_k$. 
\item[(2)] or it is a single vertex;
\end{enumerate}
Then the factorisation of $\pplayrk$ denoted $Fact(\pplayrk)$ is a sequence of factors inductively defined as follows (we underline factors to make them explicit):
$Fact(\pplayrk)=\underline{v_0\cdots v_k},Fact(v_{k+1}\cdots v_n)$ if there exists some $k$ such that $v_0\cdots v_k$ is as in (1) above, and 
$Fact(\pplayrk)=\underline{v_0},Fact(v_1\cdots v_n)$ otherwise.
%
%$$Fact(\pplayrk)=
%	\begin{cases}
%		[v_0\cdots v_k],Fact(v_{n+1}\cdots v_n) & \text{if } \exists k \text{ s.t. } v_0\cdots v_k \text{ is as in (1)}\\
%		[v_0],Fact(v_1\cdots v_n) & \text{otherwise}\\
%	\end{cases}
%$$
In the following, we refer to the \emph{colour} of a factor as the minimal colour of its elements. 

Note that the previous definition is also valid for infinite plays. Now we easily get the following proposition (the result is obtained by reasoning on partial play using a simple induction combined with a case analysis. Then it directly extends to infinite plays).

\begin{proposition}
Let $\pplayrk$ be some infinite play in $\pgamerk$ starting from $v_0$ where \Eloise respects $\pstratrk$ and assume that there is no $\collapse$ that follows (possibly a copy of) an $n$-link already in $\stack_0$. Let $\pplaylf$ be the associated infinite play in $\pgamelf$ constructed from $\pstratrk$. Let $\play_{\rk,0},\play_{\rk,1},\cdots$ be the factorisation of $\pplayrk$ and, for every $i\geq 0$, let $c_i$ be the colour of $\play_{\rk,i}$. 

Then the sequence $(c_i)_{i\geq 0}$ and the sequence of colours visited in $\pplaylf$ have the same $\liminf$.
\end{proposition}

The previous proposition directly implies that $\pstratrk$ is a winning strategy for \Eloise from $v_{0}$ in $\pgamerk$.

%% Fin 4 avril
\subsection{Regularity of the Winning Region is Preserved}

We established in Lemma~\ref{lemma:removinglinks} that 
\Eloise wins in $\pgamerk$ from some configuration $v_0$ if and only if she wins in $\pgamelf$ from $\nu(v_0)$. We now prove that regular sets of winning positions are preserved by inverse image by $\nu$.

\begin{proposition}\label{prop:inverseimagenu2}
Assume that we have an automaton $\Blf$ that recognises the set of winning configurations in $\pgamelf$. Then, one can compute an automaton $\Brk$ that recognises the set of winning configurations in $\pgamerk$.
\end{proposition}

\begin{proof}
We can safely assume that any control state of $\Blf$ is of the form $(\xi,R)$ with $R\subseteq \Qlf$ and such that, after reading some input stack $\stack$ (possibly with some pending open brackets) $\Blf$ is in a state of the form $(\xi,R)$ with $R=\{r\mid \Blf \text{ accepts }(r,\stack')\}$ where $\stack'$ is the stack obtained from $\stack$ by closing all the pending open brackets (\ie $\stack'=\stack]^k$ for some well chosen $k\leq n$).

On an input $(p_0,\stack_0)$ the automaton $\Brk$ computes \emph{on-the-fly} the image of $(p_0,\stack_0)$ by $\nu$ and simulates $\Blf$ on it.
In order to compute $\nu((p_0,\stack_0))$, $\Brk$ needs to retrieve, when reading a stack symbol with an $n$-link, the states that are winning for the stack obtained by collapsing the $n$-link. This is simple as it is given by the $2^{\Qlf}$ component of $\Blf$ (recall that $\Brk$ simulates $\Blf$, hence keeps track of this information) and hence the automaton can access it by definition of the model of automata. 
Indeed, the information (\ie the states winning when doing a collapse) is correct before reading the first stack symbol coming with an $n$-link, and by induction on the number of $n$-links, if it is correct after processing the $k$ first symbols with an $n$-link, on reading the $(k+1)$-th symbol with an $n$-link, the information is still correct as it was correct for the prefix read so far and therefore $\Brk$ correctly simulated $\Blf$ on this prefix.

We do not formally describe $\Brk$ as it is rather straightforward but we note that the size of $\Brk$ is linear in the size of $\Blf$.
\end{proof}

\subsection{Strategies}

In order to complete the proof of Theorem~\ref{theo:outmost} it remains to establish the following proposition.

\begin{proposition}\label{proposition:outmost-ws}
If there is an $n$-CPDA transducer $\Slf$ synchronised with $\pprocesslf$ realising a well-defined winning strategy for \Eloise in $\pgamelf$ from $\nu((q_{0,\rk},\bot_{n}))$, then one can effectively construct an $n$-CPDA transducer $\Srk$ synchronised with $\pprocessrk$ realising a well-defined winning strategy for \Eloise in $\pgamerk$ from the initial configuration $(q_{0,\rk},\bot_{n})$.
\end{proposition}

\os{I think I am formal enough in the following proof. However note that it is actually tricky how to deal with the collapse involving $n$-links.}
\begin{proof}
The result follows from a carefully analysis of how we defined $\pstratrk$ from $\pstratlf$ in the proof of Lemma~\ref{lemma:removinglinks}. As we now only focus on the initial configuration $(q_{0,\rk},\bot_{n})$ we will not have to deal with the special case of doing a $\collapse$ following (possibly a copy of) an $n$-link originally in the initial configuration. Also note that $\nu((q_{0,\rk},\bot_{n}))=(q_{0,\rk},\bot_{n})$.

Recall that $\pstratrk$ uses as a memory a partial play $\pplaylf$ in $\pgamelf$ and considers the value of $\pstratlf(\pplaylf)$ to determine the next move to play. Now assume that $\pstratlf$ is realised by an $n$-CPDA transducer $\Slf$ synchronised with $\pprocesslf$. Hence, instead of storing $\pplaylf$ it suffices to store the configuration $\Slf$ is in after reading $\pplaylf$.

One can also notice that the stack $\stack_\rk$ in the last configuration of some partial play $\pplayrk$ and the stack $\stack_\lf$ in the last configuration of the associated $\pplaylf$ have the same shapes \emph{provided} one replaces in $\stack_\lf$ every $1$-link from a symbol in $\Gammark\times (2^{\Qrk})^{\maxcolor+1}$ by an $n$-link. Recall that these $1$-links are never used to perform a $\collapse$: hence replacing those $1$-links by $n$-links does not change the issue of the game, and if one does a similar transformation on $\Slf$ it still realises a winning strategy, and it is synchronised with the transformed version of $\pplaylf$.

Now, it follows from the way one defined $\pstratrk$ (both the choice of the move and the memory update) that one can design an $n$-CPDA transducer $\Srk$ synchronised with $\pprocessrk$ realising a well-defined winning strategy for \Eloise in $\pgamerk$ from the initial configuration $(q_{0,\rk},\bot_{n})$. In all cases but one $\Srk$ simulates $\Slf$. The only problematic case is when the move to play is some $(r,\collapse)$ involving an $n$-link.
Indeed, one needs to backtrack in $\pplaylf$ (namely retrieve the configuration of $\Slf$ right after the link ancestor) and extend it by doing $(r^\ell,id)$ (where $\ell$ is the link rank) and then $(r,\popn{n})$; one needs to retrieve the configuration of $\Slf$ right after this. 
If one performs a $\collapse$ in $\Srk$, one directly retrieves the stack content, but the control state of $\Slf$ is still missing. However, one can modify $\Slf$ so that after the simulation of the creation of an $n$-link, \ie after a symbol of the form $(\beta,\vect{R})$ is pushed, it stores in its $\topn{1}$-element the control state it will be in after doing the transitions $(r^\ell,id)(r,\popn{n})$, for each $0\leq \ell \leq d$ and each $r\in R_\ell$ (this can easily be computed).
As this information is then propagated when copying the symbol/link, it is available in the $\topn{1}$-element before doing a $\collapse$ involving an $n$-link, hence $\Srk$ can also correctly retrieve the control state of $\Slf$. 

From this (somehow informal) description of $\Srk$ the reader should be convinced that $\Srk$ correctly simulates $\Slf$ on $\pplaylf$ and hence, realises a winning strategy in $\pgamerk$. The fact that $\Srk$ is synchronised with $\pprocessrk$ follows from the fact that it is synchronised with the variant of $\Slf$ that itself is synchronised with the variant of $\pplaylf$ which is synchronised with $\pplayrk$.
\end{proof}

\subsection{Optimising the Construction}\label{section:optimisationlf}
\os{I am not very happy with the following. Actually, I noticed this after doing the whole proof. Of course I could fix the proof but I think it would even make it more complicated whilst I think the patch below should be convincing enough...}

The set $\Qlf$ has size $\mathcal{O}(|\Qrk|(|\Gammark|+d+3))$, which is not very satisfactory for complexity reasons. Actually, one would prefer a variant of the construction where $|\Gammark|$ does not appear in the blowup concerning states. This factor actually comes from states $\{q^\gamma\mid q\in \Qrk,\ \gamma\in\Gammark\}$, and one can easily get rid of them by doing the following modification on $\pprocesslf$. When simulating a $\pushlk{n}{\beta}$, instead of going to $q^\beta$, one stores the information on $\beta$ (thanks to a $\toprew{}$ operation) in the $\topn{1}$ element of the stack (hence, the stack alphabet increases by a linear factor in $|\Gammark|$) and goes to a special state $q^!$. State $q^!$ is controlled by \Eloise and the transition function is the same as from $q^\beta$ where $\beta$ is the symbol stored on the $\topn{1}$-element of the stack.

It is straightforward that this modification does not change the validity of Proposition~\ref{prop:inverseimagenu2} nor Proposition~\ref{proposition:outmost-ws}.
%\am{Can you revise/remove the last sentence? The preceding one seems to imply that the Propositions still hold. }
\subsection{Complexity}

If we summarise, the overall blowup in the transformation from $\pgamerk$ to $\pgamelf$ given by Theorem \ref{theo:outmost} is as follows.

\begin{proposition}\label{proposition:complexity-step2}
Let $\pprocessrk$ and $\pprocesslf$ be as in Theorem \ref{theo:outmost}. Then the set of states of $\pprocesslf$ has size 
$\mathcal{O}(|\Qrk|(|\colors|+3))$ and the stack alphabet of $\pprocesslf$ has size $\mathcal{O}(|\Gammark|^2 \cdot 2^{|\Qrk||\colors|})$.
Finally, the set of colours used in $\pgamerk$ and $\pgamelf$ are the same.
\end{proposition}

\begin{proof}
By construction together with the optimisation discussed in Section \ref{section:optimisationlf}.
\end{proof}

%% file: ReducingOrder.tex
% !TEX root = main.tex

\section{Reducing the Order}\label{section:reducingOrder}

In the previous section,  given a game played on a \emph{rank-aware} $n$-CPDA,
we have constructed another game played on an $n$-CPDA that does not create $n$-links. The winning region (\emph{resp.} a winning strategy realised by an $n$-CPDA transducer) in the original game can then be recovered from the winning region (\emph{resp.}
a  winning strategy realised by $n$-CPDA transducer) in the latter game. 
 
 In this section, we prove a result of a similar flavour. Namely, starting from a game played on an $n$-CPDA that does not create $n$-links, we construct a game played on an $(n-1)$-CPDA, and we show that the winning region (\emph{resp.} a winning strategy realised by an $n$-CPDA transducer) in the original game can be recovered from the winning region (\emph{resp.} a winning strategy realised by an $(n-1)$-CPDA transducer) in the latter game.

We situate the techniques developed here in a general and abstract framework of (\hbox{order-1}) pushdown automata whose stack alphabet is a \emph{possibly infinite} set: \emph{\sdps.}
We start by introducing this concept and show how $n$-CPDA that do not create $n$-links fit into it. Then, we introduce a model of automata, \emph{automata with oracles},  that accept configurations of \sdps and we relate this model with automata accepting configurations of $n$-CPDA as defined in Section~\ref{ssection:regSets}. Then, we introduce the notion of \emph{conditional games} and show that it is the notion that captures the winning region in the original game. Finally, we show how such games can be solved by reduction to an $(n-1)$-CPDA parity game, and from the proof we also get the expected result on the regularity of the winning region and on the existence of a winning strategy realised by a CPDA transducer.

\subsection{Abstract Pushdown Automata}

We introduce a general and abstract framework of (order-1) pushdown automata whose stack alphabet is a \emph{possibly infinite} set.

An \concept{\sdp} is a tuple $\pprocess= \anglebra{A,Q,\Delta,q_0}$ where $A$ is a (possibly infinite) set called an \concept{\sda} and containing a bottom-of-stack symbol denoted $\bot\in A$, $Q$ is a finite set of states, $q_0\in Q$ is an initial state and
$$\Delta:Q\times A\rightarrow 2^{Q\times A^{\leq 2}}$$ 
is the transition relation (here $A^{\leq 2}=\{\epsilon\}\cup A\cup A\cdot A$ are the words over $A$ of length at most $2$). 
%We additionally require that for all $a\neq\bot$, $\Delta(q,a)$ does not contain any element of the form $(q,a\bot)$ nor $(q,\bot a)$, and that $\Delta(q,\bot)$ does not contain any element of the form $(q',\epsilon)$ nor $(q',a)$ nor $(q,a b)$ with $a\neq \bot$ or $b=\bot$, \ie the bottom-of-stack symbol can only occur at the bottom of the stack, and is never popped nor rewritten.
We additionally require that for all $a\neq\bot$, $\Delta(q,a)\subseteq Q\times(A\setminus\{\bot\})^{\leq 2}$ and that $\Delta(q,\bot)\subseteq Q\times (\{\bot\}\cup\{\bot b\mid b\neq \bot\})$, \ie the bottom-of-stack symbol can only occur at the bottom of the stack, and is never popped nor rewritten.

An \concept{\sd} is a word in $\stacks=\bot(A\setminus\{\bot\})^*$. A configuration of $\pprocess$ is a pair $(q,\stack)$ with $q\in Q$ and $\stack\in\stacks$. 

\begin{remark}
In general an \sdp is not finitely describable, as the domain of $\Delta$ is infinite and no further assumption is made on $\Delta$.
\end{remark}

A \sdp $\pprocess$ induces a possibly infinite graph $G=(V,E)$, called an
\concept{\sdg,} whose vertices are
the configurations of $\pprocess$ and edges are defined by the
transition relation $\Delta$, \ie, from a vertex
$(q,\stack\cdot a)$ one has an edge to $(q',\stack\cdot u)$ whenever $(q',u)\in\Delta(q,a)$.

\begin{example}\label{example:abstractPDA}
An order-$1$ pushdown automaton is an \sdp whose stack alphabet is finite.
\end{example}

\begin{example}\label{ex:cpda}
Order-$n$ CPDA that do not create $n$-links are special cases of \sdps. Indeed, let $n>1$ and consider such an order-$n$ CPDA $\mathcal{A}=\anglebra{\Gamma,Q,\Delta,q_0}$. Let $A$ be the set of all order-$(n-1)$ stacks over $\Gamma$, and for every $p\in Q$ and $a\in A$ with $\gamma=top_1(a)$, we define $\Delta'(p,a)$ by
\begin{itemize}
\item $(q,\epsilon)\in\Delta'(p,a)$ iff $(q,\toprew{\alpha};\popn{n})\in\Delta(q,\gamma)$;
\item $(q,a'\cdot a')\in\Delta'(p,a)$ with $a'=\toprew{\alpha}(a)$ iff $(q,\toprew{\alpha};\pushn{n})\in\Delta(q,\gamma)$;
\item $(q,a')\in\Delta'(p,a)$ with $a'=op(\toprew{\alpha}(a))$ iff $(q,\toprew{\alpha};op)\in\Delta(q,\gamma)$ and $op\notin\{\popn{n},\pushn{n}\}$.
\end{itemize}
It follows from the definitions that $\mathcal{A}$ and the \sdp $\anglebra{A,Q,\Delta',q_0}$ have isomorphic transition graphs.
\end{example}

Consider now a partition $Q_\Ei\cup Q_\Ai$ of $Q$ between \Eloise
and \Abelard. It induces a natural partition $V_\Ei\cup V_\Ai$
of $V$ by setting $V_\Ei=Q_\Ei\times\stacks$ and
$V_\Ai=Q_\Ai\times\stacks$. The resulting arena
$\pggraphabs=(V_\Ei,V_\Ai,E)$ is called an \concept{\sdgg}. 
Let $\col$ be a colouring function from $Q$ to a finite set of
colours $\colors\subset \mathbb{N}$. This function is easily
extended to a function from $V$ to $\colors$ by setting
$\col((q,\varstack))=\col(q)$.
Finally, an \concept{\sdpgame} is a parity game played on such an \sdgg where the colouring function is defined as above.

%In the following we show how to construct from an \sdpgame an \emph{equivalent} simpler game. In the special case where one starts from an $n$-CPDA parity game without $n$-links (and sees it has a special instance of an \sdpgame thanks to Example~\ref{ex:cpda}) one argues that the resulting equivalent simpler game is an order $(n-1)$-CPDA parity game.

\subsection{Automata with Oracles}

We now define a class of automata to accept the winning positions in an \sdgame.
An \concept{automaton with oracles} is a tuple $\mathcal{B}=\anglebra{P,Q,A,\delta,p_0,\mathcal{O}_1\cdots\mathcal{O}_k,Acc}$ where $P$ is a finite set of control states, $Q$ is a set of input states, $A$ is a (possibly infinite) input alphabet, $p_0\in P$ is the initial state, $\mathcal{O}_i$ are subsets of $A$ (called \concept{oracles}) and $\delta: P\times \{0,1\}^k\rightarrow S$ is the transition function. Finally $Acc$ is a function from $P$ to $2^Q$. Such an automaton is designed to accept in a \emph{deterministic} way configurations of an \sdp whose \sd alphabet is $A$ and whose set of control states is $Q$.

Let $\mathcal{B}=(P,Q,A,\delta,p_0,\mathcal{O}_1\cdots\mathcal{O}_k,Acc)$ be such an automaton. With every $a\in A$ we associate a Boolean vector 
  $\pi(a)=(b_1,\cdots b_k)$ where
  \[ b_i = \begin{cases}
             1 & \text{if $a\in\mathcal{O}_i$}\\
             0 & \text{otherwise.}
           \end{cases}
  \]    

The automaton reads a configuration $C=(q,a_1 a_2\cdots a_{\ell})$ from left to right. A \concept{run} over $C$ is the sequence $r_0,\cdots, r_{\ell+1}$ such that $r_0=p_0$ and $r_{i+1}=\delta(r_i,\pi(a_i))$ for every $i=0,\cdots,\ell$. Finally the run is \concept{accepting} if and only if $q\in Acc(r_{\ell+1})$.

\begin{remark}
\label{remark:oraclefinitealphabet}

When the input alphabet is finite, it is easily seen that automata with oracles have the same expressive power as usual deterministic finite automata.
\end{remark}

We are going to
use automata with oracles to accept sets of configurations of $n$-CPDA that do not have $n$-links. As seen in Example~\ref{ex:cpda} for an order-$n$
CPDA that does not have $n$-links, we take $A$ to be the set of all order-$(n-1)$ stacks.
The sets of configurations of an order-$n$ CPDA without $n$-links accepted by automata that use as oracles regular sets of order-$(n-1)$ stacks are easily seen to be regular.

%\begin{proposition}\label{remark:oracleregular}
%Fix an order-$n$ CPDA $\pprocess$ that never creates $n$-links and consider
%an automaton $\mathcal{B}$ with oracles
%$\mathcal{O}_1,\ldots,\mathcal{O}_n$ respectively accepted by automata $\mathcal{B}_1,\ldots, \mathcal{B}_n$ (hence working on order-$(n-1)$ stacks). Let $C$ be the set of configurations of $\pprocess$ accepted by $\mathcal{B}$.
%Then we can construct an automaton (hence working on order-$n$ stacks), of size $\mathcal{O}(|\mathcal{B}| |\mathcal{B}_1| \cdots |\mathcal{B}_n|)$, accepting the set $C$.
%\end{proposition}

\begin{proposition}\label{remark:oracleregular}
Let $\pprocess$ be an order-$n$ CPDA $\pprocess$ that never creates $n$-links. 
Let $\mathcal{B}$ be an automaton with oracles $\mathcal{O}_1,\ldots,\mathcal{O}_k$ and 
assume that  each $\mathcal{O}_i$ is a regular set of $(n-1)$-stacks (and denote by $\mathcal{C}_i$ an associated automaton).
Let $C$ be the set of configurations of $\pprocess$ accepted by $\mathcal{B}$.
Then $C$ is regular and we can construct an automaton $\mathcal{C}$ (now working on order-$n$ stacks without $n$-links) of size $\mathcal{O}(n|\mathcal{B}| |\mathcal{C}_1| \cdots |\mathcal{C}_k|)$ accepting it.
\end{proposition}

%\os{Is it ok or is it too short (going to the very detailed proof will be technical and I am not sure it helps)}
\begin{proof}
It suffices to mimic the behaviour of $\mathcal{B}$ and to run in parallel the $\mathcal{C}_i$s to compute the value of the oracles. 
Hence, the automaton $\mathcal{C}$ is obtained by taking a synchronised product of $\mathcal{B}$ together with the automata $\mathcal{C}_1,\cdots,\mathcal{C}_k$. An extra component, coding a counter taking its values in $\{0,1,\dots,n\}$, is needed to keep track of the bracketing depth (initially the counter equals $0$; on reading an opening bracket $[$ the counter is incremented, on reading a closing bracket $]$ it is decremented). When the counter is equal to $0$ or $1$ one simulates $\mathcal{B}$. When the counter goes to $2$ (and as long as it differs from $1$) one simulates in parallel the $\mathcal{C}_i$s. When the counter returns to $1$ the components corresponding to the $\mathcal{C}_i$s give the value of the oracles on the last $(n-1)$-stack (\ie $b_i=1$ if and only if the control state of the $\mathcal{C}_i$s component is final). Hence the $\mathcal{B}$ component can be updated. Then the control states of the $\mathcal{C}_i$s are put back to the initial state and the next $(n-1)$-stack is processed. Finally, when the counter is again equal to $0$ (\ie the last closing bracket has been read), the control state $q$ of the input configuration is read and $\mathcal{C}$ goes to a final state if and only if the current state $p$ in the $\mathcal{B}$ component is such that $q\in Acc(p)$.
\end{proof}

\subsection{Conditional Games and Winning Regions of \SDPGAMES}

We fix an \sdp  $\pprocess=\anglebra{A,Q,\Delta,q_0}$ together with a partition $Q_\Ei\cup Q_\Ai$ of $Q$ and a colouring function $\col$ using a finite set of colours $\colors$. We denote respectively by $\pggraphabs=(V,E)$ and $\pgameabs$ the associated \sdgg and abstract pushdown parity game.

We show in Lemma~\ref{lemma:reg} below how to define an automaton with oracles that accepts \Eloise's winning region
in the game $\pgameabs$.  The oracles of this automaton are defined using the concept of conditional game.  For every subset $R\subseteq Q$ we define the \concept{conditional game induced by $R$ over $\pggraphabs$}, denoted $\pgameabs(R)$, as the game played over $\pggraphabs$ where a play $\play$ is
winning for \Eloise iff one of the following happens:
\begin{itemize}
\item In $\play$ no configuration with an empty stack, \ie of the form $(q,\bot)$, is visited, and $\play$ satisfies the parity condition.
\item In $\play$ a configuration with an empty stack is visited and the control state in the first such configuration belongs to $R$.
\end{itemize}
More formally, the set of winning plays $\WC(R)$ in $\pgameabs(R)$ is defined as follows:
$$\WC(R)=[\WC_{\col}\setminus V^*(Q\times\{\bot\})V^\omega]\; \cup \; V^*(R\times\{\bot\})V^\omega$$

For any state $q$, any stack letter $a\neq\bot$, and any subset $R\subseteq Q$ it follows from Martin's Determinacy theorem \cite{Martin75} that either \Eloise or \Abelard has a winning strategy from $(q,\bot a)$ in $\pgameabs(R)$. We denote by $\mathcal{R}(q,a)$ the set of subsets $R$ for which \Eloise wins in $\pgameabs(R)$ from $(q,\bot a)$:
$$\mathcal{R}(q,a)=\{
R\subseteq Q\mid (q,\bot a) \text{ is winning for \Eloise in } \pgameabs(R)
\}$$

Then one has the following characterisation of the set of winning positions in $\pgameabs$ in terms of automaton with oracles.

\begin{lemma}\label{lemma:reg}
Let $\pgameabs$ be an abstract pushdown parity game induced by an \sdp $\pprocess=\anglebra{A,Q,\Delta,q_0}$. Then the set of winning positions in $\pgameabs$ for \Eloise is accepted by an automaton with oracles $\mathcal{A}=(P,Q,A,\delta,p_0,\mathcal{O}_1\cdots\mathcal{O}_k,Acc)$ such that
\begin{itemize}
\item $P=2^Q$
\item $p_0=\emptyset$
\item There is an oracle $\mathcal{O}_{q,R}$ for every $q\in Q$ and $R\subseteq Q$, and $a \in \mathcal{O}_{q,R}$ iff $R\in\mathcal{R}(q,a)$ and $a\neq \bot$
\item There is an oracle $\mathcal{O}_{\bot}$ and $a \in \mathcal{O}_{\bot}$ iff $a=\bot$
\item Using the oracles, $\delta$ is designed so that:
\begin{itemize}
\item From state $\emptyset$ on reading $\bot$, $\mathcal{A}$ goes to $\{q\mid (q,\bot)\text{ is winning for \Eloise in }\pgameabs\}$
\item From state $R$ on reading $a$, $\mathcal{A}$ goes to $\{q\mid R\in\mathcal{R}(q,a)\}$
\end{itemize}
\item $Acc$ is the identity function
\end{itemize}
\end{lemma}

The proof of Lemma~\ref{lemma:reg} is a direct consequence of the following proposition.

\begin{proposition}
Let $\stack\in(A\setminus\{\bot\})^*$, $q\in Q$ and $ a\in A\setminus\{\bot\}$. Then \Eloise has a winning strategy in $\pgameabs$ from $(q,\bot \stack a)$ if and only if there exists some $R\in\mathcal{R}(q, a)$ such that $(r,\bot \stack)$ is winning for \Eloise in $\pgameabs$ for every $r\in R$.
\end{proposition}

\begin{proof}
Assume \Eloise has a winning strategy from $(q,\bot \stack a)$ in $\pgameabs$ and call it $\strat$. Consider the set $\mathcal{L}$ of all plays in $\pgameabs$ that start from $(q,\bot \stack a)$ and where \Eloise respects $\strat$. Define $R$ to be the (possibly empty) set that consists of all $r\in Q$ such that there is a play in $\mathcal{L}$ of the form $v_0\cdots v_k (r,\bot \stack) v_{k+1}\cdots$ where each $v_i$ for $0\leq i\leq k$ is of the form $(p_i,\bot \stack \varstack_i)$ for some $ \varstack_i\neq\epsilon$. In other words, $R$ consists of all states that can be reached on popping (possibly a rewriting of) $a$ for the first time in a play where \Eloise respects $\strat$.
Define a (partial) function $\tau:V\rightarrow V$ by letting $\tau(p,\bot \stack \varstack)=(p,\bot \varstack)$ for every $p\in Q$. Define a function $\tau^{-1}:V\rightarrow V$ by letting $\tau^{-1}(p,\bot \varstack)=(p,\bot \stack \varstack)$ for all $\varstack\in A^*$. 
We extend $\tau^{-1}$ as a morphism over $V^*$. 

It is easily shown that $R\in\mathcal{R}(q, a)$. Indeed a winning strategy for \Eloise in $\pgameabs(R)$ is defined as follows:
\begin{itemize}
\item if some empty stack configuration has already been visited, play any legal move, 
\item otherwise go to $\tau(\strat(\tau^{-1}(\play))$, where $\play$ is the partial play seen so far.
\end{itemize}
By definition of $\mathcal{L}$ and $R$, it easily follows that the previous strategy is winning for \Eloise in $\pgameabs(R)$, and therefore $R\in\mathcal{R}(q, a)$.

Finally, for every $r\in R$ there is, by definition of $\mathcal{L}$, a partial play $\play_r$ that starts from $(q,\bot \stack a)$, where \Eloise respects $\strat$ and that ends in $(r,\bot \stack)$. A winning strategy for \Eloise in $\pgameabs$ from $(r,\bot \stack)$ is given by $\psi(\play)=\strat(\play'_r\cdot\play)$, where $\play'_r$ denotes the partial play obtained from $\play_r$ by removing its last vertex $(r,\bot \stack)$.

Conversely, let us assume that there is some $R\in\mathcal{R}(q, a)$ such that $(r,\bot \stack)$ is winning for \Eloise in $\pgameabs$ for every $r\in R$. and denote by $\strat_r$ a winning strategy for \Eloise from $(r,\bot \stack)$ in $\pgameabs$. Let $\strat_R$ be a winning strategy for \Eloise in $\pgameabs(R)$ from $(q,\bot a)$. We define $\tau$ and $\tau^{-1}$ as in the direct implication and extend them as (partial) morphism over $V^*$. We now define a strategy $\strat$ for \Eloise in $\pgameabs$ for plays starting from $(q,\bot \stack a)$. For any partial play $\play$,
\begin{itemize}
\item if $\play$ does not contain a configuration of the form $(p,\bot \stack)$ then $\strat(\play)=\tau^{-1}(\strat_R(\tau(\play)))$;
\item otherwise let $\play = \play'\cdot(r,\bot \stack)\cdot \play''$ where $\play'$ does not contain any configuration of the form $(p,\bot \stack)$. From how $\strat$ is defined in the previous case, it is follows that $r\in R$. One finally sets $\strat(\play)=\strat_r((r,\bot \stack)\cdot\play'')$.
\end{itemize} 

It is then easy to check that $\strat$ is a winning strategy for \Eloise in $\pgameabs$ from $(q,\bot \stack a)$.
\end{proof}

\subsection{Reducing the Conditional Game}\label{section:recusingConditionalGame}

The main purpose of this section is to build a new parity game $\fgame$ whose winning region provides all the information needed to compute the sets $\mathcal{R}(q, a)$.  Moreover, in the underlying arena the vertices no longer encode stacks.

To help readability, we will use upper-case letters, \emph{e.g.} $\play$ or $\strat$, to denote objects (plays, strategies\dots) in $\pgameabs$, and lower-case letters, \emph{e.g.} $\fplay$ or $\fstrat$, to denote objects in $\fgame$.

For an infinite play $\play=v_0v_1\cdots$ in $\pgameabs$, let
$\Stepsg{\play}$ be the set of indices of positions where no
configuration of strictly smaller stack height is visited later in the
play. More formally, $\Stepsg{\play}=\{i\in\mathbb{N}\mid \forall
j\geq i\ \sh(v_j)\geq \sh(v_i)\}$, where $\sh((q,\bot a_1\cdots a_n))=n+1$ is the stack height. Note that $\Stepsg{\play}$ is always
infinite and hence induces a decomposition of the play $\play$ into infinitely many finite pieces.

In the decomposition induced by $\Stepsg{\play}$, a factor $v_i\cdots v_j$ is called a
\defin{bump} if $\sh(v_j)=\sh(v_i)$, called a \defin{Stair} otherwise (that is, if $\sh(v_j)=\sh(v_i)+1$ and $j=i+1$).

For any play $\play$ with $\Stepsg{\play}=\{n_0<n_1<\cdots\}$, we can define the sequence $(\pcol{\play}_i)_{i\geq
0}\in\mathbb{N}^{\mathbb{N}}$ by letting $\pcol{\play}_i=\min\{\col(v_k)\mid n_i\leq k\leq n_{i+1}\}$.
Obviously, this sequence fully characterises the parity condition.

\begin{proposition}\label{prop:trans_cond}
For every play $\play$, one has $\play\in\WC_\col$ iff $\liminf((\pcol{\play}_i)_{i\geq 0})$ is even.
\end{proposition}

In the sequel, we build a new parity game $\fgame$ over a new arena $\fggraph=(\widetilde{V},\widetilde{E})$.
This game
\emph{simulates} the \sdgame, in the sense that the
sequence of visited colours during a \emph{correct} simulation of a play $\play$ in $\pgameabs$ is
exactly the sequence $(\pcol{\play}_i)_{i\geq 0}$. Moreover, a play in which
a player does not correctly simulate the \sdgame is losing
for that player. We will then show how the winning region in $\fgame$ permits to compute the sets $\{ a\in A\mid R\in\mathcal{R}(q, a)\}$.

Before providing a description of the arena
$\fggraph$, let us consider the following informal description of
this simulation game. We aim at simulating a play in the \sdgame from its initial configuration $(q_0,\bot)$. In $\fggraph$ we
keep track of only the control state and the top stack symbol of
the simulated configuration.

The interesting case is when the simulated play is in a configuration with control state $p$ and top
stack symbol $a$, and the player owning $p$ wants to perform transition $(q,a'b)$, \ie go to state $q$, rewrite $a$ into $a'$ and push $b$ on top of it. 
For every strategy of \Eloise, there is a certain set of possible
(finite) prolongations of the play (consistent with her strategy) that will end with popping
$b$ (or actually a symbol into which $b$ was rewritten in the meantime) from the stack. We require \Eloise to declare a vector
$\vect{R}=(R_0,\dots,R_\maxcolor)$ of $(d+1)$ subsets of
$Q$, where $R_i$ is the set of all
states the game can be in after popping (possibly a rewriting of) $b$ along those
plays where in addition the smallest visited colour whilst (possibly a rewriting of) $b$ was on
the stack is $i$.

\Abelard has two choices. He can continue the game by pushing
$b$ onto the stack and updating the state; we call this a
\defin{pursue move}. Otherwise, he can select a set $R_i$
and pick a state $r\in R_i$, and
continue the simulation from that state $r$; we call this a
\defin{jump move}. If he does a pursue move, then he remembers the
vector $\vect{R}$ claimed by \Eloise; if later on, a transition of the form $(r,\epsilon)$ is simulated, the play goes into a sink state (either $\ttrue$ or $\ffalse$)
that is winning for \Eloise if and only if the resulting state is in
$R_\theta$ where $\theta$ is the smallest colour seen in
the current level (this information will be encoded in the control
state, reseted after each pursue move and updated after each jump
move). If \Abelard
does a jump move to a state $r$ in $R_{i}$, the currently
stored value for $\theta$ is updated to $\min(\theta,i,\col(r))$,
which is the smallest colour seen since the current stack level was
reached.

\begin{figure}[htb]
\begin{center}
\begin{tikzpicture}[>=stealth',thick,scale=1,transform shape]
\tikzstyle{Adam}=[draw,inner sep=4]
\tikzstyle{Eve}=[draw,rounded rectangle,inner sep=4]
\tikzstyle{AnyPlayer}=[inner sep=4]
\node[AnyPlayer] (current) at (0,0) {$(p,a,\vect{R},\theta)$};
\node[Eve] (nextEve) at (0,-1.5) {$(p,a',\vect{R},\theta,q,b)$};
\node[Eve] (nextEveG) at (-3.5,-1.5) {\phantom{$(p,a',\vect{R},\theta,q,b)$}};
\node[Eve] (nextEveD) at (3.5,-1.5) {\phantom{$(p,a',\vect{R},\theta,q,b)$}};

\node at (3.4,-0.6) {\textcolor{teal}{$\forall (q,a'b)\in\Delta(p,a)$}};

\node at (0,2.5) {\textcolor{teal}{$\forall (q,a')\in\Delta(p,a)$}};

\node[AnyPlayer] (int) at (0,2) {$(q,a',\vect{R},\min(\theta,\col(q) ))$};

\node[Eve] (ntrue) at (-3,1) {$(\ttrue,a)$};
\node at (-5,1.6) {\textcolor{teal}{If $\exists (r,\epsilon)\in\Delta(p,a)$ s.t. $r\in R_\theta$}};
%\node at (-4.1,1.6) {\textcolor{teal}{s.t. $r\in R_\theta$}};

\node[Eve] (nfalse) at (3,1) {$(\ffalse,a)$};
\node at (4.5,1.6) {\textcolor{teal}{If $\exists (r,\epsilon)\in\Delta(p,a)$ s.t. $r\notin R_\theta$}};
%\node at (4.1,1.6) {\textcolor{teal}{s.t. $r\notin R_\theta$}};

\node[Adam] (nextAdam) at (0,-3.5) {$(p,a',\vect{R},\theta,q,b,\vect{R'})$};
\node[Adam] (nextAdamG) at (-3.5,-3.5) {\phantom{$(p,a',\vect{R},\theta,q,b,\vect{R'})$}};
\node[Adam] (nextAdamD) at (3.5,-3.5) {\phantom{$(p,a',\vect{R},\theta,q,b,\vect{R'})$}};

\node at (3.4,-2.6) {\textcolor{teal}{$\forall \vect{R'}\in (2^{Q})^{d+1}$}};

\node[AnyPlayer] (currentJ) at (-4,-6) {$(q,b,\vect{R'},\col(q))$};
\node[AnyPlayer] (currentB) at (0,-6) {$(s,a',\vect{R},\min(\theta,i,\col(r)),i)$};
\node[AnyPlayer] (currentC) at (5,-6) {$(s,a',\vect{R},\min(\theta,i,\col(r)))$};

\node at (0,-6.6) {\textcolor{teal}{$\forall s \in R'_i$}};

\path[->] (current) edge (nextEve);\path[->] (current) edge (nextEveG);\path[->] (current) edge (nextEveD);
\path[->] (current) edge (ntrue);\path[->] (current) edge (nfalse);
\path[->] (current) edge (int);
\path[->] (nextEve) edge (nextAdam);\path[->] (nextEve) edge (nextAdamG);\path[->] (nextEve) edge (nextAdamD);

\path[->] (nextAdam) edge node[above right] {}  (currentB);
\path[->] (nextAdam) edge node[above left] {} (currentJ);

\path[->] (ntrue) edge  [loop left, loop] node[left] {} (ntrue);
\path[->] (nfalse) edge  [loop right, loop] node[right] {} (nfalse);

\path[->] (currentB) edge node[above left] {} (currentC);

\end{tikzpicture}
\end{center}
\caption{Local structure of $\fggraph$.}\label{fig:graphe_reduit}
\end{figure}

Let us now precisely describe the arena
$\fggraph$. We refer the reader to Figure~\ref{fig:graphe_reduit}.

\begin{itemize}
\item The main vertices of $\fggraph$ are those of the form
$(p, a,\vect{R},\theta)$, where $p\in Q$, $ a\in
 A$, $\vect{R}=(R_0,\dots,R_\maxcolor)\in
(2^Q)^{\maxcolor+1}$ and $\theta\in\{0,\dots,\maxcolor\}$. A vertex $(p, a,\vect{R},\theta)$ is reached when
simulating a partial play $\play$ in $\pgameabs$ such that:
\begin{itemize}

\item The last vertex in $\play$ is $(p,\stack a)$ for some $\stack\in A^*$.

\item \Eloise claims that she has a strategy to continue $\play$ in such
  a way that if $a$ (or a rewriting of it) is eventually popped, the control state
  reached after popping belongs to $R_i$, where $i$ is
  the smallest colour visited since the stack height was at least $|\stack a|$.

\item The colour $\theta$ is the smallest one since the current stack level was reached from a lower stack level.
\end{itemize}

A vertex $(p, a,\vect{R},\theta)$ is controlled by \Eloise if
and only if $p\in Q_\Ei$.

\item The vertices $(\ttrue,a)$ and $(\ffalse,a)$ are here to ensure
  that the vectors $\vect{R}$ encoded in the main vertices are correct. Both are sink vertices and are controlled by \Eloise. Vertex $(\ttrue,a)$ gets colour $0$ and vertex $(\ffalse,a)$ gets colour $1$. 
As these vertices are sinks, a play reaching $(\ttrue,a)$ is won by \Eloise whereas a play reaching $(\ffalse,a)$ is won by \Abelard.

There is a transition from some vertex
$(p, a,\vect{R},\theta)$ to $(\ttrue,a)$, if and
only if there exists a transition rule $(r,\epsilon)\in\Delta(p, a)$,
such that $r\in R_{\theta}$ (this means that $\vect{R}$ is correct with respect
to this transition rule).
Dually, there is a transition from a vertex
$(p, a,\vect{R},\theta)$ to $(\ffalse,a)$
if and only if there exists a transition rule
$(r,\epsilon)\in\Delta(p, a)$ such that $r\notin R_{\theta}$ (this means that
$\vect{R}$ is not correct with respect to this transition rule).

\item To simulate a transition rule $(q,a')\in\Delta(p, a)$, the
  player that controls $(p, a,\vect{R},\theta)$ moves to
$(q, a',\vect{R},\min(\theta,\col(q)))$. Note that the
  last component has to be updated as the
  smallest colour seen since the current stack level was reached is now $\min(\theta,\col(q))$.

\item To simulate a transition rule
  $(q,a'b)\in\Delta(p, a)$, the player that controls
  $(p, a,\vect{R},\theta)$ moves to
  $(p, a',\vect{R},\theta,q, b)$. This vertex is
  controlled by \Eloise who has to give a vector
$\vect{R'}=(R'_0,\dots,R'_\maxcolor)\in
(2^Q)^{\maxcolor+1}$ that describes the control states that can be
  reached if $ b$ (or a symbol that rewrites it later) is eventually popped. To describe this vector,
  she goes to the corresponding vertex $(p, a',\vect{R},\theta,q, b,\vect{R'})$.

Any vertex $(p, a',\vect{R},\theta,q, b,\vect{R'})$ is
controlled by \Abelard who chooses either to simulate a bump or a
stair. In the first case, he additionally has to pick the minimal colour of the
bump. To simulate a bump with minimal
colour $i$, he goes to a vertex
$(r', a',\vect{R},\min(\theta,i,\col(s)))$, for some $r'\in
R'_i$, through an intermediate vertex $(r', a',\vect{R},\min(\theta,i,\col(s)),i)$ coloured by $i$.

To simulate a stair, \Abelard goes to the vertex
$(q, b,\vect{R'},\col(q))$.

The last component of the vertex (that stores the
smallest colour seen since the currently simulated stack level was
reached) has to be updated in all those cases. After simulating a bump
of minimal colour $i$, the minimal colour is
$\min(\theta,i,\col(r'))$. After simulating a stair, this colour has to
be initialised (since a new stack level is simulated). Its value, is
therefore $\col(q)$, which is the unique colour since the (new) stack
level was reached.

\end{itemize}

The vertices of the form $(p, a,\vect{R},\theta)$ get colours $\col(p)$. 
Intermediate vertices of the form $(p, a',\vect{R},\theta,q, b)$ or $(p, a',\vect{R},\theta,q, b,\vect{R'})$ get colours $\maxcolor$ and hence, will be neutral with respect to the parity condition.

%\os{Should I say something on the size now?}

The following lemma relates the winning region in $\fgame$ with $\pgameabs$ and the conditional games induced over $\pggraphabs$.

\begin{lemma}\label{lemma:games}
For every $p_0,q\in Q$ and $a\in A$ the following holds.
\begin{enumerate}
\item Configuration $(p_{0},\bot)$ is winning for \Eloise in $\pgameabs$ if and only if $(p_{0},\bot,(\emptyset,\dots,\emptyset),\col(p_{0}))$ is winning for \Eloise in $\fgame$.
\item For every $R\subseteq Q$, $R\in\mathcal{R}(q,a)$ if and only if $(q,a,(R,\dots,R),\col(q))$ is winning for \Eloise in $\fgame$.
\end{enumerate}
\end{lemma}

\begin{remark}
Note that the above lemma is proved in~\cite[Theorem~5.1]{SerrePHD} in the case of usual pushdown automata, \ie when $A$ is finite as remarked in Example~\ref{example:abstractPDA}. A careful analysis of that proof shoes that it does not make use of the fact that $A$ is finite and therefore the proof of Lemma~\ref{lemma:games} could be skipped. Nevertheless, we give it below for completeness and also because we need a careful analysis later when dealing with the regularity of the winning configuration and when constructing a $(n-1)$-transducer realising a winning strategy (in Theorem~\ref{theo:reducingOrder} below).
\end{remark}

%\os{Ajouter une phrase expliquant que ce théorème est connu dans le cas d'un alphabet fini et qu'il y a rien à changer dans la preuve. Cependant ici on donne la preuve par soucis de complétude et on veut aussi s'en servir plus finement derrière}

The rest of the section is devoted to the proof of Lemma~\ref{lemma:games}. We mainly focus on the proof of the first item, the proof of the second one being a subpart of it.
We start by introducing some useful concept and then prove both implications.

\subsubsection{Factorisation of plays in $\pgameabs$ and in $\fgame$}\saut

Recall that for an infinite play $\play=v_0v_1\cdots$ in $\pgameabs$, 
$\Stepsg{\play}$ denotes the set of indices of positions where no
configuration of strictly smaller stack height is visited later in the
play. %More formally, $\Stepsg{\play}=\{i\in\mathbb{N}\mid \forall j\geq i\ \sh(v_j)\geq \sh(v_i)\}$, where $\sh((q,\bot a_1\cdots a_n))=n+1$. Note that $\Stepsg{\play}$ is always infinite and hence induces a factorisation of the play $\play$ into infinitely many finite pieces. 
Recall that for any play $\play$ with $\Stepsg{\play}=\{n_0<n_1<\cdots\}$, we define the sequence $(\pcol{\play}_i)_{i\geq
0}\in\mathbb{N}^{\mathbb{N}}$ by letting $\pcol{\play}_i=\min\{\col(v_k)\mid n_i\leq k\leq n_{i+1}\}$.

Indeed, for any play $\play$ with $\Stepsg{\play}=\{n_0<n_1<\cdots\}$, one can define the sequence $(\play_i)_{i\geq
0}$ by letting ${\play}_i=v_{n_i}\cdots v_{n_{i+1}}$. Note that each of the $\play_i$ is either a bump or a stair.
In the later we designate $(\play_i)_{i\geq 0}$ as the \concept{rounds factorisation} of $\play$.

%Recall that in $\fggraph$ some edges are coloured. Hence, to represent a play, we
%have to encode this information on edge colouring. We only need to encode the
%colours in $\{0,\dots,\maxcolor\}$ that appears when simulating a bump: a play will be
%represented as a sequence of vertices together with colours in
%$\{0,\dots,\maxcolor\}$ that correspond to colours appearing on edges.

For any play $\fplay$ in $\fgame$, a \defin{round} is a factor between two
visits through vertices of the form
$(p, a,\vect{R},\theta)$. We have the following possible forms for
a round.

\begin{itemize}
\item The round is of the form
  $(p, a,\vect{R},\theta)(q, a',\vect{R},\theta)$
  and corresponds therefore to the simulation of a transition $(q,a')$. We
  designate it as a \concept{trivial bump}.

\item The round is of the form 
$(p, a,\vect{R},\theta)(p, a',\vect{R},\theta,q, b)(p, a',\vect{R},\theta,q, b,\vect{R'})(s, a',\vect{R},\min(\theta,i,$\linebreak $\col(s)),i)(s, a',\vect{R},\min(\theta,i,\col(s)))$ and
corresponds therefore to the simulation of a transition $(q,a'b)$ pushing $b$
followed by a sequence of moves that ends by popping $b$ (or a rewriting of it). Moreover, $i$ is the smallest colour encountered whilst $b$ (or other stack symbol obtained by successively rewriting it) was on the stack. We designate it as a \concept{(non-trivial) bump}.

\item The round is of the form
$(p, a,\vect{R},\theta)(p, a',\vect{R},\theta,q, b)(p, a',\vect{R},\theta,q, b,\vect{R'})(q, b,\vect{R'},\col(q))$ and
corresponds therefore to the simulation of a transition $(q,a'b)$ pushing a symbol $ b$
leading to a new stack level below which the play will never go. We designate it as a \concept{stair}.
\end{itemize}

We define the \defin{colour} of a round as the smallest colour of the vertices in the round.

For any play $\fplay=v_0v_1v_2\cdots$ in $\fgame$, we consider the
subset of indices corresponding to vertices of the form
$(p, a,\vect{R},\theta)$. More precisely:
\begin{align*}
\Rounds{\fplay}=\{n\mid v_n=(p, a,\vect{R},\theta),\ p\in Q,\
 a\in A, \vect{R}\in(2^Q)^{\maxcolor+1},\
0\leq\theta\leq\maxcolor\}
\end{align*}

The set $\Rounds{\fplay}$ induces a natural factorisation of $\fplay$ into rounds. 
Indeed,  let $\Rounds{\fplay}=\{n_0<n_1<n_2<\cdots\}$, then for all $i\geq 0$, %$0\leq i<|\Rounds{\fplay}|$ 
we let $\fplay_i=v_{n_i}\cdots v_{n_{i+1}}$. 
We call the sequence $(\fplay_i)_{i\geq 0}$ the \concept{round factorisation} of $\fplay$. 
For every $i\geq 0$, $\fplay_i$ is a round and the first vertex in $\fplay_{i+1}$ equals the last one in $\fplay_i$. Moreover,
$\fplay=\fplay_0\odot\fplay_1\odot\fplay_2\odot\cdots$, where $\fplay_i\odot\fplay_{i+1}$ denotes the concatenation of $\fplay_i$ with $\fplay_{i+1}$ without its first vertex.

In order to prove both implications of Lemma~\ref{lemma:games}, we
build from a winning strategy for \Eloise in one game a winning strategy
for her in the other game. The main argument to prove that the new
strategy is winning is to prove a correspondence between the
factorisations of plays in both games.

\subsubsection{Direct implication}\saut

Assume that the configuration $(p_{0},\bot)$ is winning for \Eloise in $\pgameabs$,
and let $\strat$ be a corresponding winning strategy for her.

Using $\strat$, we define a strategy $\fstrat$ for \Eloise in
$\fgame$ from $(p_{0},\bot,(\emptyset,\dots,\emptyset),\col(p_{0}))$.
The strategy $\fstrat$ maintains as a memory a partial play $\play$ in $\pgameabs$. At the beginning $\play$ is initialised to
the vertex $(p_{0},\bot)$. We first describe $\fstrat$, and then we
explain how $\play$ is updated. Both the strategy $\fstrat$ and the
update of $\play$, are described for a round.

\noindent\textbf{Choice of the move. } Assume that the play is in some
vertex $(p, a,\vect{R},\theta)$ for $p\in Q_\Ei$. The
move given by $\fstrat$ depends on $\strat(\play)$:
\begin{itemize}
\item If $\strat(\play)=(r,\epsilon)$, then \Eloise goes to $(\ttrue,a)$ (Proposition~\ref{prop:par_dir_dep_paritexp} will  prove that this move is always possible).
\item If $\strat(\play)=(q,a')$, then \Eloise goes to $(q,a',\vect{R},\min(\theta,\col(q)))$.
\item If $\strat(\play)=(q,a'b)$, then \Eloise goes to $(p, a',\vect{R},\theta,q, b)$.
\end{itemize}

In this last case, or in the case where $p\in Q_\Ai$ and \Abelard goes to $(p, a',\vect{R},\theta,q, b)$, we also have to explain
how \Eloise behaves from
$(p, a',\vect{R},\theta,q, b)$. She has to provide a
vector $\vect{R'}\in (2^Q)^{\maxcolor+1}$ that describes which states
can be reached if $b$ (or its successors by top rewriting) is eventually popped, depending on the smallest visited colour in the
meantime. In order to define $\vect{R'}$, \Eloise considers the set of
all possible continuations of $\play\cdot(q, \stack a' b)$ (where
$(p, \stack a)$ denotes the last vertex of $\play$) where she
respects her strategy $\strat$. For each such play, she checks whether some
configuration of the form $(r', \stack a')$ is visited after $\play\cdot
(q, \stack a' b)$, that is if the stack level of $b$ is eventually left. If it
is the case, she considers the first configuration $(r', \stack a')$
appearing after $\play\cdot (q, \stack a' b)$ and the smallest
colour $i$ since
$b$ and (possibly) its successors by top-rewriting were on the stack.
For every $i\in\{0,\dots, d\}$, $R'_i$ is exactly the set of states
$r'\in Q$ such that the preceding case happens. 
More formally, 
%$$
\begin{multline*}
R'_i=\{r'\mid \exists\ \play\cdot(q, \stack a' b)v_0\cdots v_k(r', \stack a')\cdots\text{ play in } \pgameabs \text{ where \Eloise respects } \strat \text{ and} \\ \text{s.t. } |v_j|\geq| \stack a'b|,\ \forall j=0,\dots,k  \text{ and }\min(\{\col(v_j)\mid j=0,\dots,k\}\cup\{\col(q)\})=i\}
\end{multline*}
%$$
Finally, we let $\vect{R'}=(R'_0,\dots,R'_\maxcolor)$ and \Eloise moves to $(p, a',\vect{R},\theta,q, b,\vect{R'})$.

\noindent\textbf{Update of $\play$. } The memory $\play$ is updated after
each visit to a vertex of the form $(p, a,\vect{R},\theta)$.
We have three cases depending on the kind of the last round:

\begin{itemize}
\item The round is a trivial bump and therefore a $(q,a')$
  transition was simulated. Let $(p, \stack a)$ be the last vertex in
  $\play$, then the updated memory is $\play\cdot(q, \stack a')$.

\item The round is a bump, and therefore a bump
  of colour $i$ (where $i$ is the colour of the round) starting with some
  transition $(q,a'b)$ and ending in a state $r'\in R'_i$ was simulated. Let $(p, \stack a)$ be the last vertex in
  $\play$. Then the memory becomes $\play$ extended by
  $(q, \stack a' b)$ followed by a sequence of moves, where \Eloise
  respects $\strat$, that ends by popping $b$ and reaches
  $(r', \stack a')$ whilst visiting $i$ as smallest colour. By definition of $R'_i$ such a sequence of moves always exists.

\item The round is a stair and therefore we have simulated a
  $(q,a' b)$ transition. If $(p, \stack a)$ denotes the last
  vertex in $\play$, then the updated memory is $\play\cdot (q, \stack a' b)$.

\end{itemize}

Therefore, with any partial play $\fplay$ in $\fgame$ in which \Eloise
respects her strategy $\fstrat$, is associated a partial play $\play$ in
$\pgameabs$. An immediate induction shows that \Eloise respects $\strat$ in
$\play$. The same arguments works for an infinite play $\fplay$, and the
corresponding play $\play$ is therefore infinite, starts from
$(p_{0},\bot)$ and \Eloise respects $\strat$ in that play. Therefore it is
a winning play.

The following proposition is a direct consequence of how $\fstrat$ was defined.

\begin{proposition}\label{prop:par_dir_dep_paritexp}
Let $\fplay$ be a partial play in $\fgame$ that starts from
$(p_{0},\bot,(\emptyset,\dots,\emptyset),\col(p_{0}))$,
ends in a vertex of the form $(p, a,\vect{R},\theta)$,
and where \Eloise respects $\fstrat$. Let $\play$ be the partial play associated with $\fplay$
built by the strategy $\fstrat$. Then the following holds:
\begin{enumerate}
\item $\play$ ends in a vertex of the form $(p, \stack a)$ for some $\stack \in A^*$.

\item $\theta$ is the smallest visited colour in $\play$ since
  $a$ (or a symbol that was later rewritten as $a$) has been pushed.

\item Assume that $\play$ is extended, that \Eloise keeps respecting
  $\strat$ and that the next move after $(p, \stack a)$ is to some
  vertex $(r, \stack)$. Then $r\in R_\theta$.
\end{enumerate}
\end{proposition}

Proposition~\ref{prop:par_dir_dep_paritexp} implies that the
strategy $\fstrat$ is well defined when it provides a move to some 
$(\ttrue,a)$. Moreover, one can deduce that, if \Eloise respects $\fstrat$, no vertex of the form $(\ffalse,a)$ is reached.

For plays that never reach a sink vertex $(\ttrue,a)$, using the definitions of $\fggraph$ and $\fstrat$, we easily deduce the
following proposition.

\begin{proposition}\label{prop:colorsFact}
Let $\fplay$ be a play in $\fgame$ that starts from
$(p_{0},\bot,(\emptyset,\dots,\emptyset),\col(p_{0}))$,
and where \Eloise respects $\fstrat$. Assume that $\fplay$ never visits $\ttrue$, let $\play$ be the associated
play built by the strategy $\fstrat$, and let $(\play_i)_{i\geq 0}$ be its rounds factorisation. Let $(\fplay_i)_{i\geq 0}$ be
the rounds factorisation of $\fplay$. Then, for every $i\geq 0$ the
following hold:
\begin{enumerate}
\item $\fplay_i$ is a bump if and only if $\play_i$ is a bump

\item $\fplay_i$ has colour $\pcol{\play}_i$.
\end{enumerate}
\end{proposition}

Now consider a play
$\fplay$ in $\fgame$ starting from
$(p_{0},\bot,(\emptyset,\dots,\emptyset),$ $\col(p_{0}))$
where \Eloise respects $\fstrat$. Either $\fplay$ loops in some $(\ttrue,a)$ (hence, is won by \Eloise). 
Or, thanks to Proposition~\ref{prop:colorsFact} the sequence of visited colours in $\fplay$ is $(\pcol{\play}_i)_{i\geq 0}$ for the corresponding play $\play$
in $\pgameabs$. Hence, using Proposition~\ref{prop:trans_cond} we conclude that $\fplay$ is winning if and only if $\play$ is winning; as $\play$ is winning for \Eloise, it follows that $\fplay$ is winning for her as well.

\subsubsection{Converse implication}\label{proof:theogamesconverse}\saut

%\os{Make the last sentence more precise on strategies}
First note that in order to prove the converse implication one could follow the same approach as for the direct implication by considering now the point of view of \Abelard. Nevertheless the proof we give here starts from a winning strategy for \Eloise in $\fgame$ and constructs a strategy for her in $\pgameabs$: this induces a more involved proof but has the advantage of leading to an effective construction of a winning strategy for \Eloise in $\pgameabs$ if one has an effective winning strategy for her in $\fgame$.

Assume now that \Eloise has a winning strategy $\fstrat$ in $\fgame$
from $(p_{0},\bot,(\emptyset,\dots,\emptyset),\col(p_0))$.
Using $\fstrat$, we build a strategy $\strat$ for \Eloise in
$\pgameabs$ for plays starting from $(p_{0},\bot)$.

The strategy $\strat$ maintains as a memory a partial play $\fplay$ in $\fgame$, that is an element in $\widetilde{V}^*$. At the beginning $\fplay$ is initialised to $(p_0,\bot,(\emptyset,\dots,\emptyset),\col(p_0))$.

For any play $\play$ where \Eloise respects $\strat$ the following will hold.
\begin{itemize}
\item $\fplay$ is a play in $\fgame$ that starts from $(p_{0},\bot,(\emptyset,\dots,\emptyset),\col(p_0))$ and where \Eloise respects her winning strategy $\fstrat$.
\item The last vertex of $\fplay$ is some $(p, a,\vect{R},\theta)$ if and only if the current configuration in $\play$ is of the form $(p, \stack a)$.
\item If \Eloise keeps respecting $\strat$, and if $ a$ (or a symbol that rewrites it later) is eventually popped the configuration reached will be of the form $(r, \stack)$ for some $r\in R_i$, where $i$ is the smallest visited colour since $ a$ (or some symbol that was later rewritten as $ a$) was on the stack.
\end{itemize}

Note that initially the previous invariants trivially hold.

In order to describe $\strat$, we assume that we are in some
configuration $(p, \stack a)$ and that the last vertex of $\fplay$ is some $(p, a,\vect{R},\theta)$. We first
describe how \Eloise plays if $p\in Q_\Ei$, and then we explain how
$\fplay$ is updated.

\noindent\textbf{Choice of the move.} Assume that $p\in Q_\Ei$. Then the move given by $\strat$ depends on $\fstrat(\fplay)$.
\begin{itemize}
\item If $\fstrat(\fplay)=(q, a',\vect{R},\min(\theta,\col(q)))$, \Eloise plays transition $(q,a')$.
\item If $\fstrat(\fplay)=(p, a',\vect{R},\theta,q, b)$, then \Eloise applies plays transition $(q,a'b)$.
\item If $\fstrat(\fplay)=(\ttrue,a)$, \Eloise plays transition $(r,\epsilon)$ for some state $r\in R_\theta$. Lemma~\ref{ini:lemma:games:ReturningSets_paritexp} will prove that such an $r$ always exists.
\end{itemize}

\noindent\textbf{Update of $\fplay$.} 
The memory $\fplay$ is updated after each move (played by any of the two players). We have several cases depending on the last transition.
\begin{itemize}
\item If the last move was from $(p, \stack a)$ to $(q, \stack a')$ then the updated memory is $\fplay\cdot(q, a',\vect{R},\min(\theta,\col(q)))$.
\item If the last move was from $(p, \stack a)$ to $(q, \stack a' b)$, let $(p, a',\vect{R},\theta,q, b,\vect{R'})=\fstrat(\fplay\cdot(p, a',\vect{R},\theta,q, b))$.
%Intuitively, $\vect{R'}$ describes which states \Eloise can force a play to reach if $b$ is eventually popped. 
Then 
the updated memory is $\fplay\cdot(p, a',\vect{R},\theta,q, b)\cdot(p, a',\vect{R},\theta,q, b,\vect{R'})\cdot(q, b,\vect{R'},\col(q))$.
\item If the last move was from $(p, \stack a)$ to $(r, \stack)$ the update of $\fplay$ is as follows. 
One backtracks in $\fplay$ until one finds a configuration of the form
$(p', a',\vect{R'},\theta',p'', a'',\vect{R})$ that is not immediately followed by a vertex of the form $(s, a'',\vect{R},\theta'',i)$. This configuration is therefore in the stair that simulates the pushing
of $ a''$ onto the stack (here if $ a''\neq a$ it simply means that $ a''$ was later rewritten as $ a$). Call $\fplay'$ the prefix of $\fplay$ ending in this configuration. The updated memory is $\fplay'\cdot (r,a',\vect{R'},\min(\theta',\theta,\col(r)),\theta)\cdot (r,a',\vect{R'},\min(\theta',\theta,\col(r)))$.
Formally, write $\fplay=\fplay_0\odot\fplay_1\odot\cdots\odot\fplay_k$ where $(\fplay_i)_{0\leq i\leq k}$ is the round factorisation of $\fplay$. Let $h\leq k$ be the largest integer such that $\fplay_h$ is a stair and let $\fplay_h = (p', a',\vect{R'},\theta')(p', a',\vect{R'},\theta',p'', a'')(p', a',\vect{R'},\theta',p'', a'',\vect{R})(p'',a'',\vect{R},\col(p''))$.
Define $\fplay_h' = (p', a',\vect{R'},\theta')(p', a',\vect{R'},\theta',p'', a'')(p', a',\vect{R'},\theta',p'', a'',\vect{R})(r,a',\vect{R'},\min(\theta',\theta,\col(r)),\theta)\cdot (r,a',\vect{R'},\min(\theta',\theta,\col(r)))$.
Then the updated memory is $\fplay_1\odot\fplay_2\odot\cdots\odot\fplay_{h-1}\odot \fplay_h'$.
\end{itemize}

The following lemma gives the meaning of the information stored
in $\fplay$.

\begin{lemma}\label{ini:lemma:games:ReturningSets_paritexp}
Let $\play$ be a partial play in $\pgameabs$, where \Eloise respects
$\strat$, that starts from $(p_{0},\bot)$ and ends in a configuration $(p, \stack a)$. We have the
following facts:

\begin{enumerate}

\item The last vertex of $\fplay$ is of the form $(p, a,\vect{R},\theta)$ with
$\vect{R}\in(2^Q)^{\maxcolor+1}$ and $0\leq\theta\leq\maxcolor$.

\item $\fplay$ is a partial play in $\fgame$ that starts
from $(p_{0},\bot,(\emptyset,\dots,\emptyset),\col(p_0))$,
that ends with $(p, a,\vect{R},\theta)$ and where
\Eloise respects $\fstrat$.

\item $\theta$ is the smallest colour visited since $ a$ (or some symbol that was later rewritten as $ a$) was
pushed.

\item If $\play$ is extended by some move
that pops $ a$, the configuration $(r, \stack)$ that is reached
is such that $r\in R_\theta$.
\end{enumerate}
\end{lemma}

\begin{proof}
We first note that the last point is a consequence of the second and third points. Indeed, assume that the next move after $(p, \stack a)$ is to play a transition
$(r,\epsilon)\in\Delta(p, a)$. The second point implies that
$(p, a,\vect{R},\theta)$ is winning for \Eloise in
$\fgame$. If $p\in Q_\Ei$, by definition of $\strat$, there is
some edge from that vertex to $(\ttrue,a)$, which means that
$r\in R_\theta$ and allows us to conclude. If $p\in Q_\Ai$, note that there is no
edge from $(p, a,\vect{R},\theta)$ (winning position
for \Eloise) to the losing vertex $(\ffalse,a)$. Hence we
conclude the same way.

Let us now prove the other points by induction on $\play$. Initially, they trivially hold. Now assume that the
result is proved for some play $\play$, and let $\play'$ be an
extension of $\play$. We have two cases, depending on how $\play'$
extends $\play$:

\begin{itemize}
\item $\play'$ is obtained by applying a transition of the form $(q,a')$ or
$(q,a'b)$. The result is trivial in that case.

\item $\play'$ is obtained by applying a transition of the form $(r,\epsilon)$. Let
$(p, \stack a)$ be the last configuration in $\play$, and let
$\vect{R}$ be the last vector component in the last vertex of $\fplay$ when
in configuration $(p, \stack a)$. By the induction
hypothesis, it follows that $\play'=\play\cdot(r, \stack)$ with
$r\in R_\theta$. Considering how $\fplay$ is updated, and
using the fourth point, we easily deduce that the new memory $\fplay$ is as desired.
\end{itemize}
\end{proof}

Actually, we easily deduce a more precise result.

\begin{lemma}\label{lemme:toto_paritexp}
Let $\play$ be a partial play in $\pgameabs$ starting from
$(p_0,\bot)$ and where \Eloise respects $\strat$ and let $(\play_i)_{i\geq 0}$ be its rounds factorisation. Let
$(\fplay_i)_{i=0,\dots,k}$ be the rounds factorisation of $\fplay$.
Then the following holds for every $i\geq 0$.
\begin{itemize}
\item $\fplay_i$ is a bump if and only if ${\play}_i$ is a bump.

\item $\fplay_i$ has colour $\pcol{\play}_i$.
\end{itemize}
\end{lemma}

Both lemmas~\ref{ini:lemma:games:ReturningSets_paritexp} and
\ref{lemme:toto_paritexp} are for partial plays. A version for
infinite plays would allow us to conclude. Let $\play$
be an infinite play in $\pgameabs$. We define an
infinite version of $\fplay$ by considering the limit of the $(\fplay_i)_{i\geq 0}$ where $\fplay_i$ is the
memory after the $i$ first moves in $\play$. See Footnote~\ref{footnote:limit} on page \pageref{footnote:limit} for a similar construction. 
It is easily seen that such a limit
always exists, is infinite and corresponds to a play won by \Eloise in $\fgame$.
Moreover the results of Lemma~\ref{lemme:toto_paritexp} remain true.

Let $\play$ be a play in $\pgameabs$ with initial
vertex $(p_0,\bot)$, and where \Eloise respects $\strat$,
and let $\fplay$ be the associated play in $\fgame$.
Therefore $\fplay$ is won by \Eloise. Using Lemma
\ref{lemme:toto_paritexp} and Proposition~\ref{prop:trans_cond},
we conclude, as in the direct implication that $\play$ is
winning.

\subsection{Main Result}

Following Example~\ref{ex:cpda} we see an $n$-CPDA that does not create $n$-links as an \sdp and we apply the construction of Section~\ref{section:recusingConditionalGame}. We argue that the resulting game $\fgame$ is associated with an $(n-1)$-CPDA, which leads the following result.

\begin{theorem}\label{theo:reducingOrder}
For any $n$-CPDA $\pprocesslf=\anglebra{\Gammalf,\Qlf,\Deltalf,q_{0,\lf}}$ that \emph{does not create $n$-links} and any associated parity game $\pgamelf$, one can construct an $(n-1)$-CPDA $\fpprocess=\anglebra{\Gammaf,\Qf,\Deltaf,\qinif}$ and an associated parity game $\fgame$ such that the following holds.
\begin{itemize} 
	\item $(q_{0,\lf},\bot_{n})$ is winning for \Eloise in $\pgamelf$ if and only if $(\qinif,\bot_{n-1})$ is winning for \Eloise in $\fgame$  .
	\item If the set of winning configurations for \Eloise in $\fgame$ is regular, then the set of winning configurations for \Eloise in $\pgamelf$ is regular as well.
	\item If there is an $(n-1)$-CPDA transducer $\Sf$ synchronised with $\fpprocess$ realising a well-defined winning strategy for \Eloise in $\fgame$ from $(\qinif,\bot_{n-1})$, then one can effectively construct an $n$-CPDA transducer $\Slf$ synchronised with $\pprocesslf$ realising a well-defined winning strategy for \Eloise in $\pgamelf$ from the initial configuration $(q_{0,\lf},\bot_{n})$.
\end{itemize}
\end{theorem}

\begin{proof}

Following Example~\ref{ex:cpda}, $\pprocesslf$ can be seen as an \sdp hence, we can apply the construction of Section~\ref{section:recusingConditionalGame}. We claim that the resulting game $\fgame$ is associated with an $(n-1)$-CPDA. 

Indeed, one simply needs to consider how the graph $\fgraph$ is defined and make the following observations concerning the local structure given in Figure~\ref{fig:graphe_reduit} when $\pgame$ is played on the transition graph of an $n$-CPDA that does not create links.
\begin{enumerate}
\item For every vertex of the form $(p,a,\vect{R},\theta)$, $(\ttrue,a)$,$(\ffalse,a)$,$(p,a,\vect{R},\theta,q,b)$, $(p,a,\vect{R},\theta,q,b,\vect{R'})$ or $(s,a,\vec{R},\theta',i)$, $a$ and $b$ are $(n-1)$-stacks.
\item For every vertex of the form $(p,a,\vect{R},\theta,q,b)$ or $(p,a,\vect{R},\theta,q,b,\vect{S})$, one has $a=b$.
\end{enumerate}
This implies that any vertex in $\fgraph$ can be seen as a pair formed by a state in a finite set and an $(n-1)$-stack. Then one concludes the proof by checking that the edge relation is the one of an $(n-1)$-CPDA.%

Therefore, the first point follows from Lemma~\ref{lemma:games} and the second one follows by combining Lemma~\ref{lemma:reg} with Proposition~\ref{remark:oracleregular} and Lemma~\ref{lemma:games}.

We now turn to the third point and therefore assume that there is an $(n-1)$-CPDA transducer $\Sf$ synchronised with $\fpprocess$ realising a well-defined winning strategy $\fstrat$ for \Eloise in $\fgame$ from $(\qinif,\bot_{n-1})$. We argue that the strategy $\strat$ constructed in the proof of Lemma~\ref{lemma:games} can be realised, when $\pgameabs$ is obtained from an $n$-CPDA $\pprocesslf$ that does not create $n$-links, by an $n$-CPDA transducer $\Slf$ synchronised with $\pprocesslf$.

For this, let us first have a closer look at $\strat$. The key ingredient in $\strat$ is the play $\fplay$ in $\fgame$, and the value of $\strat$ uniquely depends on $\fstrat(\fplay)$. In particular, if $\fstrat$ is realised by an $(n-1)$-CPDA transducer $\Sf$, it suffices to know the configuration of $\Sf$ after reading $\fplay$ in order to define $\strat$. We claim that it can be computed by an $n$-CPDA transducer $\Slf$ (synchronised with $\pprocesslf$); the hard part being to establish that such a device can update correctly its memory.

Let $\fplay = v_0v_1\cdots v_\ell$ and let $r_{\fplay}=(p_0,\stack_0)(p_1,\stack_1)\cdots (p_\ell,\stack_\ell)$ be the run of $\Sf$ associated with $\fplay$, \ie after having played $v_0\cdots v_k$, $\Sf$ is in configuration $(p_k,\stack_k)$. Denote by $\Last(r_{\fplay})$ the last configuration of $r_\fplay$, \ie $(p_\ell,\stack_\ell)$. To define $\strat$, $\Last(r_\fplay)$ suffices but of course, in order to update $\Last(r_\fplay)$, we need to recall some more configurations from $r_\fplay$. In the case where the last transition applies an order-$k$ stack operation with $k<n$ (\ie it is neither $\popn{n}$ nor $\pushn{n}$), then the update is simple, as it consists in simulating one step of $\Sf$. 
If the last stack operation is $\pushn{n}$ then the update of $\fplay$ consists in adding three vertices and the corresponding update of $r_\fplay$ is simple (as the only operation on the $(n-1)$-stack is to rewrite the $\topn{1}$-element).  
If the last stack operation is $\popn{n}$ one needs to backtrack in $\fplay$ (hence in $r_\fplay$): the backtrack is to some $v_k$ with $k$ maximal such that $v_k$ is of the form $(p',a',\vect{R'},\theta',p'',a'',\vect{R})$ and $v_{k+1}=(p'',a'',\vect{R},\col(p''))$. Once $v_k$ has been found, the update is fairly simple for both $\fplay$ and $r_{\fplay}$ (one simply extends the remaining prefix of $\fplay$ by two extra vertices whose stack content is unchanged compared with the one in $v_k$).

Define the following set of indices where $\fplay = v_0v_1\cdots v_\ell$
$$\Extremal(\fplay)=\{h\mid \text{$v_h$ is of the form $(p',a',\vect{R'},\theta',p'',a'',\vect{R})$ and $v_{h+1}=(p'',a'',\vect{R},\col(p''))$}\}\cup\{\ell\}$$
Note that after a partial play $\play$ the cardinality of $\Extremal(\fplay)$ is equal to the height of the stack in the last configuration of $\play$. 

For any partial play $\play$ in $\pgamelf$ define the following $n$-stack (note that it does not contain any $n$-link)
$$\Memory(\play)=\mksk{\stack'_{k_1}\stack'_{k_2}\cdots \stack'_{k_h}}$$ where we let
\begin{itemize}
\item $\Extremal(\fplay)=\{k_1<\cdots< k_h\}$, $\fplay$ being the memory associated with $\play$ as in the proof of Lemma~\ref{lemma:games};
\item $\stack'_j$ is the $(n-1)$-stack obtained from $\stack_j$ (recall that $(p_j,\stack_j)$ denotes the $j$-th configuration of $r_{\fplay}$) by appending $p_j$ to its $\topn{1}$-symbol (\ie we work on an enriched stack alphabet).
\end{itemize}

Note that $\Last(r_\fplay)$ is essentially $\topn{1}(\Memory(\play))$ as the only difference is that now the control state is stored in the stack. Moreover $\Memory(\play)$ can easily be updated by an $n$-CPDA transducer: for the case of a transition involving an order-$k$ stack operation with $k<n$ one simulates $\Sf$ on $\topn{1}(\Memory(\play))$; for the case of a transition involving a $\pushn{n}$ one first simulates $\Sf$ on $\topn{1}(\Memory(\play))$ (as one may do a $\toprew{}$ before $\pushn{n}$) and then makes a $\pushn{n}$ to duplicate the topmost $(n-1)$-stack in $\Memory(\play)$; finally, for the case of a $\popn{n}$, one simply needs to do a $\popn{n}$ in $\Memory(\play)$ to backtrack and then update the control state. This is how we define $\Slf$\footnote{Technically speaking, if we impose that a transition of $\Slf$ does a $\toprew{}$ (or $\id$) followed by another stack operation, we may not be able to do the update of the stack after doing a $\popn{n}$. However, we can use the same trick as the one used to define $\pprocessrk$, \ie we postpone the $\toprew{}$ action to the next transition (see Remark~\ref{rk:rankAwareTopRew}).}.

The fact that $\Slf$ is synchronised with $\pprocesslf$ comes from the definition of how $\Slf$ behaves when the transition in $\pprocesslf$ involves a $\popn{n}$ or a $\pushn{n}$, and for the other cases it follows from the initial assumption of $\Sf$ being synchronised with $\fpprocess$.
\end{proof}

\os{Il let this remark but I think we do not use it later on…}
\begin{remark}\label{rk:popnsuffices}
When applying the general construction of Section~\ref{section:recusingConditionalGame} to an $n$-CPDA $\pprocesslf$ that does not create links, we can safely enforce the following extra constraint on the vectors $\vect{R}$ and $\vect{S}$: they should be element in $(2^{Q_{\lf}^{\popn{n}}})^{d+1}$ where we let $Q_{\lf}^{\popn{n}}$ denote the set of control states of $\pprocesslf$ from which a $\popn{n}$ operation can be performed. Indeed, the various component of such vectors aims at representing set of states reachable by doing a $\popn{n}$. This is important later in the overall complexity for Theorem~\ref{theorem:main}.
\end{remark}

\subsection{Complexity}

If we summarise, the overall blowup in the transformation from $\pgamelf$ to $\fgame$ given by Theorem~\ref{theo:reducingOrder} is as follows.

\begin{proposition}\label{proposition:complexity-step3}
Let $\pprocesslf$ and $\fpprocess$ be as in Theorem~\ref{theo:reducingOrder}. Then the set of states of $\fpprocess$ has size $\mathcal{O}( 2^{2|\colors||\Qlf|})$ and the stack alphabet of $\fpprocess$ has size $\mathcal{O}(|\Gammalf|)$.
Finally, the set of colours used in $\pgamelf$ and $\fgame$ are the same.
\end{proposition}

\begin{proof}
By construction.% The $2^{2(d+1)|Q_{\lf}^{\popn{n}}|}$ comes from the $\vect{R}$ and $\vect{S}$ component in vertices of the form $(p,a,\vect{R},\theta,q,\vect{S})$ where we put the extra constraint that the states appearing in the $\vect{R}$ or $\vect{S}$ components should be one reachable by a $\popn{n}$ transition (see Remark~\ref{rk:popnsuffices}).
\end{proof}

%% file: Summary.tex
\section{Proof of Theorem~\ref{theorem:main} and Complexity}\label{section:summary}

The proof of Theorem~\ref{theorem:main} consists in combining theorems \ref{lemma:rank-aware}, \ref{theo:outmost} and \ref{theo:reducingOrder}. Indeed, starting from an $n$-CPDA, we apply Theorem~\ref{lemma:rank-aware} to obtain a rank-aware $n$-CPDA, then Theorem~\ref{theo:outmost} to remove the order-$n$ links, and finally Theorem~\ref{theo:reducingOrder} to obtain an $(n-1)$-CPDA.  By $(n-1)$ successive applications of these three results, we end-up with a $1$-CPDA parity game. If we apply to this latter (pushdown) game the construction of Section~\ref{section:recusingConditionalGame} we end up with a game on a finite graph. Solving this game and following the chain of equivalences provided by theorems \ref{lemma:rank-aware}, \ref{theo:outmost} and \ref{theo:reducingOrder} concludes the proof.

Concerning complexity, one step of successive application of the construction in theorems \ref{lemma:rank-aware}, \ref{theo:outmost} and \ref{theo:reducingOrder} results in an $(n-1)$-CPDA with a state set of size $\mathcal{O}( 2^{2|Q|(|\colors|+3)^{n+5}})$, a stack alphabet of size $\mathcal{O}(|\Gamma|^2 \cdot 2^{|Q|(|\colors|+1)^{n+5}})$ and an unchanged number of colours.
Indeed,
\begin{itemize}
\item by Proposition~\ref{proposition:complexity-step1} one has $|Q_\rk| = \mathcal{O}(|Q|\cdot(|\colors|+1)^{n+3})$ and $|\Gamma_\rk| = \mathcal{O}(|\Gamma|\cdot(|\colors|+1)^{2n+5})$;
\item by Proposition~\ref{proposition:complexity-step2} one has $|Q_\lf| = \mathcal{O}(|\Qrk|\cdot(|\colors|+3)) = \mathcal{O}(|Q|\cdot(|\colors|+3)^{n+4})$ and\linebreak $|\Gamma_\lf| = \mathcal{O}(|\Gammark|^2 \cdot 2^{|\Qrk||\colors|})= \mathcal{O}(|\Gamma|^2 \cdot(|\colors|+1)^{4n+10}\cdot 2^{|Q|(|\colors|+1)^{n+4}})=\mathcal{O}(|\Gamma|^2 \cdot 2^{|Q|(|\colors|+1)^{n+5}})$; 
\item and finally, by Proposition~\ref{proposition:complexity-step3}, one has 
	$|\Qf|=\mathcal{O}( 2^{2|\colors||\Qlf|}) = \mathcal{O}( 2^{2|Q|(|\colors|+3)^{n+5}}) $ and\linebreak $|\Gammaf|=\mathcal{O}(|\Gammalf|)=\mathcal{O}(|\Gamma|^2 \cdot 2^{|Q|(|\colors|+1)^{n+5}})$.
\end{itemize}

If one lets, for a constant $K$, $\expn{h}^K$ be the function defined by $\expn{0}^K(x) = x$ for all $x$ and $\expn{h+1}^K(x) = 2^{K\expn{h}^K(x)}$, we conclude that the $1$-CPDA obtained after $(n-1)$ successive applications of the three reductions has
\begin{itemize}
\item a state set of size $\mathcal{O}(\expn{n-1}^{2(|\colors|+3)^{n+5}}(|Q|))$ and 
\item a stack alphabet of size $\mathcal{O}(|\Gamma|^{2(n-1)}\cdot \expn{n-1}^{(|\colors|+1)^{n+5}}(|Q|))$.
\end{itemize}

Solving this latter game can be done by reducing it using the construction of Section~\ref{section:recusingConditionalGame} which leads to solve a parity game on a finite graph with 
$\mathcal{O}(
\expn{n}^{2(|\colors|+3)^{n+5}}(|Q|)\cdot
(|\Gamma|^{2(n-1)}\cdot \expn{n-1}^{(|\colors|+1)^{n+5}}(|Q|))^2
)$ vertices. Solving this game can be achieved in time $\mathcal{O}(N^{|\colors|})$ where $N$ denotes the number of vertices. 
% Finally the finite game we obtain is a parity game with $|\colors|$ colours on a graph with $\mathcal{O}(|\Gamma|^{4(n-1)} \expn{n}(|Q|(d+1)^{n+6}))$ vertices. This latter game can be solve in \oschanged{$\mathcal{O}([|\Gamma|^{4(n-1)d} \expn{n}(|Q|(d+1)^{n+7})])$}
%
Hence, the overall complexity of deciding the winner in an $n$-CPDA parity game is:
\begin{itemize}
\item $n$-times exponential in the number of states of the CPDA;
\item {$n$}-times exponential in the number of colours;
\item polynomial in the size of the stack alphabet of the CPDA.
\end{itemize}

\os{Should we mention that the crucial parameter is the stack alphabet if one is interested in using this to model-check a higher-order recursive program and that the algorithm is therefore FPT…?}

Regarding lower bound, the problem is $n$-\exptime-hard. In fact, hardness already holds when one considers reachability condition (\emph{i.e.} does the play eventually visit a configuration with a final control state?) for games generated by higher-order pushdown automata (\emph{i.e.} CPDA that never use $\collapse$). 
A self-contained proof of this result was established by Cachat and  Walukiewicz, but is fairly technical \cite{CachatWalukiewicz07}.

  Here we sketch a much simpler proof of this result that relies on the
  following well-known result: checking emptiness of a nondeterministic order-$n$
  higher-order pushdown automaton is an $(n-1)$-\exptime-complete problem \cite{Engelfriet91} (here one uses higher-order pushdown automata as word acceptors)\footnote{
  The following result is also proved in \cite{Engelfriet91}:
  checking emptiness of an alternating order-$n$ higher-order pushdown automaton is an $n$-EXPTIME complete problem. Nevertheless, note
  that this result does not directly imply hardness for games on
  higher-order pushdown graphs. Indeed, in general it is 
  \emph{more difficult} to check emptiness for an alternating device than
  to solve a reachability game on the corresponding class of
  graphs: for instance, solving a reachability game on a
    finite graph is in $P$ while checking emptiness for an alternating
    automata on finite word (even if one considers a $1$-letter
    alphabet) is PSPACE-complete; the problems are trivially
  equivalent only when considering infinite words on a single letter
  alphabet.
}. 
  Trivially, this result is still true
  if we assume that the input alphabet is reduced to a single
  letter. 
  Now consider an order-$(n+1)$ nondeterministic higher-order pushdown
  automaton $\pprocess$ whose input alphabet is reduced to a single
  letter. The language accepted by $\pprocess$ is non-empty if and
  only if there is a path from the initial configuration of
  $\pprocess$ to a final configuration of $\pprocess$ in the
  transition graph $G$ of $\pprocess$. Equivalently, the language
  accepted by $\pprocess$ is non-empty if and only if \Eloise wins
  the reachability game $\pgame$ over $G$ where she controls all
  vertices (and where the play starts from the initial configuration
  of $\pprocess$ and where final vertices are those corresponding to
  final configurations of $\pprocess$). Now, consider the reduction
  used to prove Theorem \ref{theorem:main} and apply it to
  $\pgame$. As $\pprocess$ does not use links, we only need to do the third step, which leads to an \emph{equivalent} reachability game $\fgame$
  that is now played on the transition graph of an order-$n$ higher
  order pushdown automaton. In the new arena, the main vertices
  are of the form $(p,\stack,\vect{R},\theta)$: here $\stack$ is an $n$-stack (without links), $\vect{R}$ is actually
  a pair $(R_0,R_1)$ (we consider a reachability condition) and
  $\theta$ is either $0$ or $1$. The important fact is that $R_0$ and
  $R_1$ can be forced to be singletons: this follows from the fact
  that all vertices in $\pgame$ are controlled by \Eloise (and thus she can precisely force in which state the play goes if some $\popn{n+1}$ is eventually done). Therefore, one concludes that the
  size of the arena associated with $\fgame$ is polynomial in the
  size of $\pprocess$. Hence, one has shown the following: checking
  emptiness for an order-$(n+1)$ nondeterministic higher-order
  pushdown automaton whose input alphabet is reduced to a single
  letter can be polynomially reduced to solve a reachability game over
  the transition graph of an order-$n$ higher-order pushdown
  automaton. In conclusion, this latter problem is $n$-\exptime-hard.

%% file: Consequences.tex
\section{Consequences}\label{section:consequences}

\os{This is somehow new in its presentation and needs careful proofreading}

\subsection{Marking The Winning Region}\label{section:markingWR}

If one combines the fact that the winning region in a CPDA parity game is regular (Theorem~\ref{theorem:main}) together with the fact that the model of CPDA can perform regular test (Theorem~\ref{theo:closure-reg-test}) one directly gets the following result.

\begin{corollary}\label{corollary:marking}
Let $\mathcal{A}=\anglebra{\Gamma, Q,\delta, q_0}$ be an $n$-CPDA and let $\pgame$ be an $n$-CPDA parity game defined from $\mathcal{A}$. Then, one can build an order-$n$ CPDA $\mathcal{A}'$ with a state-set $Q'$, a subset $F\subseteq Q'$ and a mapping $\chi:Q'\rightarrow Q$ such that the following holds.
\begin{enumerate}
\item Restricted to the reachable configurations from their respective initial configuration, the transition graph of $\mathcal{A}$ and $\mathcal{A'}$ are isomorphic.
\item For every configuration $(q,\stack)$ of $\mathcal{A}$ that is reachable from the initial configuration, the corresponding configuration $(q',\stack')$ of $\mathcal{A'}$ is such that $q=\chi(q')$, and $(q,\stack)$ is winning for \Eloise in $\pgame$ if and only if $q'\in F$.
\end{enumerate}
\end{corollary}

In other words, it means that from $\pgame$ one can build a new game that behaves the same but where the winning region is explicitly marked (thanks to the subset $F$).

\subsection{Logical Consequences}

We now discuss the consequences of our main result regarding logical properties of structures generated by CPDA. Due to its strong connections with parity games, we obtain positive results regarding the $\mu$-calculus. Before discussing them, we will start with some consideration regarding monadic second-order (MSO) logic.

For both $\mu$-calculus and MSO logic, it is usual to consider structures given by an edge-labelled graphs coming with a labelling function that maps each vertex to a set of properties that hold in it. 

In the setting of CPDA, a natural way to define such a structure is by adding an input alphabet to the CPDA and defining the transition relation as a partial function depending on the current control state, the current top stack symbol and the input letter; the labelling function mapping vertices (\ie configurations) to properties can simply depend on the current control state (as we did when defining the colour in CPDA parity games). Rather than giving a formal definition we give an example that illustrates how to generate an edge-labelled graph using a CPDA with an input alphabet.

\begin{example}\label{Ex:Structure}
Let $\mathcal{A}=\anglebra{\Gamma,Q,\Delta,q_0}$	 be an order-2 CPDA over the input alphabet $A=\{a,b,c,1,2\}$ where $\Gamma=\{\alpha,\beta,\bot\}$, $Q=\{q_0,q_1,q_2\}$ and $\Delta: Q\times \Gamma\times A\rightarrow 2^{Q\times \Op{2}{\Gamma}\times \Op{2}{\Gamma}}$ is defined by 
\begin{itemize}
	\item $\Delta(q_0,\bot,2)=\Delta(q_0,\alpha,2)=\{(q_1,id;\pushn{2})\}$;
	\item $\Delta(q_1,\bot,a)=\Delta(q_1,\alpha,a)=\{(q_0,id;\pushlk{2}{\alpha})\}$;
	\item $\Delta(q_1,\bot,b)=\Delta(q_1,\alpha,b)=\{(q_2,id;\pushlk{2}{\beta})\}$;
	\item $\Delta(q_2,\alpha,1)=\Delta(q_2,\beta,1)=\{(q_2,id;\popn{1})\}$;
	\item $\Delta(q_2,\alpha,c)=\Delta(q_2,\beta,c)=\{(q_0,id;\collapse)\}$;
\end{itemize}
Then $\mathcal{A}$ generates the edge labelled graph from Figure~\ref{Fig:Example:Structure}.
\end{example}

\begin{figure}
\begin{tikzpicture}[scale=0.6,transform shape]
\tikzstyle{every node}=[font=\normalsize]
\tikzset{>=stealth}
\tikzset{edge from parent/.style={draw,-}}
\node (n00) at (0.4,0) {$(q_0,\mksk{\mksk{\bot}})$};

\node (n10) at (3.5,0) {$(q_1,\mksk{\mksk{\bot}\mksk{\bot}})$};
\draw[->] (n00) -- (n10) node[pos=0.5,above]{\small $2$};;
\node (n11) at (3.5,-2){$(q_2,\pstr[.1cm]{\mksk{\mklksk{s1}{\bot}\mksk{\bot \nd(n2-s1){b}}}})$};
\draw[->] (n10) -- (n11) node[pos=0.2,right]{\small $b$};;
\draw[->] (n11) -- (n00) node[pos=0.5,below left]{\small $c$};;
\node (n12) at (3.5,-4) {$(q_2,\mksk{\mksk{\bot}\mksk{\bot}})$};
\draw[->] (n11) -- (n12) node[pos=0.5,right]{\small $1$};;

\node (n20) at (7,0) {$(q_0,\pstr[.1cm]{\mksk{\mklksk{s1}{\bot}\mksk{\bot \nd(n2-s1){a}}}})$};
\draw[->] (n10) -- (n20) node[pos=0.5,above]{\small $a$};;

\node (n30) at (11,0) {$(q_1,\pstr[.1cm]{\mksk{\mklksk{s1}{\bot}\mksk{\bot \nd(n2-s1){a}}\mksk{\bot \nd(n3-s1){a}}}})$};
\draw[->] (n20) -- (n30) node[pos=0.5,above]{\small $2$};;
\node (n31) at (11,-2){$(q_2,\pstr[.1cm]{\mksk{\mklksk{s1}{\bot}\mklksk{s2}{\bot \nd(n2-s1){a}}\mksk{\bot \nd(n3-s1){a}\nd(n4-s2){b}}}})$};
\draw[->] (n30) -- (n31) node[pos=0.3,right]{\small $b$};;
\draw[->] (n31) -- (n20) node[pos=0.5,below left]{\small $c$};;
\node (n32) at (11,-4){$(q_2,\pstr[.1cm]{\mksk{\mklksk{s1}{\bot}\mklksk{s2}{\bot \nd(n2-s1){a}}\mksk{\bot \nd(n3-s1){a}}}})$};
\draw[->] (n32) -- (n00) node[pos=0.5,below left]{\small $c$};;
\draw[->] (n31) -- (n32) node[pos=0.25,right]{\small $1$};;
\node (n33) at (11,-6){$(q_2,\pstr[.1cm]{\mksk{\mklksk{s1}{\bot}\mklksk{s2}{\bot \nd(n2-s1){a}}\mksk{\bot }}})$};
\draw[->] (n32) -- (n33) node[pos=0.5,right]{\small $1$};;

\node (n40) at (15.7,0) {$(q_0,\pstr[.1cm]{\mksk{\mklksk{s1}{\bot}\mklksk{s2}{\bot \nd(n2-s1){a}}\mksk{\bot \nd(n3-s1){a}\nd(n4-s2){a}}}})$};
\draw[->] (n30) -- (n40) node[pos=0.5,above]{\small $a$};;

\node (n50) at (21,0) {$(q_0,\pstr[.1cm]{\mksk{\mklksk{s1}{\bot}\mklksk{s2}{\bot \nd(n2-s1){a}}\mksk{\bot \nd(n3-s1){a}\nd(n4-s2){a}}\mksk{\bot \nd(n5-s1){a}\nd(n6-s2){a}}}})$};
\draw[->] (n40) -- (n50) node[pos=0.5,above]{\small $2$};;
\node (n51) at (21,-2) {$(q_2,\pstr[.1cm]{\mksk{\mklksk{s1}{\bot}\mklksk{s2}{\bot \nd(n2-s1){a}}\mklksk{s3}{\bot \nd(n3-s1){a}\nd(n4-s2){a}}\mksk{\bot \nd(n5-s1){a}\nd(n6-s2){a}\nd(n7-s3){b}}}})$};
\draw[->] (n50) -- (n51) node[pos=0.5,right]{\small $b$};;
\draw[->] (n51) -- (n40) node[pos=0.5,below left]{\small $c$};;
\node (n52) at (21,-4) {$(q_2,\pstr[.1cm]{\mksk{\mklksk{s1}{\bot}\mklksk{s2}{\bot \nd(n2-s1){a}}\mklksk{s3}{\bot \nd(n3-s1){a}\nd(n4-s2){a}}\mksk{\bot \nd(n5-s1){a}\nd(n6-s2){a}}}})$};
\node (n53) at (21,-6) {$(q_2,\pstr[.1cm]{\mksk{\mklksk{s1}{\bot}\mklksk{s2}{\bot \nd(n2-s1){a}}\mklksk{s3}{\bot \nd(n3-s1){a}\nd(n4-s2){a}}\mksk{\bot \nd(n5-s1){a}}}})$};
\node (n54) at (21,-8) {$(q_2,\pstr[.1cm]{\mksk{\mklksk{s1}{\bot}\mklksk{s2}{\bot \nd(n2-s1){a}}\mklksk{s3}{\bot \nd(n3-s1){a}\nd(n4-s2){a}}\mksk{\bot}}})$};
\draw[->] (n51) -- (n52) node[pos=0.5,right]{\small $1$};;
\draw[->] (n52) -- (n53) node[pos=0.5,right]{\small $1$};;
\draw[->] (n53) -- (n54) node[pos=0.5,right]{\small $1$};;
\draw[->] (n52) -- (n20) node[pos=0.5,below left]{\small $c$};;
\draw[->] (n53) -- (n00) node[pos=0.5,below left]{\small $c$};;

\node at (23.5,0) {\dots};
\node at (23.5,-2) {\dots};
\node at (23.5,-4) {\dots};
\node at (23.5,-6) {\dots};
\node at (23.5,-8) {\dots};

\end{tikzpicture}
\caption{The edge labelled graph generated from the CPDA with input from Example~~\ref{Ex:Structure}.}\label{Fig:Example:Structure}
\end{figure}
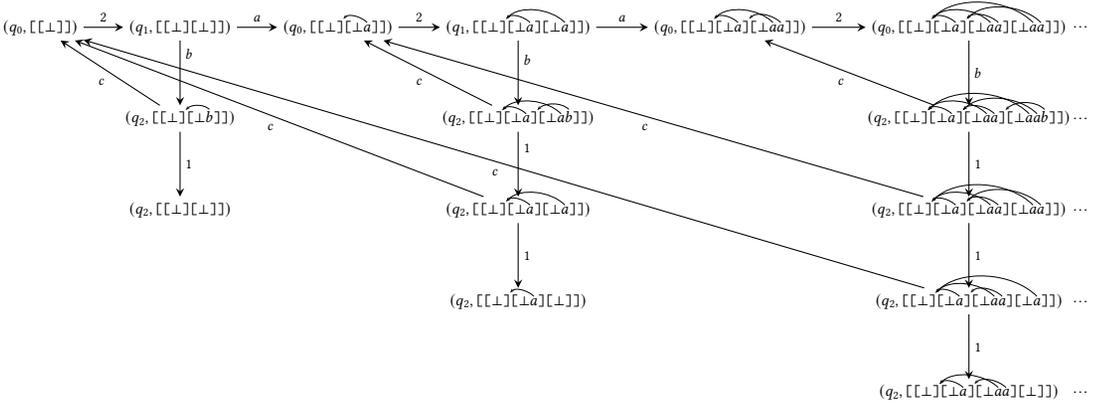

\subsection{Monadic Second-Order Logic}

We refer the reader to \cite{Thomas97} for classical definitions regarding MSO logic over graphs seen as relational structures. 

If one restricts its attention to higher-order pushdown automata, \ie CPDA that do not use the $\collapse$ operation, MSO logic is known to be decidable.

\begin{theorem}{\cite{Caucal02}}\label{thm:MSO-HOPDA}
The structures generated by higher-order pushdown automata have decidable MSO theories.
\end{theorem}

The next theorem shows that this is no longer the case for collapsible pushdown automata. In the statement below, FO(TC) is the \emph{transitive closure first-order logic} which is defined by extending the first-order logic with a transitive closure operator (see \eg \cite{WohrleT07}); in particular it subsumes the extension of first-order logic with a reachability predicate.

\begin{theorem}\label{thm:MSO-CPDA}
There exists a structure generated by a collapsible pushdown automata that has an undecidable MSO theory (actually even an undecidable FO(TC) theory).
\end{theorem}

\begin{proof}
Consider the following MSO interpretation $\mathcal{I}$\footnote{In this proof think of an interpretation as a collection of formulas of the form $\phi_A(x,y)$. Applying such an interpretation to a structure leads to a new structure with the same domain but different transitions: there is an $A$-labelled edge from $x$ to $y$ in the new structure if and only if $\phi_A(x,y)$ holds in the original structure.} (see \eg \cite{Courcelle94})  applied to the structure defined by the order-$2$ CPDA from Example~\ref{Ex:Structure}.

\[\begin{array}{lll}
\phi_A(x, y) & = & x \stackrel{C}{\longrightarrow} y \; \wedge \;
x \stackrel{R}{\longrightarrow} y \\
\phi_B(x, y) & = & x \stackrel{1}{\longrightarrow} y\\
\end{array}\]
with $C  =  {\overline 1}^\ast \, \overline b \, a \, 2 \, b \, 1^\ast$
and $R  =  c \, 2 \, a \, \overline c \; \vee \; \overline 1 \, c \,2 \, a \, \overline c \, 1$ where a bar-version of an edge label refers to an edge which is taken in the other direction. Hence, $C$ is used to enforce that $A$-edges occur only between vertices from consecutive columns in the original structure while $R$ is used to enforce that $A$-edges occurs only between vertices from consecutive rows in the original structure. 

We observe that the image of the structure generated by $\mathcal{A}$ by the interpretation $\mathcal{I}$, when restricted to its non-isolated vertices, is the “infinite half-grid” (see Figure~\ref{fig:hgrid}).

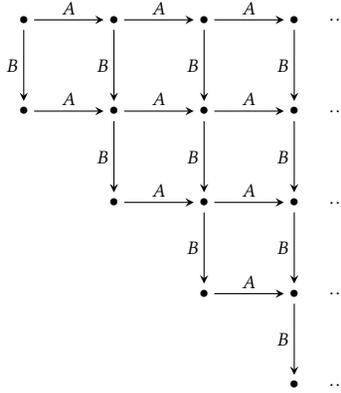
\begin{figure}
\begin{tikzpicture}[scale=.6,transform shape]
\tikzset{>=stealth}
\tikzstyle{every node}=[font=\Large]
\node(a1) at (0,0) {$\bullet$};
\node(a2) at (2,0) {$\bullet$};
\node(a3) at (4,0) {$\bullet$};
\node(a4) at (6,0) {$\bullet$};
\draw[->] (a1) -- (a2)  node[pos=0.5,above]{$A$};
\draw[->] (a2) -- (a3)  node[pos=0.5,above]{$A$};
\draw[->] (a3) -- (a4)  node[pos=0.5,above]{$A$};

\node(b1) at (0,-2) {$\bullet$};
\node(b2) at (2,-2) {$\bullet$};
\node(b3) at (4,-2) {$\bullet$};
\node(b4) at (6,-2) {$\bullet$};
\draw[->] (b1) -- (b2)  node[pos=0.5,above]{$A$};
\draw[->] (b2) -- (b3)  node[pos=0.5,above]{$A$};
\draw[->] (b3) -- (b4)  node[pos=0.5,above]{$A$};

\draw[->] (a1) -- (b1)  node[pos=0.5,left]{$B$};
\draw[->] (a2) -- (b2)  node[pos=0.5,left]{$B$};
\draw[->] (a3) -- (b3)  node[pos=0.5,left]{$B$};
\draw[->] (a4) -- (b4)  node[pos=0.5,left]{$B$};

\node(c2) at (2,-4) {$\bullet$};
\node(c3) at (4,-4) {$\bullet$};
\node(c4) at (6,-4) {$\bullet$};
\draw[->] (c2) -- (c3)  node[pos=0.5,above]{$A$};
\draw[->] (c3) -- (c4)  node[pos=0.5,above]{$A$};

\draw[->] (b2) -- (c2)  node[pos=0.5,left]{$B$};
\draw[->] (b3) -- (c3)  node[pos=0.5,left]{$B$};
\draw[->] (b4) -- (c4)  node[pos=0.5,left]{$B$};

\node(d3) at (4,-6) {$\bullet$};
\node(d4) at (6,-6) {$\bullet$};
\draw[->] (d3) -- (d4)  node[pos=0.5,above]{$A$};

\draw[->] (c3) -- (d3)  node[pos=0.5,left]{$B$};
\draw[->] (c4) -- (d4)  node[pos=0.5,left]{$B$};

\node(e4) at (6,-8) {$\bullet$};

\draw[->] (d4) -- (e4)  node[pos=0.5,left]{$B$};

\node at (7,0) {\dots};
\node at (7,-2) {\dots};
\node at (7,-4) {\dots};
\node at (7,-6) {\dots};
\node at (7,-8) {\dots};

\end{tikzpicture}
\caption{The “infinite half-grid”.}\label{fig:hgrid}
\end{figure}

As the infinite (half-) grid has an undecidable MSO theory and as MSO interpretations preserve MSO decidability we conclude that the MSO theory of the structure generated by $\mathcal{A}$ is undecidable.

To refine the result to FO(TC), we simply remark that the interpretation $I$ is FO(TC) definable and that the infinite (half) grid has an undecidable FO(TC) theory \cite{WohrleT07}.
	\end{proof}

\begin{remark}
One can wonder about fragments of MSO weaker than FO(TC), \emph{e.g.} FO(Reach) (the extension of first-order logic with the reachability predicate) or the classical first-order logic (FO). 
On a positive side, Kartzow proved in \cite{Kartzow10} that the structures generated by order-$2$- CPDA have decidable FO(Reach) theories.
But moving to order-3 leads to undecidability, even if one restricts to FO, as proved by Broadbent in \cite{Broadbent12}.
\end{remark}

The following is a direct consequence of Theorem~\ref{thm:MSO-HOPDA} and Theorem~\ref{thm:MSO-CPDA}.

\begin{corollary}
	The class of graphs generated by collapsible pushdown automata strictly contains the class of graphs generated by higher-order pushdown automata.
\end{corollary} 

\subsection{$\mu$-Calculus}

We refer the reader to \cite{AN01} for classical definitions regarding $\mu$-calculus as well as its connections with games.

Due to the tight connection between $\mu$-calculus model-checking and solving parity games, and the fact that the class of structures generated by CPDA is trivially closed by taking a synchronised product with a finite graph, Theorem~\ref{theorem:main} directly leads the following result.

\begin{corollary}\label{cor:mu-calculus} The following holds.
\begin{enumerate}
\item The $\mu$-calculus model-checking problem against structures generated by collapsible pushdown automata is decidable and its complexity (where $n$ denotes the order of the CPDA) is $n$-times exponential in the number of states of the CPDA, $n$-times exponential in the alternation depth of greatest and smallest fixpoints in the $\mu$-calculus formula and polynomial in the size of the stack alphabet of the CPDA. 
\item The sets of configurations definable by a $\mu$-calculus formula over a graph generated by a collapsible pushdown automata are regular.
\end{enumerate}
\end{corollary}

\begin{remark}
In the case of higher-order pushdown automata, links are useless and therefore stacks can be seen as finite words over the alphabet $\Gamma\cup\{[,]\}$ (where $\Gamma$ denotes the stack alphabet) and regular sets of configurations are regular languages is the traditional sense of finite words. Hence, Corollary~\ref{cor:mu-calculus} permits to retrieve the main result in~\cite[Theorem~6]{CHMOS08} where the $\mu$-calculus global model-checking problem against higher-order pushdown automata was tackled. 

Also note that in this setting, a stronger notion of regularity was introduced in \cite{Carayol05} and shown to exactly capture MSO-definable subsets of configuration.
\end{remark}

As we did in Section~\ref{section:markingWR} to mark winning regions, combining item~(2) from Corollary~\ref{cor:mu-calculus}  together with the fact that the model of CPDA can perform regular test (Theorem~\ref{theo:closure-reg-test}) one directly gets the following result about marking a $\mu$-calculus defined subset of vertices in the transition graph of a CPDA.

\begin{corollary}\label{corollary:marking-mu-calculus}
Let $\mathcal{A}=\anglebra{\Gamma, Q,\delta, q_0}$ be an $n$-CPDA and let $\phi$ be a $\mu$-calculus formula defining a subset of vertices in the transition graph of $\mathcal{A}$. Then, one can build an order-$n$ CPDA $\mathcal{A}'$ with a state-set $Q'$, a subset $F\subseteq Q'$ and a mapping $\chi:Q'\rightarrow Q$ such that the following holds.
\begin{enumerate}
\item Restricted to the reachable configurations from their respective initial configuration, the transition graph of $\mathcal{A}$ and $\mathcal{A'}$ are isomorphic.
\item For every configuration $(q,\stack)$ of $\mathcal{A}$ that is reachable from the initial configuration, the corresponding configuration $(q',\stack')$ of $\mathcal{A'}$ is such that $q=\chi(q')$, and $\phi$ holds in $(q,\stack)$ if and only if $q'\in F$.
\end{enumerate}
\end{corollary}

\section{Perspectives}

A natural perspective is to combine the results presented here with the equi-expressivity result~\cite{HMOS08,HMOS17,CS12} between  higher-order recursion schemes and collapsible pushdown automata for generating trees.
In particular they imply the decidability of the MSO model-checking problem, both its local~\cite{HMOS08} and global version (also known as reflection)~\cite{BCOS10}, and the MSO selection problem (a synthesis-like problem)~\cite{CS12}. 

These results and other consequences are discussed in full detail in a companion paper~\cite{BCOS20}.

%% file: main.bbl
%%% -*-BibTeX-*-
%%% Do NOT edit. File created by BibTeX with style
%%% ACM-Reference-Format-Journals [18-Jan-2012].

\newcommand{\noopsort}[1]{} \newcommand{\singleletter}[1]{#1}
  \newcommand{\etal}{et al.}
\begin{thebibliography}{41}

%%% ====================================================================
%%% NOTE TO THE USER: you can override these defaults by providing
%%% customized versions of any of these macros before the \bibliography
%%% command.  Each of them MUST provide its own final punctuation,
%%% except for \shownote{}, \showDOI{}, and \showURL{}.  The latter two
%%% do not use final punctuation, in order to avoid confusing it with
%%% the Web address.
%%%
%%% To suppress output of a particular field, define its macro to expand
%%% to an empty string, or better, \unskip, like this:
%%%
%%% \newcommand{\showDOI}[1]{\unskip}   % LaTeX syntax
%%%
%%% \def \showDOI #1{\unskip}           % plain TeX syntax
%%%
%%% ====================================================================

\ifx \showCODEN    \undefined \def \showCODEN     #1{\unskip}     \fi
\ifx \showDOI      \undefined \def \showDOI       #1{#1}\fi
\ifx \showISBNx    \undefined \def \showISBNx     #1{\unskip}     \fi
\ifx \showISBNxiii \undefined \def \showISBNxiii  #1{\unskip}     \fi
\ifx \showISSN     \undefined \def \showISSN      #1{\unskip}     \fi
\ifx \showLCCN     \undefined \def \showLCCN      #1{\unskip}     \fi
\ifx \shownote     \undefined \def \shownote      #1{#1}          \fi
\ifx \showarticletitle \undefined \def \showarticletitle #1{#1}   \fi
\ifx \showURL      \undefined \def \showURL       {\relax}        \fi
% The following commands are used for tagged output and should be
% invisible to TeX
\providecommand\bibfield[2]{#2}
\providecommand\bibinfo[2]{#2}
\providecommand\natexlab[1]{#1}
\providecommand\showeprint[2][]{arXiv:#2}

\bibitem[\protect\citeauthoryear{Aehlig, de~Miranda, and Ong}{Aehlig
  et~al\mbox{.}}{2005}]%
        {AdMO05a}
\bibfield{author}{\bibinfo{person}{Klaus Aehlig}, \bibinfo{person}{Jolie de
  Miranda}, {and} \bibinfo{person}{C.-H.~Luke Ong}.}
  \bibinfo{year}{2005}\natexlab{}.
\newblock \showarticletitle{Safety is not a Restriction at Level 2 for String
  Languages}. In \bibinfo{booktitle}{\emph{Proceedings of the 8th International
  Conference on Foundations of Software Science and Computational Structures
  (FoSSaCS 2005)}} \emph{(\bibinfo{series}{Lecture Notes in Computer
  Science})}, Vol.~\bibinfo{volume}{3411}.
  \bibinfo{publisher}{Springer-Verlag}, \bibinfo{pages}{490--501}.
\newblock


\bibitem[\protect\citeauthoryear{Arnold and {Niwi{\'n}ski}}{Arnold and
  {Niwi{\'n}ski}}{2001}]%
        {AN01}
\bibfield{author}{\bibinfo{person}{Andr\'e Arnold} {and}
  \bibinfo{person}{Damian {Niwi{\'n}ski}}.} \bibinfo{year}{2001}\natexlab{}.
\newblock \bibinfo{booktitle}{\emph{Rudiments of mu-Calculus}}.
  \bibinfo{series}{Studies in Logic and the Foundations of Mathematics},
  Vol.~\bibinfo{volume}{146}.
\newblock \bibinfo{publisher}{Elsevier}.
\newblock


\bibitem[\protect\citeauthoryear{Bouajjani and Meyer}{Bouajjani and
  Meyer}{2004}]%
        {BM04}
\bibfield{author}{\bibinfo{person}{Ahmed Bouajjani} {and}
  \bibinfo{person}{Antoine Meyer}.} \bibinfo{year}{2004}\natexlab{}.
\newblock \showarticletitle{Symbolic Reachability Analysis of Higher-Order
  Context-Free Processes}. In \bibinfo{booktitle}{\emph{Proceedings of the 24th
  International Conference on Foundations of Software Technology and
  Theoretical Computer Science (FST\&TCS~2004)}}
  \emph{(\bibinfo{series}{Lecture Notes in Computer Science})},
  Vol.~\bibinfo{volume}{3328}. \bibinfo{publisher}{Springer-Verlag},
  \bibinfo{pages}{135--147}.
\newblock


\bibitem[\protect\citeauthoryear{Broadbent}{Broadbent}{2012}]%
        {Broadbent12}
\bibfield{author}{\bibinfo{person}{Christopher~H. Broadbent}.}
  \bibinfo{year}{2012}\natexlab{}.
\newblock \showarticletitle{The Limits of Decidability for First Order Logic on
  CPDA Graphs}. In \bibinfo{booktitle}{\emph{Proceedings of the 29th Symposium
  on Theoretical Aspects of Computer Science (STACS~2012)}}
  \emph{(\bibinfo{series}{LIPIcs})}, Vol.~\bibinfo{volume}{14}.
  \bibinfo{publisher}{Schloss Dagstuhl - Leibniz-Zentrum f\"ur Informatik},
  \bibinfo{pages}{589--600}.
\newblock


\bibitem[\protect\citeauthoryear{Broadbent, Carayol, Hague, and
  Serre}{Broadbent et~al\mbox{.}}{2012}]%
        {BCHS12}
\bibfield{author}{\bibinfo{person}{Christopher~H. Broadbent},
  \bibinfo{person}{Arnaud Carayol}, \bibinfo{person}{Matthew Hague}, {and}
  \bibinfo{person}{Olivier Serre}.} \bibinfo{year}{2012}\natexlab{}.
\newblock \showarticletitle{A Saturation Method for Collapsible Pushdown
  Systems}. In \bibinfo{booktitle}{\emph{Proceedings of the 39th International
  Colloquium on Automata, Languages, and Programming (ICALP~2012)}}
  \emph{(\bibinfo{series}{Lecture Notes in Computer Science})},
  Vol.~\bibinfo{volume}{7392}. \bibinfo{publisher}{Springer-Verlag},
  \bibinfo{pages}{165--176}.
\newblock


\bibitem[\protect\citeauthoryear{Broadbent, Carayol, Hague, and
  Serre}{Broadbent et~al\mbox{.}}{2013}]%
        {BCHS13}
\bibfield{author}{\bibinfo{person}{Christopher~H. Broadbent},
  \bibinfo{person}{Arnaud Carayol}, \bibinfo{person}{Matthew Hague}, {and}
  \bibinfo{person}{Olivier Serre}.} \bibinfo{year}{2013}\natexlab{}.
\newblock \showarticletitle{{C-SHORe}: a Collapsible Approach to Higher-Order
  Verification}. In \bibinfo{booktitle}{\emph{Proceedings of the 18th ACM
  SIGPLAN International Conference on Functional Programming (ICFP~2013)}}.
  \bibinfo{publisher}{ACM}, \bibinfo{pages}{13--24}.
\newblock


\bibitem[\protect\citeauthoryear{Broadbent, Carayol, Hague, and
  Serre}{Broadbent et~al\mbox{.}}{2020}]%
        {BCOS20}
\bibfield{author}{\bibinfo{person}{Christopher~H. Broadbent},
  \bibinfo{person}{Arnaud Carayol}, \bibinfo{person}{Matthew Hague}, {and}
  \bibinfo{person}{Olivier Serre}.} \bibinfo{year}{2020}\natexlab{}.
\newblock \bibinfo{title}{Higher-Order Recursion Schemes and Collapsible
  Pushdown Automata: Logical Properties}.  (\bibinfo{year}{2020}).
\newblock


\bibitem[\protect\citeauthoryear{Broadbent, Carayol, Ong, and Serre}{Broadbent
  et~al\mbox{.}}{2010}]%
        {BCOS10}
\bibfield{author}{\bibinfo{person}{Christopher~H. Broadbent},
  \bibinfo{person}{Arnaud Carayol}, \bibinfo{person}{C.-H.~Luke Ong}, {and}
  \bibinfo{person}{Olivier Serre}.} \bibinfo{year}{2010}\natexlab{}.
\newblock \showarticletitle{Recursion Schemes and Logical Reflexion}. In
  \bibinfo{booktitle}{\emph{Proceedings of the 25th Annual IEEE Symposium on
  Logic in Computer Science (LiCS~2010)}}. \bibinfo{publisher}{IEEE Computer
  Society}, \bibinfo{pages}{120--129}.
\newblock


\bibitem[\protect\citeauthoryear{Cachat}{Cachat}{2002}]%
        {Cachat02}
\bibfield{author}{\bibinfo{person}{Thierry Cachat}.}
  \bibinfo{year}{2002}\natexlab{}.
\newblock \showarticletitle{Uniform Solution of Parity Games on
  Prefix-Recognizable Graphs}. In \bibinfo{booktitle}{\emph{4th International
  Workshop on Verification of Infinite-State Systems, Infinity 2002}}
  \emph{(\bibinfo{series}{Electronic Notes in Theoretical Computer Science})},
  Vol.~\bibinfo{volume}{68}. \bibinfo{publisher}{Elsevier Science Publishers}.
\newblock
Issue 6.


\bibitem[\protect\citeauthoryear{Cachat}{Cachat}{2003a}]%
        {CachatPHD}
\bibfield{author}{\bibinfo{person}{Thierry Cachat}.}
  \bibinfo{year}{2003}\natexlab{a}.
\newblock \emph{\bibinfo{title}{Games on Pushdown Graphs and Extensions}}.
\newblock \bibinfo{thesistype}{Ph.D. Dissertation}. \bibinfo{school}{RWTH
  Aachen}.
\newblock


\bibitem[\protect\citeauthoryear{Cachat}{Cachat}{2003b}]%
        {Cachat03}
\bibfield{author}{\bibinfo{person}{Thierry Cachat}.}
  \bibinfo{year}{2003}\natexlab{b}.
\newblock \showarticletitle{Higher Order Pushdown Automata, the {Caucal}
  Hierarchy of Graphs and Parity Games}. In
  \bibinfo{booktitle}{\emph{Proceedings of the 30th International Colloquium on
  Automata, Languages, and Programming (ICALP~2003)}}
  \emph{(\bibinfo{series}{Lecture Notes in Computer Science})},
  Vol.~\bibinfo{volume}{2719}. \bibinfo{publisher}{Springer-Verlag},
  \bibinfo{pages}{556--569}.
\newblock


\bibitem[\protect\citeauthoryear{Cachat and Walukiewicz}{Cachat and
  Walukiewicz}{2007}]%
        {CachatWalukiewicz07}
\bibfield{author}{\bibinfo{person}{Thierry Cachat} {and} \bibinfo{person}{Igor
  Walukiewicz}.} \bibinfo{year}{2007}\natexlab{}.
\newblock \showarticletitle{The Complexity of Games on Higher Order Pushdown
  Automata}.
\newblock \bibinfo{journal}{\emph{CoRR}}  \bibinfo{volume}{abs/0705.0262}
  (\bibinfo{year}{2007}).
\newblock


\bibitem[\protect\citeauthoryear{Carayol}{Carayol}{2005}]%
        {Carayol05}
\bibfield{author}{\bibinfo{person}{Arnaud Carayol}.}
  \bibinfo{year}{2005}\natexlab{}.
\newblock \showarticletitle{Regular Sets of Higher-Order Pushdown Stacks}. In
  \bibinfo{booktitle}{\emph{Proceedings of the 30th Symposium, Mathematical
  Foundations of Computer Science (MFCS~2005)}} \emph{(\bibinfo{series}{Lecture
  Notes in Computer Science})}, Vol.~\bibinfo{volume}{3618}.
  \bibinfo{publisher}{Springer-Verlag}, \bibinfo{pages}{168--179}.
\newblock


\bibitem[\protect\citeauthoryear{Carayol, Meyer, Hague, Ong, and Serre}{Carayol
  et~al\mbox{.}}{2008}]%
        {CHMOS08}
\bibfield{author}{\bibinfo{person}{Arnaud Carayol}, \bibinfo{person}{Antoine
  Meyer}, \bibinfo{person}{Matthew Hague}, \bibinfo{person}{C.-H.~Luke Ong},
  {and} \bibinfo{person}{Olivier Serre}.} \bibinfo{year}{2008}\natexlab{}.
\newblock \showarticletitle{Winning Regions of Higher-Order Pushdown Games}. In
  \bibinfo{booktitle}{\emph{Proceedings of the 23rd Annual IEEE Symposium on
  Logic in Computer Science (LiCS~2008)}}. \bibinfo{publisher}{IEEE Computer
  Society}, \bibinfo{pages}{193--204}.
\newblock


\bibitem[\protect\citeauthoryear{Carayol and Serre}{Carayol and Serre}{2012}]%
        {CS12}
\bibfield{author}{\bibinfo{person}{Arnaud Carayol} {and}
  \bibinfo{person}{Olivier Serre}.} \bibinfo{year}{2012}\natexlab{}.
\newblock \showarticletitle{Collapsible Pushdown Automata and Labeled Recursion
  Schemes: Equivalence, Safety and Effective Selection}. In
  \bibinfo{booktitle}{\emph{Proceedings of the 27th Annual IEEE Symposium on
  Logic in Computer Science (LiCS~2012)}}. \bibinfo{publisher}{IEEE Computer
  Society}, \bibinfo{pages}{165--174}.
\newblock


\bibitem[\protect\citeauthoryear{Carayol and Slaats}{Carayol and
  Slaats}{2008}]%
        {CS08}
\bibfield{author}{\bibinfo{person}{Arnaud Carayol} {and}
  \bibinfo{person}{Michaela Slaats}.} \bibinfo{year}{2008}\natexlab{}.
\newblock \showarticletitle{Positional Strategies for Higher-Order Pushdown
  Parity Games}. In \bibinfo{booktitle}{\emph{Proceedings of the 33rd
  Symposium, Mathematical Foundations of Computer Science (MFCS~2008)}}
  \emph{(\bibinfo{series}{Lecture Notes in Computer Science})},
  Vol.~\bibinfo{volume}{5162}. \bibinfo{publisher}{Springer-Verlag},
  \bibinfo{pages}{217--228}.
\newblock


\bibitem[\protect\citeauthoryear{Caucal}{Caucal}{2002}]%
        {Caucal02}
\bibfield{author}{\bibinfo{person}{Didier Caucal}.}
  \bibinfo{year}{2002}\natexlab{}.
\newblock \showarticletitle{On Infinite Terms Having a Decidable Monadic
  Theory}. In \bibinfo{booktitle}{\emph{Proceedings of the 27th Symposium,
  Mathematical Foundations of Computer Science (MFCS~2002)}}
  \emph{(\bibinfo{series}{Lecture Notes in Computer Science})},
  Vol.~\bibinfo{volume}{2420}. \bibinfo{publisher}{Springer-Verlag},
  \bibinfo{pages}{165--176}.
\newblock


\bibitem[\protect\citeauthoryear{Courcelle}{Courcelle}{1994}]%
        {Courcelle94}
\bibfield{author}{\bibinfo{person}{Bruno Courcelle}.}
  \bibinfo{year}{1994}\natexlab{}.
\newblock \showarticletitle{Monadic Second-Order Definable Graph Transductions:
  A Survey.}
\newblock \bibinfo{journal}{\emph{Theoretical Computer Science}}
  \bibinfo{volume}{126}, \bibinfo{number}{1} (\bibinfo{year}{1994}),
  \bibinfo{pages}{53--75}.
\newblock


\bibitem[\protect\citeauthoryear{Emerson and Jutla}{Emerson and Jutla}{1991}]%
        {EJ91}
\bibfield{author}{\bibinfo{person}{E.~Allen Emerson} {and}
  \bibinfo{person}{Charanjit~S. Jutla}.} \bibinfo{year}{1991}\natexlab{}.
\newblock \showarticletitle{Tree Automata, mu-Calculus and Determinacy
  (Extended Abstract)}. In \bibinfo{booktitle}{\emph{Proceedings of the 32nd
  Annual Symposium on Foundations of Computer Science (FoCS 1991)}}.
  \bibinfo{publisher}{IEEE Computer Society}, \bibinfo{pages}{368--377}.
\newblock


\bibitem[\protect\citeauthoryear{Engelfriet}{Engelfriet}{1991}]%
        {Engelfriet91}
\bibfield{author}{\bibinfo{person}{Joost Engelfriet}.}
  \bibinfo{year}{1991}\natexlab{}.
\newblock \showarticletitle{Iterated Stack Automata and Complexity Classes}.
\newblock \bibinfo{journal}{\emph{Information and Computation}}
  \bibinfo{volume}{95}, \bibinfo{number}{1} (\bibinfo{year}{1991}),
  \bibinfo{pages}{21--75}.
\newblock


\bibitem[\protect\citeauthoryear{Gurevich and Harrington}{Gurevich and
  Harrington}{1982}]%
        {GurevichH82}
\bibfield{author}{\bibinfo{person}{Yuri Gurevich} {and} \bibinfo{person}{Leo
  Harrington}.} \bibinfo{year}{1982}\natexlab{}.
\newblock \showarticletitle{Trees, Automata, and Games}. In
  \bibinfo{booktitle}{\emph{Proceedings of the Fourteenth Annual ACM Symposium
  on the Theory of Computing (STOC'82)}}. \bibinfo{publisher}{ACM},
  \bibinfo{pages}{60--65}.
\newblock


\bibitem[\protect\citeauthoryear{Hague}{Hague}{2008}]%
        {Hag08}
\bibfield{author}{\bibinfo{person}{Matthew Hague}.}
  \bibinfo{year}{2008}\natexlab{}.
\newblock \emph{\bibinfo{title}{Global Model Checking of Higher-Order Pushdown
  Systems}}.
\newblock \bibinfo{thesistype}{Ph.D. Dissertation}. \bibinfo{school}{University
  of Oxford}.
\newblock


\bibitem[\protect\citeauthoryear{Hague, Murawski, Ong, and Serre}{Hague
  et~al\mbox{.}}{2008}]%
        {HMOS08}
\bibfield{author}{\bibinfo{person}{Matthew Hague}, \bibinfo{person}{Andrzej~S.
  Murawski}, \bibinfo{person}{C.-H.~Luke Ong}, {and} \bibinfo{person}{Olivier
  Serre}.} \bibinfo{year}{2008}\natexlab{}.
\newblock \showarticletitle{Collapsible Pushdown Automata and Recursion
  Schemes}. In \bibinfo{booktitle}{\emph{Proceedings of the 23rd Annual IEEE
  Symposium on Logic in Computer Science (LiCS~2008)}}.
  \bibinfo{publisher}{IEEE Computer Society}, \bibinfo{pages}{452--461}.
\newblock


\bibitem[\protect\citeauthoryear{Hague, Murawski, Ong, and Serre}{Hague
  et~al\mbox{.}}{2017}]%
        {HMOS17}
\bibfield{author}{\bibinfo{person}{Matthew Hague}, \bibinfo{person}{Andrzej~S.
  Murawski}, \bibinfo{person}{C.{-}H.~Luke Ong}, {and} \bibinfo{person}{Olivier
  Serre}.} \bibinfo{year}{2017}\natexlab{}.
\newblock \showarticletitle{Collapsible Pushdown Automata and Recursion
  Schemes}.
\newblock \bibinfo{journal}{\emph{ACM Transactions on Computational Logic}}
  \bibinfo{volume}{18}, \bibinfo{number}{3} (\bibinfo{year}{2017}),
  \bibinfo{pages}{25:1--25:42}.
\newblock


\bibitem[\protect\citeauthoryear{Hague and Ong}{Hague and Ong}{2008}]%
        {HagueO08}
\bibfield{author}{\bibinfo{person}{Matthew Hague} {and}
  \bibinfo{person}{C.-H.~Luke Ong}.} \bibinfo{year}{2008}\natexlab{}.
\newblock \showarticletitle{Symbolic Backwards-Reachability Analysis for
  Higher-Order Pushdown Systems}.
\newblock \bibinfo{journal}{\emph{Logical Methods in Computer Science}}
  \bibinfo{volume}{4}, \bibinfo{number}{4} (\bibinfo{year}{2008}).
\newblock


\bibitem[\protect\citeauthoryear{Hague and Ong}{Hague and Ong}{2011}]%
        {HagueO11}
\bibfield{author}{\bibinfo{person}{Matthew Hague} {and}
  \bibinfo{person}{C.-H.~Luke Ong}.} \bibinfo{year}{2011}\natexlab{}.
\newblock \showarticletitle{A Saturation Method for the Modal $\mu$-Calculus
  over Pushdown systems}.
\newblock \bibinfo{journal}{\emph{Information and Computation}}
  \bibinfo{volume}{209}, \bibinfo{number}{5} (\bibinfo{year}{2011}),
  \bibinfo{pages}{799--821}.
\newblock


\bibitem[\protect\citeauthoryear{Kartzow}{Kartzow}{2010}]%
        {Kartzow10}
\bibfield{author}{\bibinfo{person}{Alexander Kartzow}.}
  \bibinfo{year}{2010}\natexlab{}.
\newblock \showarticletitle{Collapsible Pushdown Graphs of Level 2 are
  Tree-Automatic}. In \bibinfo{booktitle}{\emph{Proceedings of the 27th
  Symposium on Theoretical Aspects of Computer Science (STACS~2010)}}
  \emph{(\bibinfo{series}{LIPIcs})}, Vol.~\bibinfo{volume}{5}.
  \bibinfo{publisher}{Schloss Dagstuhl - Leibniz-Zentrum f\"ur Informatik},
  \bibinfo{pages}{501--512}.
\newblock


\bibitem[\protect\citeauthoryear{Knapik, Niwi{\'n}ski, Urzyczyn, and
  Walukiewicz}{Knapik et~al\mbox{.}}{2005}]%
        {KNUW05}
\bibfield{author}{\bibinfo{person}{Teodor Knapik}, \bibinfo{person}{Damian
  Niwi{\'n}ski}, \bibinfo{person}{Pawel Urzyczyn}, {and} \bibinfo{person}{Igor
  Walukiewicz}.} \bibinfo{year}{2005}\natexlab{}.
\newblock \showarticletitle{Unsafe Grammars and Panic Automata}. In
  \bibinfo{booktitle}{\emph{Proceedings of the 32nd International Colloquium on
  Automata, Languages, and Programming (ICALP~2005)}}
  \emph{(\bibinfo{series}{Lecture Notes in Computer Science})},
  Vol.~\bibinfo{volume}{3580}. \bibinfo{publisher}{Springer-Verlag},
  \bibinfo{pages}{1450--1461}.
\newblock


\bibitem[\protect\citeauthoryear{Martin}{Martin}{1975}]%
        {Martin75}
\bibfield{author}{\bibinfo{person}{Donald~A. Martin}.}
  \bibinfo{year}{1975}\natexlab{}.
\newblock \showarticletitle{{B}orel Determinacy}.
\newblock \bibinfo{journal}{\emph{Ann. Math.}} \bibinfo{volume}{102},
  \bibinfo{number}{2} (\bibinfo{year}{1975}), \bibinfo{pages}{363--371}.
\newblock


\bibitem[\protect\citeauthoryear{Muller and Schupp}{Muller and Schupp}{1985}]%
        {MullerS85}
\bibfield{author}{\bibinfo{person}{David.~E. Muller} {and}
  \bibinfo{person}{Paul.~E. Schupp}.} \bibinfo{year}{1985}\natexlab{}.
\newblock \showarticletitle{The Theory of Ends, Pushdown Automata, and
  Second-Order Logic}.
\newblock \bibinfo{journal}{\emph{Theoretical Computer Science}}
  \bibinfo{volume}{37} (\bibinfo{year}{1985}), \bibinfo{pages}{51--75}.
\newblock


\bibitem[\protect\citeauthoryear{Rabin}{Rabin}{1969}]%
        {Rabin69}
\bibfield{author}{\bibinfo{person}{Michael~O. Rabin}.}
  \bibinfo{year}{1969}\natexlab{}.
\newblock \showarticletitle{Decidability of Second-Order Theories and Automata
  on Infinite Trees}.
\newblock \bibinfo{journal}{\emph{Trans. Amer. Math. Soc.}}
  \bibinfo{volume}{141} (\bibinfo{year}{1969}), \bibinfo{pages}{1--35}.
\newblock


\bibitem[\protect\citeauthoryear{Serre}{Serre}{2003}]%
        {Serre03}
\bibfield{author}{\bibinfo{person}{Olivier Serre}.}
  \bibinfo{year}{2003}\natexlab{}.
\newblock \showarticletitle{{N}ote on Winning Positions on Pushdown Games with
  $\omega$-Regular Winning Conditions}.
\newblock \bibinfo{journal}{\emph{{I}nformation {P}rocessing {L}etters}}
  \bibinfo{volume}{85} (\bibinfo{year}{2003}), \bibinfo{pages}{285--291}.
\newblock


\bibitem[\protect\citeauthoryear{Serre}{Serre}{2004}]%
        {SerrePHD}
\bibfield{author}{\bibinfo{person}{Olivier Serre}.}
  \bibinfo{year}{2004}\natexlab{}.
\newblock \emph{\bibinfo{title}{Contribution \`a l'\'etude des jeux sur des
  graphes de processus \`a pile}}.
\newblock \bibinfo{thesistype}{Ph.D. Dissertation}.
  \bibinfo{school}{Universit\'e Paris 7}.
\newblock


\bibitem[\protect\citeauthoryear{Stirling}{Stirling}{2009}]%
        {Sti09b}
\bibfield{author}{\bibinfo{person}{Colin Stirling}.}
  \bibinfo{year}{2009}\natexlab{}.
\newblock \showarticletitle{Dependency Tree Automata}. In
  \bibinfo{booktitle}{\emph{Proceedings of the 12th International Conference on
  Foundations of Software Science and Computational Structures (FoSSaCS 2009)}}
  \emph{(\bibinfo{series}{Lecture Notes in Computer Science})},
  Vol.~\bibinfo{volume}{5504}. \bibinfo{publisher}{Springer-Verlag},
  \bibinfo{pages}{92--106}.
\newblock


\bibitem[\protect\citeauthoryear{Thomas}{Thomas}{1997}]%
        {Thomas97}
\bibfield{author}{\bibinfo{person}{Wolfgang Thomas}.}
  \bibinfo{year}{1997}\natexlab{}.
\newblock \showarticletitle{Languages, Automata, and Logic}.
\newblock In \bibinfo{booktitle}{\emph{Handbook of Formal Language Theory}},
  \bibfield{editor}{\bibinfo{person}{G.~Rozenberg} {and}
  \bibinfo{person}{A.~Salomaa}} (Eds.). Vol.~\bibinfo{volume}{III}.
  \bibinfo{publisher}{Springer-Verlag}, \bibinfo{pages}{389--455}.
\newblock


\bibitem[\protect\citeauthoryear{Vardi}{Vardi}{1998}]%
        {Vardi98}
\bibfield{author}{\bibinfo{person}{Moshe~Y. Vardi}.}
  \bibinfo{year}{1998}\natexlab{}.
\newblock \showarticletitle{Reasoning about The Past with Two-Way Automata}. In
  \bibinfo{booktitle}{\emph{Proceedings of the 25th International Colloquium on
  Automata, Languages, and Programming (ICALP~1998)}}
  \emph{(\bibinfo{series}{Lecture Notes in Computer Science})},
  Vol.~\bibinfo{volume}{1443}. \bibinfo{publisher}{Springer-Verlag},
  \bibinfo{pages}{628--641}.
\newblock


\bibitem[\protect\citeauthoryear{Walukiewicz}{Walukiewicz}{1996}]%
        {Walukiewicz96}
\bibfield{author}{\bibinfo{person}{Igor Walukiewicz}.}
  \bibinfo{year}{1996}\natexlab{}.
\newblock \showarticletitle{Pushdown Processes: Games and Model Checking}. In
  \bibinfo{booktitle}{\emph{Proceeding of the 8th International Conference on
  Computer Aided Verification (CAV~1996)}} \emph{(\bibinfo{series}{Lecture
  Notes in Computer Science})}, Vol.~\bibinfo{volume}{1102}.
  \bibinfo{publisher}{Springer-Verlag}, \bibinfo{pages}{62--74}.
\newblock


\bibitem[\protect\citeauthoryear{Walukiewicz}{Walukiewicz}{2001}]%
        {Walukiewicz01}
\bibfield{author}{\bibinfo{person}{Igor Walukiewicz}.}
  \bibinfo{year}{2001}\natexlab{}.
\newblock \showarticletitle{Pushdown Processes: Games and Model-Checking}.
\newblock \bibinfo{journal}{\emph{Information and Computation}}
  \bibinfo{volume}{157} (\bibinfo{year}{2001}), \bibinfo{pages}{234--263}.
\newblock


\bibitem[\protect\citeauthoryear{Walukiewicz}{Walukiewicz}{2004}]%
        {Walukiewicz04}
\bibfield{author}{\bibinfo{person}{Igor Walukiewicz}.}
  \bibinfo{year}{2004}\natexlab{}.
\newblock \showarticletitle{A Landscape with Games in the Background}. In
  \bibinfo{booktitle}{\emph{Proceedings of the 19th Annual IEEE Symposium on
  Logic in Computer Science (LiCS~2004)}}. \bibinfo{publisher}{Computer Society
  Press}, \bibinfo{pages}{356--366}.
\newblock


\bibitem[\protect\citeauthoryear{Wilke}{Wilke}{2001}]%
        {Wilke2001}
\bibfield{author}{\bibinfo{person}{Thomas Wilke}.}
  \bibinfo{year}{2001}\natexlab{}.
\newblock \showarticletitle{Alternating Tree Automata, Parity Games and Modal
  $\mu$-Calculus}.
\newblock \bibinfo{journal}{\emph{Bulletin of the Belgian Mathematical
  Society}} \bibinfo{volume}{8}, \bibinfo{number}{2} (\bibinfo{year}{2001}),
  \bibinfo{pages}{359--391}.
\newblock


\bibitem[\protect\citeauthoryear{W{\"o}hrle and Thomas}{W{\"o}hrle and
  Thomas}{2007}]%
        {WohrleT07}
\bibfield{author}{\bibinfo{person}{Stefan W{\"o}hrle} {and}
  \bibinfo{person}{Wolfgang Thomas}.} \bibinfo{year}{2007}\natexlab{}.
\newblock \showarticletitle{Model Checking Synchronized Products of Infinite
  Transition Systems}.
\newblock \bibinfo{journal}{\emph{Logical Methods in Computer Science}}
  \bibinfo{volume}{3}, \bibinfo{number}{4} (\bibinfo{year}{2007}).
\newblock


\end{thebibliography}
